\documentclass{lmcs} %
\pdfoutput=1
\usepackage[utf8]{inputenc}

\usepackage{lastpage}
\lmcsdoi{22}{3}{2}
\lmcsheading{}{\pageref{LastPage}}{}{}%
{Oct.~31,~2024}{Jul.~20,~2026}{}


\usepackage{hyperref}
\usepackage{tikz}
\usepackage{graphicx}
\usetikzlibrary{shapes,arrows,automata}

\usepackage{mathpartir}

\usepackage{amsmath,tabularx}
\usepackage{amsfonts}
\usepackage{amsthm}
\usepackage{amssymb}
\usepackage{mathtools}
\usepackage{thmtools}
\usepackage{multirow}
\usepackage{mathrsfs}  
\usepackage{stmaryrd}
\usepackage{thm-restate}
\usepackage{environ}

\usepackage{listings}
\usepackage[inline]{enumitem}

\usepackage{prftree}
\usepackage{upgreek}
\usepackage{wasysym}
\usepackage{subcaption}
\usepackage{hyperref}
\usepackage{cancel}

\usepackage{xcolor}
\usepackage{booktabs}

\newcommand{\viola}[1]{}

\lstdefinelanguage{MuAC}{
		keywords = {Gives, Me, Requester, with},
	    comment=[l]{//},
	}

\lstdefinestyle{MuAC}{
		language=MuAC,
		basicstyle=\footnotesize\ttfamily, 
		otherkeywords={:-},
		keywordstyle=\color{blue},
		commentstyle=\color{gray},
	}

\def\makenewenum#1#2{%
\newcounter{cnt#1}
\newenvironment{#1}%
{\begin{list}{\makebox[0pt][r]{#2}}%
{\setlength{\itemsep}{0pt}%
 \setlength{\parsep}{.2em}%
 \setlength{\leftmargin}{2em}%
 \setlength{\labelwidth}{.2em}%
 \usecounter{cnt#1}}}%
{\end{list}}}
\makenewenum{enu}{\arabic{cntenu}.}
\makenewenum{enum}{(\arabic{cntenum})}
\makenewenum{itenum}{\em(\arabic{cntenum})}
\makenewenum{en}{(\roman{cnten})}
\makenewenum{renum}{\textit{(\roman{cntrenum})}}
\makenewenum{enhyphen}{(\roman{cntenhyphen}')}
\makenewenum{aenum}{(\alph{cntaenum})}
\makenewenum{itemizes}{$\bullet$}

\newenvironment{restate-theorem}[1]%
  {\begin{trivlist}\item[]{\normalsize\bfseries{\sffamily}
        Restatement of Theorem~#1.}\hspace*{0mm}\it}%
  {\end{trivlist}}

\newenvironment{restate-lemma}[1]%
  {\begin{trivlist}\item[]{\normalsize\bfseries{\sffamily}
        Restatement of Lemma~#1.}\hspace*{0mm}\it}%
  {\end{trivlist}}

\newenvironment{restate-proposition}[1]%
  {\begin{trivlist}\item[]{\normalsize\bfseries{\sffamily}
        Restatement of Proposition~#1.}\hspace*{0mm}\it}%
  {\end{trivlist}}

\newcommand{\tr}{tr}

\newcommand{\semdeno}[1]{\llbracket #1 \rrbracket}
\newcommand{\semden}[1]{\llparenthesis\, #1 \,\rrparenthesis}

\def\makenewenum#1#2{%
	\newcounter{cnt#1}
	\newenvironment{#1}%
	{\begin{list}{\makebox[0pt][r]{#2}}%
			{\setlength{\itemsep}{0pt}%
				\setlength{\parsep}{.2em}%
				\setlength{\leftmargin}{2em}%
				\setlength{\labelwidth}{.2em}%
				\usecounter{cnt#1}}}%
		{\end{list}}}

\NewEnviron{smalign*}{%
			\begin{align*}
			\BODY
			\end{align*}}

\newcommand{\semantics}[1]{\llparenthesis#1\rrparenthesis}

\newcommand{\semanticsusr}[2]{\llparenthesis#1\rrparenthesis_{#2}}

\DeclareMathOperator{\linearcontract}{\multimap\hspace{-8pt}\multimap}
\DeclareMathOperator{\smlinearcontract}{\multimap\hspace{-6pt}\multimap}
\newcommand{\contract}{\twoheadrightarrow}

\newcommand{\MuACL}{\textnormal{MuACL}}

\newcommand{\MuACstate}{\mathit{st}}
\newcommand{\st}{\MuACstate}
\newcommand{\MuACStates}{\mathit{St}}
\newcommand{\usr}{\mathit{usr}}
\newcommand{\Usr}{\mathit{Usr}}
\newcommand{\Res}{\mathit{Res}}
\newcommand{\res}{\mathit{res}}
\newcommand{\Tran}{\mathit{Tr}}
\newcommand{\tran}{\mathit{tr}}
\newcommand{\exc}{\mathit{exc}}
\newcommand{\Exc}{\mathit{Exc}}

\newcommand{\mset}[1]{\{\!| { #1} |\!\}}

\newcommand{\MuACLs}{\MuACL$_{(*\text{-cut})}$}
\newcommand{\MuACLst}{\MuACL$^2_{(*\text{-cut})}$}

\newcommand{\Pol}{\mathit{Pol}}
\newcommand{\pol}{\mathit{pol}}
\newcommand{\CS}{\mathit{CS}}

\newcommand{\Pred}{\mathit{P}}

\newcommand{\GiveLs}{\mathit{GiveLs}}
\newcommand{\PredLs}{\mathit{PredLs}}

\newcommand{\spell}{\mathit{sb}}
\newcommand{\heavy}{\mathit{hw}}
\newcommand{\light}{\mathit{lw}}
\newcommand{\heal}{\mathit{hp}}

\newcommand{\exlo}{\MuACL{\textsuperscript{0}}}

\theoremstyle{plain} %

\begin{document}

\title{Policies for Fair Exchanges of Resources}

\author[L.~Ceragioli]{Lorenzo Ceragioli\lmcsorcid{0000-0002-1288-9623}}[a]
\author[P.~Degano]{Pierpaolo Degano\lmcsorcid{0000-0002-8070-4838}}[a,b]
\author[L.~Galletta]{Letterio Galletta\lmcsorcid{0000-0003-0351-9169}}[a]
\author[L.~Vigan\`o]{Luca Vigan\`o\lmcsorcid{0000-0001-9916-271X}}[c]

\address{IMT School for Advanced Studies Lucca, Italy}	%
\email{lorenzo.ceragioli@imtlucca.it, letterio.galletta@imtlucca.it}  %

\address{Dipartimento di Informatica, Universit\`a di Pisa, Italy}	%
\email{pierpaolo.degano@unipi.it}  %

\address{Department of Informatics, King’s College London, UK}	%
\email{luca.vigano@kcl.ac.uk}  %

\begin{abstract}
	People increasingly use digital platforms to exchange resources in accordance with some policies stating what resources users offer and what they require in return.
	In this paper, we propose a formal model of these environments, focussing on how users' policies are defined and enforced, so ensuring that malicious users cannot take advantage of honest ones.
	To that end, we introduce the declarative policy language MuAC and equip it with a formal semantics.
	To determine if a resource exchange is fair, i.e., if it respects the MuAC policies in force, we introduce the non-standard logic \MuACL\ that combines non-linear, linear and contractual aspects, and prove it decidable.
	Notably, the operator for contractual implication of \MuACL\  is not expressible in linear logic. 
	We define a semantics preserving compilation of MuAC policies into \MuACL, thus establishing that exchange fairness is reduced to finding a proof in \MuACL.
	Finally, we show how this approach can be 
	put to work	on a blockchain to exchange non-fungible tokens.
\end{abstract}

\maketitle

\section{Introduction}\label{sec:intro}
    
Exchanging and sharing of resources and assets have commonly occurred in diverse human contexts for thousands of years, but only the advent of the Internet and modern online platforms have enabled the idea of \emph{sharing economy} to emerge as a new disruptive socio-economic system able to challenge traditional models~\cite{sharing}.  
A typical sharing economy scenario involves a community of users who rely on a digital platform to foster collaboration and to share and transfer to each other resources and assets via peer-to-peer transactions.

As an example, consider a scenario inspired by \emph{Home Exchange}, an online community in which members agree to swap homestays for a period of time~\cite{homeexchange}.
Say Alice offers to exchange her house in Paris with Bob who offers his in Rome.
The decision is made by the two users based both on their preferences, and on the properties and availability of their houses.
More involved transactions may also occur when a direct exchange is not possible.
Consider, e.g., Bob's friend Carl, who owns a flat in London and would like to spend a week in Paris at Alice's house.
However, Alice does not plan to visit London, so there is no direct agreement with Carl, but Bob can generously ``pay for Carl'', giving Alice his house in place of Carl's.
As this example shows, when some user requests the resource(s) of another, the two, and possibly more, start bargaining until an agreement is found.

A digital platform has to support users in at least two key issues:
\begin{enumerate}
\item[(1)] The first issue is ensuring that users' requests and expectations match.
To address this problem, the platform should implement mechanisms through which each user specifies conditions on what she offers and what she requires in return, namely \emph{exchange policies}.
In addition, the platform must provide all the involved users with the guarantee that the proposed exchange is \emph{fair}, i.e., it obeys all their policies. 
\item[(2)] The second issue is guaranteeing that the agreed exchange takes place properly so as to prevent malicious users from taking advantage of others. 
For that,  so-called fair exchange protocols have been studied, proving that a trusted third party (TTP) is always required to ensure fairness and to solve disputes, even in quite restricted cases~\cite{Pagnia99}.\footnote{There are a number of other interesting issues that could be considered. For instance, the platform could facilitate the negotiation by ensuring that users can express what they offer and what they desire, or it could support the users in reaching the most advantageous agreement for them all. In this paper, we do not consider these additional issues and instead focus on the two key issues we mentioned.}
\end{enumerate}

In this paper, we provide a foundational approach to investigate these two issues, particularly the first one, 
and we introduce a formal model of digital platforms.  %
We adopt a minimalist approach by abstracting away from all details about user management that is up to a centralised authority.
We focus on resource ownership and transfers, and, in particular, on the exchange policies that regulate them.
Hence, we do not consider issues like registration to or cancellation from the platform,
handling of user profiles, group definition, interaction mechanisms between users, etc.

We provide four main contributions.
The first contribution is the basic notion of \emph{exchange environment} that formally models the behaviour of exchange platforms.
We define an exchange environment to be a
labeled transition system.
Its states record the ownership of the resources and its transitions represent transfers. 
Moreover, we introduce the notion of \emph{exchange policy} 
to formally represent the requirements of users on resource exchanges.
The exchanges in a transition must guarantee that a fair agreement has indeed been reached among users so that each of them
gives what she promised and gets what she required.
Fairness will prevent a dishonest participant from deceiving others and make them accept exchanges that do not satisfy their policies.

Our second contribution is MuAC, a declarative access control policy language similar to Datalog, 
through which users define their exchange policies in isolation.
Again, these amount to basic conditions on when a resource can be given to another user, in particular conditional promises on which resources the giving user expects in return.
These high-level policies will then be mapped to exchange policies.

Checking that a resource exchange is fair requires controlling that it obeys the policies of all the users involved, which is the key issue (1) 
discussed above.
However, a crucial point is that such agreements may be circular, as it is typical of human and of virtual contracts.
In addition, an exchange may ``consume'' the resources.
To see why, consider the example above.
Alice promises her house to Bob if Bob is willing to do the same (and vice versa):
a guarantee should be offered that the promises match and will be kept.
Moreover, once the agreement is reached and Alice has given her house to Bob, she cannot give it also to David until Bob gives it back to her, otherwise a \emph{double spending} occurs.
We delegate the task of ensuring the fairness of such a distributed agreement to the TTP that is anyway in charge of keeping the current association between users, their policies and their resources. 
Crucially, to avoid misbehaviour the TTP is required to enforce the policies during exchanges.

For that, we set up a logical framework that extends classical logic to 
deal \emph{at the same time} with consumable resources and circularity.
This is our third contribution: we propose \MuACL, a logic that 
features a linear fragment and a non-linear one inspired by LNL~\cite{LNL}.
To handle circularity, \MuACL\ is equipped with a specific operator, called \emph{linear contractual implication}, inspired by PCL~\cite{BZ}.
To the best our knowledge, \MuACL\ is the first logic that combines linear and contractual aspects.
Notably, the expressive power of the standard computational fragment of linear logic and that of \MuACL\ are different, because the operator of contractual implication cannot be encoded in the first logic.
Indeed, there is no homomorphic encoding of \MuACL\ into the standard computational fragment of linear logic (\autoref{thm:homo}).

We then compile MuAC policies into \MuACL\ formulas in a correct and complete way.
The main technical result is 
that the validity of \MuACL\ formulas is decidable (\autoref{thm:MuACLdecide}).
The TTP then finds a witness that the proposed exchange satisfies (or not) the policies of all the involved users, with no double spending.
We discuss the efficacy of our proposal by showing that the TTP only applies fair resource exchanges and prevents different kinds of misbehaviour, which is the key issue (2) discussed above.

To show our policy framework at work, we propose an implementation schema as a blockchain smart contract. 
Through it, users define their policies and exchange resources like non-fungible tokens (NFTs).
Our implementation also plays the role of TTP.
We also propose an off-chain client to reduce the on-chain cost of performing an exchange.

In summary, the main contributions of this paper are both of a theoretical, logical, nature and of a more applied one:
\begin{enumerate}
	\item The notion of exchange environment as a minimalist and abstract formal model of exchange platform. We use this model to precisely characterise when the exchange of resources is fair and when a user misbehaves.
	\item The access control language MuAC through which users of an exchange environment can easily express which resources they are willing to give and what they require in return. 
	MuAC is the first %
	logical access control policy language that allows for expressing promises and exchange contracts on consumable resources.
	\item The non-standard logic \MuACL\ that interprets MuAC policies and  certifies with a proof when a resource exchange is fair. 
	Besides standard constructs, this logic has both linear operators to deal with consumable resources and a contractual implication to express promises that require a circular reasoning to be checked, which is not expressible in linear logic.
	We prove that satisfiability is decidable for \MuACL\ and provide a correct and complete compilation procedure from MuAC policies to this logic.
	To the best of our knowledge, \MuACL\ is the first linear non-linear, contractual and decidable logic --- as such worth studying also in itself.
	\item We instantiate our policy framework on the specific case of non-fungible tokens and we show that its intrinsic policy enforcement prevents different kinds of misbehaviour. 
\end{enumerate}

In the rest of the paper we proceed as follows. In Section~\ref{sec:col-atwork}, we overview our approach.
In Section~\ref{sec:env}, we formalise our exchange environment. In Section~\ref{sec:col-MuAC}, we formalise the MuAC language for exchange policies. 
In Section~\ref{sec:col-formal}, we introduce \MuACL\, and we show how it computes fair exchanges.
In Section~\ref{sec:smart}, we present the implementation of the smart contract and show how MuAC policies only allow fair exchanges of non-fungible tokens. 
In Section~\ref{sec:discuss} and Section~\ref{sec:col-related}, we discuss limitations and the connections to related work.
Finally, in Section~\ref{sec:col-conclude}, we draw conclusions and discuss our plans for future work. 
The appendices A--E contain a summary of our notation, the proofs of our theorems and all the technical details.

\section{An Overview of the Approach}\label{sec:col-atwork}

We first introduce the notion of exchange environment.
Then we introduce a running example that allows us to provide an overview of our proposal.

\subsection{Setting the Context}\label{sec:model}

In an exchange environment, users own their resources and may transfer them to others in order to obtain something in return, thus performing exchanges.
Here, we do not consider how users communicate with the digital platform, e.g., to register themselves and handle their profile, how they interact with each other to bargain an agreement,  etc.
Rather, we focus on the basic notions of exchange
and of its fairness,
on the language users can use to define their own policies, and on mechanisms to verify whether an exchange is fair.
We assume that these mechanisms are trustworthy and that policies 
express all the exchanges that users are willing to accept in a sort of default deny approach.

\subsection{MuAC on a Running Example}\label{sec:col-runex}

\begin{figure}[t]
\begin{subfigure}{0.4\textwidth}
\small
\begin{enumerate}[label=\textbf{A\arabic*}]
	\item I give a $\spell$ if I get a $\heavy$ in return\label{rule:col-a1}
	\item I give a $\spell$ if I get a $\heal$ in return\label{rule:col-a2}
\end{enumerate}
\caption{Alice's policy.}
\end{subfigure}

\vspace{.5cm}

\begin{subfigure}{0.43\textwidth}
\small
\begin{enumerate}[label=\textbf{B\arabic*}]
	\item I give a $\light$ if I get a $\spell$ in return\label{rule:col-b1}
	\item I give you a $\heal$ if you give me a $\spell$\label{rule:col-b2}
	\item If you are a paladin, then I give you a $\light$ if you give me a $\heal$\label{rule:col-b3}
\end{enumerate}
\caption{Bob's policy.}
\end{subfigure}
\qquad
\begin{subfigure}{0.43\textwidth}
\small
\begin{enumerate}[label=\textbf{C\arabic*}]
	\item I give a $\heavy$ if I get a $\light$ in return\label{rule:col-c1}
	\item I give you a $\heal$ if you give me a $\light$\label{rule:col-c2}
	\item If you give a $\spell$ to a paladin, then I give you a $\heal$\label{rule:col-c3}
\end{enumerate}
\caption{Carl's policy.}
\end{subfigure}
\caption{Policies of Alice, Bob, and Charlie expressed in natural language.}
\label{fig:policies}
\end{figure}

Blockchains like Ethereum host several decentralised competitive card games, e.g., Gods Unchained~\cite{godsunchained}, Splinterlands~\cite{splinterlans}, Skyweaver~\cite{skyweaver}. %
In these games, cards are NFTs, associated with the owner's blockchain account.
This enables users to trade and exchange their cards freely, with the same level of ownership as if they were real, tangible cards.

We consider a fictional card game, played by Alice, Bob and Carl.
As it is common in online games, players can join guilds of players for helping each other getting stronger, so let Bob and Carl belong to the guild called \emph{paladins}.
We assume that four cards are available in the game (in multiple copies): healing potions (\emph{hp} in the following), spell books (\emph{sb}), light and heavy weapons (\emph{lw} and \emph{hw}).
Moreover, we assume that the game developers manage the creation and distributions of cards NFTs, and record the membership of users in guilds.

Finally, we assume that users define in their policies which exchanges they are willing to accept.
Let the policies of Alice (rules \ref{rule:col-a1} and \ref{rule:col-a2}), Bob (\ref{rule:col-b1}, \ref{rule:col-b2}) and Carl (\ref{rule:col-c1}, \ref{rule:col-c2}, \ref{rule:col-c3}) be  
the ones in \autoref{fig:policies}.
The rules explicitly say who is giving what to whom and what is required in return. For instance, in rule \ref{rule:col-b2} Bob is happy to give another player a $\heal$ if that player gives him a $\spell$ in return. Instead, in rule \ref{rule:col-a1}, Alice is ready to give a $\spell$ to some other player if she gets a $\heavy$ in return, from whomever.

We show some examples of agreements (see \autoref{fig:ex:agreements}) with increasing complexity that may lead to exchanges.
We assume that Alice has two $\spell$ cards, Bob has one $\light$ card, and Carl has three $\heavy$ and two $\heal$ cards. 
\begin{exa}[Direct exchange]\label{ex:col-dir}
	The simplest case is when the resources of two players are exchanged.
	Bob wants a $\heal$ card and asks Carl, who is willing to exchange $\heal$ with a $\light$ (rule \ref{rule:col-c2}).
	Bob has a $\light$, and is willing to exchange it with a $\heal$, but only with a paladin (rule \ref{rule:col-b3}).
	As a result of the exchange, Bob will get the $\heal$ he needs but no $\light$, and Carl will have also one $\light$ and a single $\heal$.
\end{exa}

\begin{exa}[I pay for you]\label{ex:col-ipfu}
	If Bob wants a $\spell$ card, he can contact Alice, who offers a $\spell$ in return for $\heal$ (rule \ref{rule:col-a2}), regardless of who gives her the $\heal$ card.
	Bob has no $\heal$ to exchange for $\spell$,
	but	luckily he is 
	a paladin,
	so he 
	asks
	Carl 
	who is willing to pay for other members of the paladin gild (rule \ref{rule:col-c3}).
	Bob takes the $\spell$ of Alice, and Alice the $\heal$ of Carl.	
\end{exa}

\begin{exa}[Circular Exchange]\label{ex:col-cir}
	Assume Alice wants a $\heavy$ card.
	She offers $\spell$ in return (rule \ref{rule:col-a1}), but no one is willing to make such an exchange.
	The only one that offers $\heavy$ is Carl, who wants $\light$ in return (rule \ref{rule:col-c1}).
	Alice has no $\light$ resource.
	No agreement is possible between any two users, but if Bob comes into play then an exchange is possible.
	Bob proposes to give $\light$ for $\spell$ (rule \ref{rule:col-b1}).
	A satisfactory agreement is proposed: Alice gives her $\spell$ to Bob (satisfying the condition of rule \ref{rule:col-b1}), 
	Bob gives his $\light$ to Carl (satisfying the condition of rule \ref{rule:col-c1}), and
	Carl gives his $\heavy$ to Alice (satisfying the condition of rule \ref{rule:col-a1}).
	In practice, every user $\usr$ is paying for some other user $\usr'$, provided that some $\usr''$ is paying for $\usr$.
	It is trivial to verify that everyone is happy: they are receiving what they wanted by paying what they promised.
\end{exa}

\begin{figure*}[t]
\centering
\begin{subfigure}{0.25\textwidth}
	\centering
	\footnotesize
	\begin{tikzpicture}
		\node[inner sep=0pt] (bob)  at (1.75,.66) {\begin{tabular}{c}\includegraphics[scale=0.04]{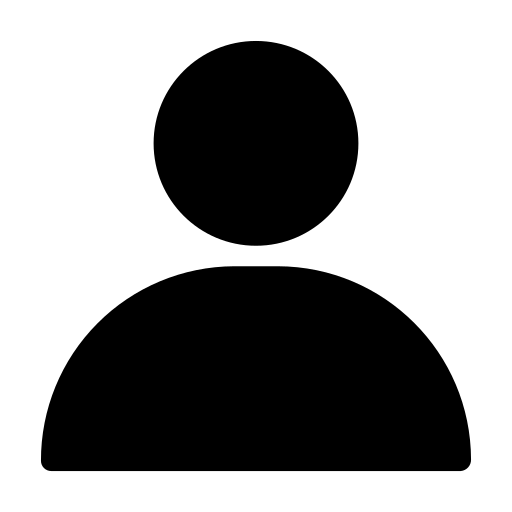}
				\\[-0.02cm] Bob
		\end{tabular}};
		
		\node[inner sep=0pt] (carl)  at (3.5,0) {\begin{tabular}{c}\includegraphics[scale=0.04]{images/user}
				\\[-0.02cm] Carl
		\end{tabular}};

		{
			\draw[o->, thick, bend left] (bob) edge node[above]{$\light$} (carl);
			\draw[<-o, thick, bend right] (bob) edge node[below]{$\heal$} (carl);
		}
	\end{tikzpicture}
	\label{fig:ex:1}
	\caption{Direct exchange.}
\end{subfigure}
\begin{subfigure}{0.4\textwidth}
	\centering
	\footnotesize
	\begin{tikzpicture}
		\node[inner sep=0pt] (alice)  at (0,0) {\begin{tabular}{c}\includegraphics[scale=0.04]{images/user}
				\\[-0.02cm] Alice
		\end{tabular}};
		
		\node[inner sep=0pt] (bob)  at (1.5,.66) {\begin{tabular}{c}\includegraphics[scale=0.04]{images/user}
				\\[-0.02cm] Bob
		\end{tabular}};
		
		\node[inner sep=0pt] (carl)  at (3,0) {\begin{tabular}{c}\includegraphics[scale=0.04]{images/user}
				\\[-0.02cm] Carl
		\end{tabular}};	
		
		{
			\draw[o->, thick, bend left] (alice) edge node[above]{$\spell$} (bob);
			\draw[<-o, thick, out=-20, in=-160] (alice) edge node[below]{$\heal$} (carl);
		}
	\end{tikzpicture}
	\label{fig:ex:2}
	\caption{I pay for you.}
\end{subfigure}
\begin{subfigure}{0.25\textwidth}
	\centering
	\footnotesize
	\begin{tikzpicture}
		\node[inner sep=0pt] (alice)  at (0,0) {\begin{tabular}{c}\includegraphics[scale=0.04]{images/user}
				\\[-0.02cm] Alice
		\end{tabular}};
		
		\node[inner sep=0pt] (bob)  at (1.5,.66) {\begin{tabular}{c}\includegraphics[scale=0.04]{images/user}
				\\[-0.02cm] Bob
		\end{tabular}};
		
		\node[inner sep=0pt] (carl)  at (3,0) {\begin{tabular}{c}\includegraphics[scale=0.04]{images/user}
				\\[-0.02cm] Carl
		\end{tabular}};	
		
		{
			\draw[o->, thick, bend left] (alice) edge node[above]{$\spell$} (bob);
			\draw[o->, thick, bend left] (bob) edge node[above]{$\light$} (carl);
			\draw[<-o, thick, out=-20, in=-160] (alice) edge node[below]{$\heal$} (carl);
		}
	\end{tikzpicture}
	\label{fig:ex:3}
	\caption{Circular exchange.}
\end{subfigure}
\\[2ex]
\begin{subfigure}{0.8\textwidth}
	\centering
	\footnotesize
	\begin{tikzpicture}
		\node[inner sep=0pt] (alice)  at (0,0) {\begin{tabular}{c}\includegraphics[scale=0.04]{images/user}
				\\[-0.02cm] Alice
		\end{tabular}};
		
		\node[inner sep=0pt] (bob)  at (2,.66) {\begin{tabular}{c}\includegraphics[scale=0.04]{images/user}
				\\[-0.02cm] Bob
		\end{tabular}};
		
		\node[inner sep=0pt] (carl)  at (4,0) {\begin{tabular}{c}\includegraphics[scale=0.04]{images/user}
				\\[-0.02cm] Carl
		\end{tabular}};	
		
		{
			\draw[<-o, thick, bend left] (alice) edge node[above]{(3) $\heal$} (bob);
			\draw[o->, thick, bend right,out=-15] (alice) edge node[below]{(4) $\spell$} (bob);
			\draw[o->, thick, bend left] (bob) edge node[above]{(1) $\light$} (carl);
			\draw[<-o, thick, out=-20, bend right] (bob) edge node[below]{(2) $\heal$} (carl);
		}
	\end{tikzpicture}
	\label{fig:ex:4}
	\caption{Resource supplier: the numbers denote the flow of the exchanges.}
\end{subfigure}
\caption{Examples of agreements among players. }
\label{fig:ex:agreements}
\end{figure*}

\begin{exa}[Resource Supplier]\label{ex:col-res}
	The last case we consider is an agreement between two parts that would be reachable, but one of the two does not have the needed resource.
	Assume Alice wants $\heal$.
	A simple agreement would be between her and Bob: Alice proposes to give $\spell$ for $\heal$ (rule \ref{rule:col-a2}), and Bob gives $\heal$ for $\spell$ (rule \ref{rule:col-b2}).
	Unfortunately, Bob has no $\heal$, but in spite of that there is an agreement where Alice gets a $\heal$ card and Bob a $\spell$ card.
	Indeed, Bob has $\light$ that can be exchanged with Carl for $\heal$ (rule \ref{rule:col-c2}), and Bob agrees since Carl is in his guild (rule \ref{rule:col-b3}).
	Thus, Carl takes $\light$ from Bob and gives him $\heal$, which Bob can now exchange with Alice for $\spell$.
\end{exa}
\noindent
So far all the agreements are fair, i.e., satisfactory for all users.
Note that checking the fairness of the agreements does not require us to consult users, as long as we know their policies that precisely reflect their wishes and expectations.

We now show an example of an unfair agreement, where one partner does not obey the policy of the other.
\begin{exa}[Policy violation]\label{ex:viol}
The exchange where Bob gives a $\light$ to Alice while she gives back a $\heavy$ violates all Bob's policy requests.
This exchange is clearly not fair and our logical machinery will rule it out.
\end{exa}

We finally display an example of an unfair agreement, caused by a double spending.
For that, we extend Alice's policy with the following rule where she offers to pay with a $\light$ for a $\heal$ in place of a paladin.
\begin{enumerate}[label=\textbf{A\arabic*}]
	\item[\textbf{A3}] 
	I give you a $\light$ if a paladin gets a $\heal$ in return.\label{rule:col-a3}
\end{enumerate}
\begin{exa}[Double Spending]\label{ex:double}
The exchange where Carl gives a $\heal$ to Bob and Alice pays it with her $\light$ is fair (rules~\textbf{A3} and~\ref{rule:col-c2}), and so is the one in which instead Bob pays (rules~\ref{rule:col-b3} and~\ref{rule:col-c2}).
However, a double spending arises when Carl gives a single $\heal$ to Bob, and both Alice and Bob pay.
Thus, this last exchange is not fair. %
\end{exa}

After the needed definitions of exchange environment, of the language for expressing policies and the consequent notion of fairness, and of the logic \MuACL, we will show that provability there ensures an exchange to be fair.
To also give an operational flavour to our proposal, we then instantiate in Section~\ref{sec:smart} our formal model to a blockchain scenario where the digital platform is implemented by a smart contract and the users interact with it in a standard way by sending transactions.
We assume that the smart contract records and publicly displays users' policies and resource ownership that it also manages.
Importantly, we delegate the smart contract to verify the fairness of exchanges.

\section{Exchange Environment}\label{sec:env}

We formalise online platforms that host users who exchange their resources, assuming completed the registration phases and the like.
Our basic model is a transition system called \emph{exchange environment},
where
the transitions 
represent the exchange of resources between users.
We neither impose a topology nor limit the number of participants and of resources, which are however conserved by exchanges.
As discussed above, exchange environments also host participants who behave dishonestly and scam others to steal their resources.
To contrast them, users resort to \emph{exchange policies}.
These policies grant a resource, or more, in return for other resources.
We %
call \emph{fair transitions} those resource exchanges where all the policies of the involved users are obeyed \emph{and} no double spending occurs, 
and we show that no attacks are possible in exchange environments with fair transitions only.

In this section, we introduce policies directly on transitions, in a 
basic form.
The next two sections will then provide users with a logical language to define their policies and with a mechanism for proving a transition fair, which are the main contributions of this paper.

\subsection{Exchange Environments}
\noindent
Below, we assume the following finite sets:
\begin{itemize}
\item a set $\Res$ of \emph{resources}, ranged over by $\res, \res', \res''$;
\item a set $\Usr$ of \emph{users}, ranged over by $\usr, \usr', \usr''$.
\end{itemize}
Hereafter, we omit specifying $\Res$ and $\Usr$ unless required, and we write $x \in X$ to define simultaneously the element $x$ and $X$, the set of all $x$.
Throughout the paper, we make use of multisets, i.e., sets where different occurrences of the same object may occur.
As usual, we represent a multiset $M$ as a function from each element $a$ of the set to the number of its occurrences
$M(a)$.
For simplicity, we carry the set notation over multisets
and we omit the curly brakets when they are not necessary.

Next, we introduce the notions of transfer and exchange, upon which we build the exchange environments as transition systems.

\begin{defi}[Exchange and Exchange Environment] 
	An \emph{exchange} is a multiset $\exc \in \Exc$ of \emph{transfers} $\tran \in \Tran$, i.e., of triples 
	$\usr \xmapsto{\res} \usr'$, 
	with $\usr' \neq \usr$.

	An \emph{exchange environment} %
	is a pair $(\MuACStates, \rightarrow)$, where
	\begin{itemize}
		\item $\MuACStates$ is the set of \emph{states} $\MuACstate \colon \Usr \rightarrow (\Res \rightarrow \mathbb{N})$;
		\item $\rightarrow\ \subseteq \MuACStates \times \Exc \times \MuACStates$ is the \emph{transition relation} that contains the triple $\MuACstate \xrightarrow{\exc} \MuACstate'$ if and only if for all $\usr \in \Usr$ and $\res \in \Res$ the following two conditions hold
	\end{itemize}
	\begin{align*}
	&\quad \text{(1)} \ \sum_{\usr'} \exc (\usr \xmapsto{\res} \usr') \leq \MuACstate(\usr)(\res) \quad \text{and}\\
	&\quad \text{(2)} \ \MuACstate'(\usr)(\res) = \MuACstate(\usr)(\res)\ -
	\sum_{\usr'} \exc (\usr \xmapsto{\res} \usr') + \sum_{\usr''} \exc (\usr'' \xmapsto{\res} \usr)
	\end{align*}
	
We say that  $\MuACstate_0 \rightarrow^* \MuACstate_1$ is a \emph{computation} from the state $\MuACstate_0$ to the state $\MuACstate_1$, where 
$\rightarrow^* \subseteq \MuACStates \times \MuACStates$ is the smallest reflexive transitive relation such that $\MuACstate \rightarrow^* \MuACstate'$ if $\exc$ exists such that $\MuACstate \xrightarrow{\exc} \MuACstate'$.

\end{defi}

An \emph{exchange environment} is a transition system where  a \emph{state} represents resource ownership as a total function $\MuACstate$ associating each user $\usr$ with the multiset of resources $\res$ she owns, and a \emph{transition} represents the occurrence of an exchange that modifies the current state. 
Intuitively, a \emph{transfer} occurs when a user $\usr$ sends her resource $\res$ to another user $\usr'$.
An \emph{exchange} is a finite multiset of transfers, so allowing one user to transfer more resources to another user in a single transition.
Through condition $(1)$ we ensure that an exchange $exc$ is possible only when a user $usr$ owns enough resources, whereas condition $(2)$ makes sure that the state is correctly updated and that no resource is created or destroyed.

\begin{exa}\label{ex:ex-env}
	The following are some states of the exchange environment of our running example of~Section~\ref{sec:col-runex}:
\begin{align*}
	\MuACstate &= \{ (\mathit{Alice}, \{\spell\}), (\mathit{Bob}, \{\light\}), (\mathit{Carl}, \{\heavy, \heavy, \heavy, \heal, \heal\}) \}\\
	\MuACstate' &= \{ (\mathit{Alice}, \{\heavy\}), (\mathit{Bob}, \{\spell\}), (\mathit{Carl}, \{\light, \heavy, \heavy, \heal, \heal\}) \}
\end{align*}
where the resources are $\heal$, $\light$ and $\heavy$, and the users are $\mathit{Alice}$, $\mathit{Bob}$ and $\mathit{Carl}$. 	

The exchange environment contains, e.g., the transition $\MuACstate \xrightarrow{\exc} \MuACstate'$ of~\autoref{ex:col-cir}, with:
\begin{gather*} 
	\exc = \{ \mathit{Alice} \xmapsto{\spell} \mathit{Bob}, \mathit{Bob} \xmapsto{\light} \mathit{Carl}, \mathit{Carl} \xmapsto{\heavy} \mathit{Alice} \}.
\end{gather*} 
\end{exa}

\subsection{Exchange Policies}
So far, users' intents play no role, and thus there is no guarantee that a transition of the exchange environment complies with them.
We introduce below a basic way to define which exchanges users agree on, hence which transitions are beneficial to all the involved users.
Every user in isolation defines her \emph{exchange policy} that specifies when one of its resources can be exchanged for some resources belonging to other users.

Roughly, an exchange policy is a set of \emph{exchange approvals}, written 
$\usr \xmapsto{\res} \usr' \triangleleft \exc$.
It reads as follows: the user $\usr$ is willing to give her resource $\res$ to the user $\usr'$ in return of the exchange $\exc$.
The exchange policy determines whether $\exc$ requires the payoff to be given directly to $\usr$ or to another user $\usr''$ chosen by $\usr$.
Formally:
\begin{defi}[Exchange Approval and Policies]\label{def:policy} 
	An \emph{exchange approval} of a user $\usr$ is a pair $\usr \xmapsto{\res} \usr' \triangleleft \exc \in \Tran \times \Exc$ such that 
	for each $\usr'' \xmapsto{\res} \usr''' \in \exc$ it is $\usr'' \neq \usr$.
	
	The \emph{exchange policy} $pol_\usr$ of $\usr$ is a set of exchange approvals.
\end{defi}
	\begin{exa}\label{ex:exc-pol}
		In the following policies, both Alice and Bob are willing to pay Carl with a $\light$ if he gives a $\heal$ to Bob, and Carl accepts to be paid by any of them.
		\begin{align*}
		\pol_{Alice} =& \{ Alice \xmapsto{\light} Carl \triangleleft \{ Carl \xmapsto{\heal} Bob \} \}\\
		\pol_{Bob} =& \{ Bob \xmapsto{\light} Carl \triangleleft \{ Carl \xmapsto{\heal} Bob \} \}\\
		\pol_{Carl} =& \{ Carl \xmapsto{\heal} Bob \triangleleft \{ Alice \xmapsto{\light} Carl\},%
		 Carl \xmapsto{\heal} Bob \triangleleft \{ Bob \xmapsto{\light} Carl \} \}
		\end{align*}
	\end{exa}

We now move towards the definition of fair exchange, which is better done in two steps. 
We begin by defining when an exchange respects the policy of a single user, which can be done locally (Definition~\ref{def:accepted-exc}).
However, a transfer may be unfair even if accepted by the policies of all the users involved as the \emph{same} resource can be \emph{offered more than once} to different users by an attacker, in other words when a \emph{double spending} occurs. Definition~\ref{def:fairexc} of Fair Exchange rules out such a case, but it requires a \emph{global} check.

Intuitively, the policy $\pol_\usr$ of a user $usr$ locally accepts an exchange $\exc$ when for all the transfers in $\exc$ involving $\usr$ as a giver there is a subset $\exc' \subseteq \exc$ of transfers that grants the payoff required by the approvals of $\pol_\usr$.
Formally:
\begin{defi}[Accepted Exchange]\label{def:accepted-exc} 
		
We write $\pol_\usr \vDash_{\exc'} \exc$ if and only if the triple $(\pol_\usr, \exc', \exc)$ belongs to the minimal subset of $\Pol \times \Exc \times \Exc$ satisfying the inference rules below:		
\begin{mathpar}
	\inferrule
	{\exc \subseteq \{ \usr' \xmapsto{\res} \usr'' \mid \usr' \neq \usr \} }
	{\pol_\usr \vDash_{\emptyset} \exc} \label{accept:others}
	
	\inferrule
	{ \usr \xmapsto{\res} \usr'  \triangleleft \exc \in \pol_\usr  \\  \pol_\usr \vDash_{\exc''} \exc'}
	{pol_\usr \vDash_{\exc \uplus \exc''} \{\usr \xmapsto{\res} \usr'\} \uplus \exc \uplus \exc'} \label{accept:me}
\end{mathpar}
When $\pol_\usr \vDash_{\exc'} \exc$  holds we say that $\exc$ is \emph{accepted} by $\pol_\usr$\emph{ because of} $\exc'$.
\end{defi}

Condition~\ref{accept:others} says that a transfer is always accepted by $\usr$ when she gives no resource. 
Condition \ref{accept:me} requires that for each transfer where $\usr$ is giving something, she should get back what specified by her policy.
Note that $\exc'$ works like a witness for the acceptance and that $\exc' \subseteq \exc$ whenever $\pol_\usr \vDash_{\exc'} \exc$ holds.
Below, we sometimes omit $\exc'$ and just say that $\exc$ is accepted by $\pol_\usr$.

As an example of double spending consider the following.

\begin{exa}\label{ex:fair-nonlocal}
	Consider~\autoref{ex:exc-pol}, and the following exchanges where 
	Carl gives two $\heal$ to Bob and both Alice and Bob pay for one of them with a $\light$:
	\begin{align*}
		\exc = \{ Carl \xmapsto{\heal} Bob, Carl \xmapsto{\heal} Bob, Alice \xmapsto{\light} Carl, Bob \xmapsto{\light} Carl \}.
	\end{align*} 
	This exchange is accepted by the three players.
	However, also the following is accepted by all of them in isolation, where the double spending of~\autoref{ex:double} occurs:
	\[ 
	\exc' = \{ Carl \xmapsto{\heal} Bob, Alice \xmapsto{\light} Carl, Bob \xmapsto{\light} Carl \}. 
	\]
\end{exa}

We finally define 
\emph{fair transitions}
(and show the use of the extra $\exc'$ in Definition~\ref{def:accepted-exc}). 
\begin{defi}[Fair Transition]\label{def:fairexc} 
	The transition $\MuACstate \xrightarrow{\exc} \MuACstate'$ is \emph{fair} if and only if for all $\usr \in \Usr$ there exists
	an exchange $\exc_{\usr}$ such that $\pol_{\usr} \vDash_{\exc_\usr} \exc$ and $\biguplus_{\usr \in \Usr} \exc_{\usr} \subseteq \exc$.
	We will occasionally call \emph{fair} the label $\exc$ of a fair transition.
A computation $\MuACstate_0 \rightarrow^* \MuACstate_1$ is \emph{fair} when its \mbox{steps are fair}.
\end{defi}

Roughly, a transition, or its label $\exc$, is fair when it is accepted by the policies of all the users involved
and, in addition, the inclusion of the disjoint union of $\exc_{\usr}$ in $\exc$ guarantees that each transfer $\tr$ in $\exc$ can be used at most once as a justification.
Clearly this prevents double spending.

\begin{exa}\label{ex:env}
	The exchange $\exc =  \{ Carl \xmapsto{\heal} Bob, Carl \xmapsto{\heal} Bob, Alice \xmapsto{\light} Carl, Bob \xmapsto{\light} Carl \}$ of \autoref{ex:fair-nonlocal}, where Bob takes two $\heal$ from Carl and both Bob and Alice pays each for one resource, is fair because it is accepted by the three users with the following witness:
	\begin{gather*}
		\exc_{Alice} = \{ Carl \xmapsto{\heal} Bob \} \qquad
		\exc_{Bob} = \{ Carl \xmapsto{\heal} Bob \} \\
		\exc_{Carl} = \{ Alice \xmapsto{\light} Carl, Bob \xmapsto{\light} Carl \}
	\end{gather*}
	\noindent	
	Instead, the exchange $\exc'$ of \autoref{ex:fair-nonlocal} is unfair because $\{Carl \xmapsto{\heal} Bob\}$ appears twice in the disjoint union of the witnesses of Alice and Bob.
\end{exa}

\section{MuAC: A logical language for Exchange Policies}\label{sec:col-MuAC}

To simplify the definition of the users' intents, we introduce the language MuAC that allows one to express exchange policies in a simple and declarative manner. 
MuAC is a logical language similar to Datalog and is parametric with respect to a set of predicates, the definition of which we leave implicit.
Intuitively, these predicates group users in categories, like fellowship or affinity, which are convenient to define policies, whereas the context stores this information on users.
In the following, we assume as given:
\begin{itemize}
\item
a \emph{set of user variables} $U$, ranged over by $u, u', u'', u_i$, and the distinguished variable \texttt{Me}$\,\notin U$ to represent the owner of the policy;
\item
a \emph{set of predicate symbols}  $\Pred$, ranged over by $p, p', p''$;
\item
a \emph{context} $C$, i.e.,
an interpretation of the predicates such that $C(p) \subseteq \Usr^n$, where $n$ is the arity
of $p$.
\end{itemize}
An exchange policy is represented as a MuAC ruleset, the syntax of which is defined below. 
Roughly, a rule in a ruleset is a Horn clause stating that the policy owner \texttt{Me} is willing to give a resource $\res$ to a requester if a (possibly empty) list of conditions are satisfied.
These conditions consist of two parts: 
the resources that the policy owner requires in return, and some properties of the users involved in the exchange.
\begin{defi}[MuAC ruleset]\label{def:MuPol}
The MuAC ruleset $R_\usr$ of $\usr$ is a set of rules $r$ given by the following grammar, under the assumption that $u \neq \texttt{Me}$,
$u \neq u'$, $\epsilon$ is the empty list, and the part within square brakets is optional:
\begin{align*}
r &::= \texttt{Gives}(\texttt{Me}, \res, u) \texttt{ :- } \GiveLs \ 
[
\texttt{with}\  \PredLs 
] \\ 
\PredLs &::= p(u_1, \dots , u_n)\ \PredLs \mid \epsilon \\
\GiveLs &::= \texttt{Gives}(u, \res, u')\ \GiveLs \mid \epsilon
\end{align*}
\end{defi}

\begin{exa}\label{ex:col-MuACpol}
	Continuing the running example of Section~\ref{sec:col-runex}, we express in MuAC the rulesets of Alice, Bob and Carl from~\autoref{fig:policies}.
	The ruleset 
	of Alice is:
	\begin{lstlisting}[style=MuAC]
Gives(Me, spell_book, u) :- Gives(u', heavy_weapon, Me)       // Rule A1
Gives(Me, spell_book, u) :- Gives(u', healing_potion, Me)     // Rule A2
	\end{lstlisting}
(Where the text after \texttt{//} is a comment). The one of Bob is:
	\begin{lstlisting}[style=MuAC]
Gives(Me, light_weapon, u) :- Gives(u', spell_book, Me)       // Rule B1
Gives(Me, healing_potion, u) :- Gives(u, spell_book, Me)      // Rule B2
Gives(Me, light_weapon, u) :- 
	Gives(u, healing_potion, Me)  with is_paladin(u)      // Rule B3
	\end{lstlisting}
	Finally, the one of Carl follows:
	\begin{lstlisting}[style=MuAC]
Gives(Me, heavy_weapon, u) :- Gives(u', light_weapon, Me)     // Rule C1
Gives(Me, healing_potion, u) :- Gives(u, light_weapon, Me)    // Rule C2
Gives(Me, healing_potion, u) :- 
	Gives(u, spell_book, u')  with is_paladin(u')         // Rule C3
	\end{lstlisting}
\end{exa}

Intuitively, the evaluation of a rule requires first to bind the distinguished element \texttt{Me}, the user variable $u$ and the properties $p$ to actual users and properties.
We interpret the MuAC policy of a user $\usr$ in terms of exchange policies given the context $C$.

\begin{defi}[MuAC ruleset interpretation]
Let $R_{\usr}$ be the set of rules of $\usr$, and let $\rho$ range over interpretations $U \rightarrow \Usr$ such that 
	$\rho(\texttt{Me}) = \usr$
	and in all $\texttt{Gives}(u, \res, u')$ it is $\rho(u) \neq \rho(u')$.
	Then
\[
\pol_\usr = \bigcup_{r \in R_\usr} \semdeno{r}\ \rho\ C
\]
where the semantics $\semdeno{r}\ \rho\ C$ of a single rule $r$ of the MuAC policy of $\usr$ is defined as
	\begin{align*}
		\semdeno{\texttt{Gives}(\texttt{Me}, \res, u') &\texttt{ :- } \GiveLs \text{ with } \PredLs}\ \rho\ C\ = \\
		& \{ \semdeno{\texttt{Gives}(\texttt{Me}, \res, u')}\ \rho\ \triangleleft\ \semdeno{\GiveLs}\ \rho\ \mid\ \semdeno{\PredLs}\ \rho\ C \}  
	\end{align*}
	with $\semdeno{\PredLs}\ \rho\ C$ defined as	
	\begin{align*}
		&\semdeno{\epsilon}\ \rho\ C = true\\
		&\semdeno{p(u_1, \dots, u_n) \PredLs}\ \rho\ C =
		(\rho(u_1), \dots, \rho(u_n)) \in C(p) \land \semdeno{\PredLs}\ \rho\ C 
	\end{align*}
	and where $\semdeno{\GiveLs}\ \rho$ is the homomorphic extension of  
	\[
	\semdeno{\texttt{Gives}(u, \res, u')}\ \rho = \mset{\rho(u) \xmapsto{\res} \rho(u')}
	\text{\ \ with \ }  \semdeno{\epsilon}\ \rho = \emptyset
	\]
\end{defi}

\begin{exa}
	Consider again~\autoref{ex:col-MuACpol}. Alice's ruleset is interpreted as the following set of exchange approvals:
	\begin{align*}
		\pol_{Alice} =  (\semdeno{\texttt{Rule A1}} C) \cup (\semdeno{\texttt{Rule A2}} C)
	\end{align*}
	where
	\begin{align*}
	    \semdeno{\texttt{Rule A1}}\ C =\ \{\  & Alice \xmapsto{\spell} Bob \triangleleft \{ Bob \xmapsto{\heavy} Alice \},\ 
		 Alice \xmapsto{\spell} Bob \triangleleft \{ Carl \xmapsto{\heavy} Alice \},\\[-0.1cm]
		& Alice \xmapsto{\spell} Carl \triangleleft \{ Carl \xmapsto{\heavy} Alice \},\ 
		Alice \xmapsto{\spell} Carl \triangleleft \{ Bob \xmapsto{\heavy} Alice \} \}\\
		 \semdeno{\texttt{Rule A2}}\ C =\ \{\ & Alice \xmapsto{\spell} Bob \triangleleft \{ Bob \xmapsto{\heal} Alice \},\ Alice \xmapsto{\spell} Bob \triangleleft \{ Carl \xmapsto{\heal} Alice \},\\[-0.1cm]
	       & Alice \xmapsto{\spell} Carl \triangleleft \{ Carl \xmapsto{\heal} Alice \},\ 
		 Alice \xmapsto{\spell} Carl \triangleleft \{ Bob \xmapsto{\heal} Alice \}\ \} 
	\end{align*}
\end{exa}

\section{A Logic for Characterizing Fair Exchanges}\label{sec:col-formal}

\begin{figure}[t]
	\begin{center}
		\small
		\begin{tabular}{c}
			\textbf{Non-linear Rules}\\
			
			\\[-0,2cm]
			
			\prftree[r]
			{($\top$-right)}
			{}
			{\Vdash \top}
			
			\qquad
			
			\prftree[r]
			{($\Omega$-Ax)}
			{}
			{\omega \Vdash \omega}

			\qquad
			
			\prftree[r]
			{(Cont)}
			{\Omega, \omega, \omega \Vdash \omega'}
			{\Omega, \omega \Vdash \omega'}
			
			\qquad
			
			\prftree[r]
			{(Weak)}
			{\Omega \Vdash \omega'}
			{\Omega, \omega \Vdash \omega'}
			
			\\[.33cm]
			
			\prftree[r]
			{($\land$-left1)}
			{\Omega, \omega \Vdash \omega''}
			{\Omega, \omega \land \omega' \Vdash \omega''}
			
			\qquad
			
			\prftree[r]
			{($\land$-left2)}
			{\Omega, \omega' \Vdash \omega''}
			{\Omega, \omega \land \omega' \Vdash \omega''}
			
			\qquad
			
			\prftree[r]
			{($\land$-right)}
			{\Omega \Vdash \omega}
			{\Omega' \Vdash \omega'}
			{\Omega, \Omega' \Vdash \omega \land \omega'}
			
			\\[.33cm]
			
			\prftree[r]
			{($\rightarrow$-left)}
			{\Omega \Vdash \omega}
			{\Omega', \omega' \Vdash \omega''}
			{\Omega, \omega \rightarrow \omega', \Omega' \Vdash \omega''}
			
			\qquad
			
			\prftree[r]
			{($\rightarrow$-right)}
			{\Omega, \omega \Vdash \omega'}
			{\Omega \Vdash \omega \rightarrow \omega'}

			\\[0.3cm]
			\textbf{Non-linear L-Rules}\\
			
			\\[-0,2cm]
						
			\prftree[r]
			{(L-Cont)}
			{\Omega, \omega, \omega; \Theta, \Delta, \Sigma \vdash \sigma}
			{\Omega, \omega; \Theta, \Delta, \Sigma \vdash \sigma}
			
			\qquad
			
			\prftree[r]
			{(L-Weak)}
			{\Omega;  \Theta, \Delta, \Sigma \vdash \sigma}
			{\Omega, \omega ;  \Theta, \Delta, \Sigma \vdash \sigma}
			
			\\[.33cm]
			
			\prftree[r]
			{(L-$\land$-left1)}
			{\Omega, \omega ;  \Theta, \Delta, \Sigma \vdash \sigma}
			{\Omega, \omega \land \omega' ;  \Theta, \Delta, \Sigma \vdash \sigma}
			
			\qquad
			
			\prftree[r]
			{(L-$\land$-left2)}
			{\Omega, \omega' ;  \Theta, \Delta, \Sigma \vdash \sigma}
			{\Omega, \omega \land \omega' ;  \Theta, \Delta, \Sigma \vdash \sigma}
			
			\\[.33cm]
			
			\prftree[r]
			{(L-$\rightarrow$-left)}
			{\Omega \Vdash \omega}
			{\Omega', \omega' ;  \Theta, \Delta, \Sigma \vdash \sigma}
			{\Omega, \omega \rightarrow \omega', \Omega' ;  \Theta, \Delta, \Sigma \vdash \sigma}
			
			\\[0.3cm]
			\textbf{Linear Rules}\\
			
			\\
			
			\prftree[r]
			{($I$-right)}
			{}
			{\vdash I}
			
			\qquad
			
			\prftree[r]
			{($\Sigma$-Ax)}
			{\Omega; \res@\usr \vdash \res@\usr}
			
			\\[.33cm]
			
			\prftree[r]
			{($\otimes$-left-$\Theta$)}
			{\Omega; \Theta, \theta, \theta', \Delta, \Sigma \vdash \sigma}
			{\Omega; \Theta, \theta \otimes \theta', \Delta, \Sigma \vdash \sigma}
			
			\qquad
			
			\prftree[r]
			{($\otimes$-left-$\Delta$)}
			{\Omega; \Theta, \Delta, \delta, \delta', \Sigma \vdash \sigma}
			{\Omega; \Theta, \Delta, \delta \otimes \delta', \Sigma \vdash \sigma}
			
			\\[.33cm]
			
			\prftree[r]
			{($\otimes$-left-$\Sigma$)}
			{\Omega; \Theta, \Delta, \Sigma, \sigma', \sigma'' \vdash \sigma}
			{\Omega; \Theta, \Delta, \Sigma, \sigma' \otimes \sigma'' \vdash \sigma}
			
			\\[.33cm]
			
			\prftree[r]
			{($\otimes$-right)}
			{\Omega; \Theta, \Delta, \Sigma \vdash \sigma\quad}
			{\Omega; \Theta', \Delta', \Sigma' \vdash \sigma'}
			{\Omega; \Theta,  \Theta', \Delta, \Delta', \Sigma, \Sigma' \vdash \sigma \otimes \sigma'}
			
			\qquad

				\prftree[r]
				{($\multimap$-left)}
				{\Omega; \Sigma \vdash \res@\usr}
				{\Omega; \Sigma, \res@\usr \multimap \res'@\usr' \vdash \res'@\usr'}

			\\[.33cm]
			
			\prftree[r]
			{($\linearcontract$-left)}
			{\delta \subseteq \delta'}
			{\Omega; \Theta, \Delta, \delta', \Sigma \vdash \sigma}
			{\Omega; \Theta, \delta \linearcontract \delta', \Delta, \Sigma \vdash \sigma}
			
			\qquad
			
			\prftree[r]
			{($\linearcontract$-split)}
			{\Omega; \Theta, \delta \otimes \delta'' \linearcontract \delta' \otimes \delta''', \Delta, \Sigma \vdash \sigma}
			{\Omega; \Theta, \delta \linearcontract \delta', \delta'' \linearcontract \delta''', \Delta, \Sigma \vdash \sigma}
			
			\\[0.3cm]
			\textbf{Linear Non-linear Interaction Rules}\\
			
			\\[-0,2cm]
			
			\prftree[r]
			{(G-left-$\theta$)}
			{\Omega; \Theta, \theta, \Delta, \Sigma \vdash \sigma}
			{\Omega, G(\theta); \Theta, \Delta, \Sigma \vdash \sigma}
			
			\qquad
			
			\prftree[r]
			{(G-left-$\delta$)}
			{\Omega; \Theta, \Delta, \delta, \Sigma \vdash \sigma}
			{\Omega, G(\delta); \Theta, \Delta, \Sigma \vdash \sigma}			
			
			\\[.33cm]
			
			\prftree[r]
			{($\Omega$-cut)}
			{\Omega \Vdash \omega}
			{\Omega', \omega; \Theta, \Delta, \Sigma \vdash \sigma}
			{\Omega, \Omega';  \Theta, \Delta, \Sigma \vdash \sigma}	
		\end{tabular}
	\end{center}
	\caption{\MuACL\ rules.}
	\label{fig:MuACLfull}
\end{figure}

So far, we have characterised fairness at the basic level of exchange environment and we have introduced a language for expressing the users' policies.
As said, verifying that an exchange respects a single user's policy can be done locally, but ruling out double spending requires a global check.
Clearly, it is crucial to devise a sound technique and a tool that users can rely on for proving an exchange fair.
To do that, we still keep the logical flavour of MuAC and we define the decidable logic \MuACL\ 
that characterises fair exchanges,
to which we compile MuAC rulesets.
Then, we show that an exchange is fair if and only if there is a \MuACL\ proof of it, which can also be used as a witness of fairness for the TTP.

\subsection{A Logic for MuAC}\label{sec:col-logic}

The logic \MuACL\ 
it is basically a linear logic with a non-linear fragment in the spirit of \emph{LNL}~\cite{LNL}.
The non-linear part encodes reasoning on the predicates $p \in \Pred$ and on the context $C$.
The linear part encodes  exchanges and resource ownership (represented by atomic predicates $\res@\usr$ stating that a resource $\res$ belongs to the user $\usr$).
The linear fragment has also an operator to express the typical offer/return in contracts, inspired by PCL~\cite{BZ}, not expressible in standard linear logic.

The syntax of \MuACL\ propositions follows.
In the style of~\cite{Kanovich94}, we define a different syntactic category for each class of propositions with its specific constraints on operators.
\begin{defi}[MuACL propositions]\label{def:prop-muacl}
	A \MuACL\ proposition is a formula $\varphi$ in either form $\sigma$, $\delta$, $\theta$ or $\omega$, as defined below
	\begin{align*}
		\sigma &::= I \mid \res@\usr \mid \sigma \otimes \sigma\\
		\delta &::= I \mid \res@\usr \multimap \res'@\usr' \mid \delta \otimes \delta\\
		\theta &::= \delta \linearcontract \delta\\
		\omega &::= \top \mid p(\usr_1, \dots, \usr_n) \mid \omega \land \omega \mid \omega \rightarrow \omega \mid G( \theta) \mid G (\delta)
	\end{align*}
\end{defi}
We refer to the common resource-based interpretation of linear logic for describing the intuitive meaning of the propositions above~\cite{Porello2010}.
Moreover, we abuse the notation: tensor products are seen as multisets, given that the conjunction $\otimes$ is associative and commutative ($I$, standing for \emph{true}, is similarly seen as the empty multiset).

A proposition $\sigma$
is a multiset of atomic linear predicates representing resource ownership, namely the computation states.
A proposition $\delta$ 
is an exchange, i.e., a linear conjunction of linear implications representing transfers, where %
 $\multimap$ 
is the usual linear implication.
A proposition $\theta$
is a linear contract
defined via the new operator $\delta \linearcontract \delta'$, called \emph{linear contractual implication}.
Roughly, it states that the promised exchange $\delta'$ will eventually be performed provided that $\delta$ is \emph{true}.
Finally, 
a proposition $\omega$
represents non-linear knowledge where $\top$, $\land$ and $\to$ are 
the usual classical operators,
 $p(\usr, \dots, \usr')$ is an atomic non-linear predicate encoding a relation among users, and $G$ ``lifts'' a linear formula to a non-linear one $\omega$. 

\begin{exa}
	A state where Alice has one $\heal$ and Bob two $\heavy$ resources is represented as
	\[
	\sigma = \heal@Alice \otimes \heavy@Bob \otimes \heavy@Bob
	\]
	The exchange where Alice is giving both her $\heal$ to Bob is
	\[
	\delta = (\heal@Alice \multimap \heal@Bob) \otimes (\heal@Alice \multimap \heal@Bob)
	\]
	A contract stating that Alice has agreed to give a $\heal$ to Bob if she receives a $\heavy$ from him is
	\[
	\theta = (\heavy@Bob \multimap \heavy@Alice) \linearcontract (\heal@Alice \multimap \heal@Bob)
	\]
	A policy saying that Alice is willing to accept the contracts as before is represented through the non-linear propositions $G(\theta)$.
\end{exa}

The sequents of \MuACL\ are defined as follows.

\begin{defi}[\MuACL\ Sequent]
A \MuACL\ \emph{sequent} is of form
\begin{align*}
\Omega; \Theta, \Delta, \Sigma \vdash \sigma.
\end{align*}
where $\Sigma$, $\Delta$, $\Theta$ and $\Omega$ are multisets with elements $\sigma$, $\delta$, $\theta$ and $\omega$, respectively.
A sequent is \emph{initial} if $\Theta, \Delta = \emptyset$,
i.e., 
$\Omega; \emptyset, \emptyset, \Sigma \vdash \sigma$, 
that we write also as $\Omega; \Sigma \vdash \sigma$
since from now onwards we will omit the empty
components for brevity.

\end{defi}

The \MuACL\ judgments have either one of the following forms:
\[
\Omega \Vdash \omega    \hspace{1.5cm}    \Omega; \Theta, \Delta, \Sigma \vdash \sigma
\]
Roughly, the left one is for non-linear reasoning and the right for mixed linear non-linear reasoning.
Also, $\Omega; \Theta, \Delta, \Sigma \vdash \sigma$ intuitively means that the state $\sigma$ is a possible transformation of $\Sigma$ under the assumption $\Omega; \Theta, \Delta$, representing the policies and some classical information $\Omega$, the proposed contracts $\Theta$ and the accepted exchanges $\Delta$.

The rules of \MuACL\ are in~\autoref{fig:MuACLfull}.
The \emph{non-linear} rules for $\Vdash$ are the standard ones of 
the fragment of non-linear logic we consider
(i.e., without negation and disjunction).
These rules
 are displayed in the top-most part of
the figure.
In the second block of rules from top,
we follow~\cite{LNL}: for each structural and left non-linear rule, such as (Weak), there is 
an almost duplicated
 \emph{non-linear L-rule} for $\vdash$ that modifies $\Omega$ in the same way, such as (L-Weak).

The \emph{linear rules} for $\vdash$ are in the third block.
They result from instantiating the standard ones on the \MuACL\ sequents.
A result of this instantiation is that some rules are replicated for each syntactic category of~\autoref{def:prop-muacl}, while other standard rules of  linear logic are omitted,
in particular the cut rule.
Indeed, the omitted rules are immaterial to our modeling purpose, and may alter our decidability and correctness results.
There
are two rules for the linear contractual implication: the ($\linearcontract$-left) rule introduces the operator on the left if what is required by the contract is satisfied by the consequences; the
($\linearcontract$-split) rule deals with composition of contracts.
Intuitively, when building proofs bottom-up, ($\linearcontract$-left) enables us to transform the state $\Sigma$ with $\delta'$ if $\delta'$ is the consequence of a satisfied promise, whereas
the rule ($\linearcontract$-split) expresses that contracts can be combined 
in order for the requirements of a contract to match the promised resources of other contracts.

Finally, the remaining rules
govern the interaction between linear and non-linear derivations~\cite{LNL}.
The rules ($G$-left-$\theta$) and ($G$-left-$\delta$) say that a $G$-labeled linear formula is 
non-linear,
and ($\Omega$-cut) is the cut rule where the left premise uses $\Vdash$ and the right one $\vdash$. 

\begin{exa}
		A linear implications $\res@\usr \multimap \res@\usr'$ naturally represents an exchange where a predicate $\res@\usr$ is consumed and a new $\res@\usr'$ is created.
		Note for example that $\res@\usr \multimap \res@\usr',$ $\res@\usr \vdash \res@\usr'$ is indeed a valid sequent.
		
		Linear contractual implication $\delta \linearcontract \delta'$ encodes a promise of $\delta'$ in return of $\delta$.
\end{exa}

\begin{exa}\label{ex:contractreason}
	We now represent agreements and exchanges of our running example of Section~\ref{sec:col-runex} in \MuACL{}.
	Circular promises like those of~\autoref{ex:col-dir} are expressed by a sequent of the form
	\[
	\delta \linearcontract \delta', \delta' \linearcontract \delta, \Sigma \vdash \sigma\,,
	\]
	where the exchange $\delta'$ is promised in return for $\delta$ and vice versa.
	The following derivation proves that the exchange is fair, provided that $\delta, \delta', \Sigma \vdash \sigma$, intuitively meaning that $\delta, \delta'$ transform the state $\Sigma$ in $\sigma$:
	\begin{align*}
		\prftree[r]
		{($\linearcontract$-split)}
		{
			\prftree[r]
			{($\linearcontract$-left)}
			{\delta \otimes \delta' \subseteq \delta' \otimes \delta}
			{
				\prftree[r]
				{($\otimes$-left-$\Delta$)}
				{\delta, \delta', \Sigma \vdash \sigma}
				{\delta' \otimes \delta, \Sigma \vdash \sigma}
			}
			{\delta \otimes \delta' \linearcontract \delta' \otimes \delta, \Sigma \vdash \sigma}	
		}
		{\delta \linearcontract \delta', \delta' \linearcontract \delta, \Sigma \vdash \sigma}	
	\end{align*}
	Similarly for a circular exchange like the one of~\autoref{ex:col-cir} represented as 
	\[\delta \linearcontract \delta', \delta' \linearcontract \delta'',\delta'' \linearcontract \delta, \Sigma \vdash \sigma.\]
	The derivation 
requires that $\delta, \delta', \delta'', \Sigma \vdash \sigma$ and uses two applications of ($\linearcontract$-split).
\end{exa}

A must for \MuACL\ to be adequate for reasoning about MuAC semantics is that of being decidable.

\begin{restatable}[\MuACL\ decidability]{thm}{MuACLsdec}\label{thm:MuACLdecide}
	An always-terminating algorithm exists that decides if an initial sequent is valid in \MuACL.
\end{restatable}

The above theorem mentions initial sequents that are sufficient to reason about fairness of exchanges as \autoref{th:fair-exchange} will make clear.
An overview of the proof of this theorem is given in the next subsection.

\subsubsection*{Overview of the proof of \MuACL\ decidability}
At the high level, we proceed as follows to prove the decidability of \MuACL.
First, we define 
two normal forms for proofs (numbered $1$ and $2$),
 and show that they are general, i.e., a proof exists for an initial sequent only if a proof in normal form exists.
Then, we reduce the problem of finding a proof in the normal form $1$ to reachability in Petri Nets, which is known to be decidable~\cite{PetriReach}.
Finally, we reduce the problem of finding a proof in the normal form $2$ to a proof in the normal form $1$.

The following notation helps:
\begin{nota}\label{not:rulesets}
Let $Sr$, $Cr$, $Lr$, $Gr$, $Pr$ be sets of \MuACL\ rules defined as follows.
\begin{align*}
	Sr &= \{ \text{(L-Weak), (L-Cont)} \}\\
	Cr &= \{ \text{($\top$-right), ($\Omega$-Ax), (Cont), (Weak), ($\land$-left1),}
	 \ \text{($\land$-left2), ($\rightarrow$-left), ($\rightarrow$-right), } \\ 
	& \qquad \text{(L-$\land$-left1), ($\Omega$-Cut)} \} \\
	Lr &=  \{ \text{($\multimap$-left), ($\otimes$-right), ($\otimes$-left-$\Theta$),}
         \ \text{($\otimes$-left-$\Delta$), ($\otimes$-left-$\Sigma$)}
         \}\\
	Gr &= \{ \text{(G-left-$\theta$), (G-left-$\delta$)} \}\\
	Pr &= \{ \text{($\linearcontract$-left), ($\linearcontract$-split)} \}
\end{align*}
Intuitively, the set $Sr$ contains structural rules; the rules in $Cr$ and $Lr$ are those  for the non-linear and the linear fragments, respectively; $Gr$ contains the rules driving the interactions between the two fragments; and the rules $Pr$ govern the contractual implication.  
\end{nota}

In the following, we will call \emph{proof} the derivation of a theorem from the axioms, and only use the term \emph{derivation} for a derivation with open assumptions, i.e., a proof tree where the leaves are not only axioms.
We also say that two proofs are \emph{equivalent} if they prove the same sequent.
Moreover, for a set $A$ of rules, we write $\Pi_A$ for a proof or derivation that only applies rules in $A$. 
Finally, we write $\Omega_G$ for a multiset that only contains formulas of the forms $G(\theta)$ and $G(\delta)$.
Recall also that in an initial sequent $\Theta$ and $\Sigma$ are empty.

\begin{figure}
	\begin{center}
		\begin{tabular}{c}
			\begin{tabular}{c}
				\prfsummary[$\Pi_{Cr \cup Sr}$]
				{
					\prfsummary[$\Pi_{Gr \cup Sr}$]
					{
						\prftree
						{\Pi_{Lr \cup \{ \text{($\Sigma$-Ax), (I-right)}\}}}
						{\Delta, \Sigma \vdash \sigma}
					}
					{\Omega_G; \Sigma \vdash \sigma}
				}
				{\Omega; \Sigma \vdash \sigma}\\[0.2cm]
				\textit{normal form 1}
			\end{tabular}
			
			\begin{tabular}{c}
				\hspace{0.5cm}	
				\prfsummary[$\Pi_{Cr \cup Sr}$]
				{
					\prfsummary[$\Pi_{Gr \cup Sr}$]
					{
						\prfsummary[$\Pi_{\{ \text{($\smlinearcontract$-split)} \}}$]
						{
							\prftree[r]
							{($\linearcontract$-left)}
							{
								\prftree
								{\Pi_{Lr \cup \{ \text{($\Sigma$-Ax), (I-right)}\}}}
								{\Delta, \Sigma \vdash \sigma}
							}
							{\theta, \Delta', \Sigma \vdash \sigma}
						}
						{
							{\Theta, \Delta', \Sigma \vdash \sigma}
						}
					}
					{\Omega_G; \Sigma \vdash \sigma}
				}
				{\Omega; \Sigma \vdash \sigma}\\[0.2cm]
				\textit{normal form 2}
			\end{tabular}
		\end{tabular}
	\end{center}
	\caption{Normal forms for \MuACL\ proofs.}
	\label{fig:normals}
\end{figure}

\begin{defi}[Normal proofs]\label{def:normalforms}
	A \MuACL\ proof for an initial sequent is \emph{normal} if 
	it can be decomposed in either one of the 
	forms in~\autoref{fig:normals}.
\end{defi}
Notice that the two normal forms essentially coincide, but normal form 1 corresponds the cases in which the proof uses no rules in $Pr$ (the ones for $\linearcontract$), and normal form 2 to those in which instead they appear at least once.
We use the former to exploit known results for linear logic, and the latter for addressing the  cases that are peculiar to \MuACL.
\autoref{app:proofs} contains some auxiliary definitions and lemmata that help proving that we can only consider normal proofs in either form 1 or 2, as stated by the following theorem.
\begin{thm}[\MuACL\ Normal proofs]\label{thm:col-normal-ncr}
	For 
	every
	$\Omega, \Sigma, \sigma$, the initial sequent $\Omega; \Sigma \vdash \sigma$ is valid in 
	\MuACL\ if and only if a normal proof $\  \Pi$ exists for $\Omega; \Sigma \vdash \sigma$.
\end{thm}

As a second auxiliary result we get rid of $\Pi_{Cr\cup Sr}$ in both forms by showing that we can build a canonical $\Omega_\star$ from $\Omega$ such that:
$(i)$ $\Omega, \Sigma \vdash \sigma$ is always derivable from $\Omega_\star, \Sigma \vdash \sigma$ using only rules in $Cr\cup Sr$, and
$(ii)$ 
every
proof for $\Omega_G, \Sigma \vdash \sigma$ can be transformed into one for $\Omega_\star, \Sigma \vdash \sigma$.
We build $\Omega_\star$ by including a single occurrence of every $G(\delta)$ and $G(\theta)$ appearing as a subterm in $\Omega$ with valid classical preconditions.

	Our next step is proving that the existence of a \MuACL\ proof in the normal form 1 for a given initial sequent is decidable.
	Note that proofs in the normal form 1 correspond to the case where no contractual rule is ever applied, and where linear implications can be used ad libitum for building the proof in a bottom-up approach (roughly, $G$ is the same as the bang operator ($!$) of linear logic).
	Then, decidability follows from a suitable application of Kanovich's technique~\cite{Kanovich94} that reduces the problem to reachability in Petri Nets, which can be decided using the algorithm proposed in~\cite{PetriReach}.

\begin{restatable}[\MuACL\ Normal form 1 decidability]{lem}{fstnfdecidencr}\label{thm:fstnfdecide-ncr}
		An always-terminating algorithm exists that decides if an initial sequent is provable in \MuACL\ using a proof in the normal form 1.	
\end{restatable}

Finally, we show how to reduce the normal form 2 case to the previous one: we prove that a proof in the normal form 2 can be effectively rewritten in the normal form 1.
Consider a vector space with 
a basis composed by the linear implications appearing as subterms in $\Omega_\star$ 
(i.e., all the transfers that we are considering).
Note that every $\Delta$ (and $\delta$) is uniquely determined by a vector $\bar u_{\Delta}$ (and $\bar u_{\delta}$), associating each linear implication with the number of occurrences in $\Delta$.
The reduction from the normal form 2 to the normal form 1 will be performed in this linear algebraic framework.

In the following, we consider the derivations in a bottom-up fashion, starting with the sequent we are proving and deriving the premises.
Consider the normal form 1, and note that no contractual rule is ever applied, hence we can assume $\Omega_G$ only contains formulas of the form $G(\delta)$.
We let a vector $\bar{x}$ represent how many occurrences of each $\delta$ rule we take in the derivation $\Pi_{Gr \cup Sr}$.
The set $\Omega_G$ itself can be represented as a linear transformation $A_{\Omega_G}$, with $\bar u_\delta$ its columns, mapping each vector $\bar{x}$ with the outcome of taking $x_i$ occurrences of each $\delta_i$ rule.
Each $\Delta$ is the outcome of composing a number of occurrences (non-negative, possibly $0$) of 
every
$\delta$ such that $G(\delta) \in \Omega_G$.

Formally, a derivation $\Pi_{Gr \cup Sr}$ exists from $\Delta, \Sigma \vdash \sigma$ to $\Omega_G, \Sigma \vdash \sigma$ if and only if $\bar u_{\Delta} = A_{\Omega_G} \bar{x}$ with $\bar{x}$ a vector of non-negative integers.
Note that also the opposite is true: we can always interpret 
a matrix
 $A$ as a specific set of rules $G(\delta)$ of some $\Omega_G$.

Consider now the normal form 2.
We encode $\Omega_G$ as three matrices: $A_{\Omega_G}$ defined as before; 
$B_{\Omega_G}$ and $C_{\Omega_G}$ with a column $\bar b_\theta$ and $\bar c_\theta$ for each rule $\theta$ such that $G(\theta) \in \Omega_G$.
The vector $\bar b_\theta$ represents the required transfers appearing to the left of $\linearcontract$ in $\theta$, whereas $\bar c_\theta$ represents the promised transfers to the right.
Note that we can take 
every
 rule in $\Omega_G$ as many times as we want, and assume $\bar{y}$ is a vector representing how many occurrences for each rule in $\Omega_G$ we take.

The encoding represents agreement as the solutions of a system of linear equations (representing possible compositions of offers) 
constrained by linear inequalities (representing fairness).
Formally, an exchange $\Delta$ is the result of a fair agreement if and only if its encoding as the vector $\bar u_{\Delta}$ satisfies 
\begin{align*}
	\bar u_{\Delta} =
	\begin{bmatrix}
		\begin{array}{c|c}
			{A_{\Omega_G}} & {C_{\Omega_G}}\\
		\end{array}
	\end{bmatrix}
		\bar{y}
	\quad\text{ and }\quad
	\begin{bmatrix}
		\begin{array}{c|c}
			\mathbf{0}	 & C_{\Omega_G} - B_{\Omega_G}\\
		\end{array}
	\end{bmatrix}
		\bar{y}
	\geq
	\bar{0}
\end{align*}
for some column vector of non-negative integers $\bar{y}$.
We then apply the Hilbert basis theorem~\cite{HB} to show that the set of nonnegative integer solutions $\bar{y}$ of the inequality above are generated by
$ \bar{y} = 
	\begin{bmatrix}
		H_{\Omega_G}\\
	\end{bmatrix}
	\bar{x}
$ for 
every
 $\bar{x}$ of nonnegative integers, where the matrix $H_{\Omega_G}$ can be computed using~\cite{computeHB}.
As a consequence, an exchange $\Delta$ results from a fair agreement if and only if $\bar u_{\Delta} = D_{\Omega_G} \bar{x}$
for some column vector of non-negative integers $\bar{x}$ and with $
D_{\Omega_G} = \begin{bmatrix}
	\begin{array}{c|c}
		{A_{\Omega_G}} & {C_{\Omega_G}}\\
	\end{array}
\end{bmatrix}
\begin{bmatrix}
	\begin{array}{c}
		H_{\Omega_G}
	\end{array}
\end{bmatrix}
\!.
$

This is exactly our encoding of the proofs in the normal form 1.
Finally, by applying the encoding backward we interpret $D_{\Omega_G}$ as a multiset of \MuACL\ non-linear propositions $\Omega_G'$ without contractual implications, and prove the following lemma.
\begin{restatable}{lem}{sftoff}\label{thm:col-sftoff}
	For 
	every
	 $\Omega_G, \Delta, \Sigma, \sigma$, there is a computable multiset of non-linear propositions $\Omega_G'$ such that there exists a derivation in the normal form 2 from $\Delta, \Sigma \vdash \sigma$ to $\Omega_G; \Sigma \vdash \sigma$ if and only if there exists a derivation in the  normal form 1 from $\Delta, \Sigma \vdash \sigma$ to $\Omega_G'; \Sigma \vdash \sigma$.
\end{restatable}
We can therefore conclude the decidability of \MuACL.
\subsection{MuACL vs Linear Logic}\label{sec:independence}

We now show that the linear contractual implication is necessary for representing the circular reasoning of contracts, otherwise expressive power is lost.
Indeed, we prove that the traditional operators of linear logic cannot express predicates with $\linearcontract$,
by rewriting them in a fully compositional way.
Some useful notation follows.

Let \exlo\ be the logic obtained by removing from \MuACL\ the formulas with $\linearcontract$ and the rules governing it.
Note that \exlo\ is essentially the computational fragment of LNL, which in turns is an alternative way of expressing linear logic, where the non-linear fragment is made explicit by the $G$ operator.
We show that $\linearcontract$ is not just syntactic sugar by proving that there is no homomorphic map $m$ from \MuACL\ to \exlo\ (multisets of) formulas, i.e., $m$ preserves the operators of \exlo.

\begin{defi}
A \emph{homomorphic map} from \MuACL\ to \exlo\ is a function $m(\cdot)$ from \MuACL\ to \exlo\ propositions that preserves the operators of \exlo\, i.e.,  such that for all \MuACL\ propositions $\varphi$ and $\varphi'$ it holds that:
\begin{align*}
m(\varphi) &= \varphi &&\text{if } \varphi = I, \top, \res@\usr, p(\usr_1, \dots, \usr_n)\\
m(\varphi \star \varphi') &= m(\varphi) \star m(\varphi') &&\text{if } \star = \otimes, \multimap, \land, \rightarrow\\
m(G(\varphi)) &=\ G(m(\varphi)) &&
\end{align*} 
\end{defi}

Notice that $m(\cdot)$ has no constraints when it maps propositions with $\linearcontract$ as the top level operator.
Hereafter, we apply $m(\cdot)$ pointwise to multisets of \MuACL\ propositions.
Moreover, we write $\vdash_{\MuACL}$ and $\vdash_{\exlo}$ to represent sequents that are valid in $\MuACL$ and $\exlo$, respectively.

\begin{defi}
A homomorphic map $m$ from \MuACL\ to \exlo\ is
\begin{itemize}
	\item  \emph{complete} if 
	for all $\Omega$, $\Theta$, $\Delta$, $\Sigma$ and $\sigma$, it holds that 
	\[
	\Omega; \Theta, \Delta, \Sigma \vdash_{\MuACL} \sigma \quad\text{ implies }\quad m(\Omega); m(\Theta), m(\Delta), m(\Sigma) \vdash_{\MuACL} m(\sigma);
	\]
	\item \emph{correct} if
	for all $\Omega$, $\Theta$, $\Delta$, $\Sigma$ and $\sigma$, it holds that 
	\[
	m(\Omega); m(\Theta), m(\Delta), m(\Sigma) \vdash_{\MuACL} m(\sigma)
	\quad\text{ implies }\quad 
	\Omega; \Theta, \Delta, \Sigma \vdash_{\MuACL} \sigma.
	\]
	\end{itemize}
 \end{defi}

The following theorem ensures that no correct and complete homomorphic map is possible:
\begin{restatable}{thm}{nohomo}\label{thm:homo}
There is no complete and correct homomorphic map of \MuACL\ to \exlo.
\end{restatable}
The proof of the theorem is by counterexample, and the sequent presented is indeed an initial one, showing that finding a homomorphic map that works well at least with initial sequents is also impossible (see~\autoref{thm:homoinit} in the Appendix).
Consequently, the computational fragment of linear logic does not natively support circular reasoning.
For this reason, \MuACL\ extends it with the operator $\linearcontract$ achieving a different expressive power.

\subsection{Compiling MuAC to MuACL}\label{sec:MuACencode}

The definition below compiles MuAC to \MuACL\ and paves the way to use \MuACL\ for proving an exchange fair.

As abbreviations, we write $[u] = u_0, \dots, u_n$ for the \emph{finite} sequence of user variables occurring in a MuAC rule $r$, and we denote with the symbol $\Uplambda [u]$ a restricted universal quantifier over the users $[u]$; 
note that this quantifier is only syntactic sugar used to compactly represent a finite conjunction of propositions $\omega$ over the \emph{finite} set of users $[u]$.

\begin{defi}[From MuAC to \MuACL]\label{def:compile}
The compilation of the MuAC ruleset $R_\usr$ of the user $\usr \in Usr$, in symbols $\semantics{R_\usr }$,  is defined as follows:
\begin{align*}
&\semantics{R_\usr } = \{ \semanticsusr{r}{\usr} \mid r \in R_\usr \}\\
&\semanticsusr{\texttt{Gives}(\texttt{Me}, \res, u) \texttt{ :- } \GiveLs \ \texttt{with}\ \PredLs}{\usr} = \\
&\qquad \Uplambda [u].\semanticsusr{\PredLs}{\usr} \rightarrow\ 
G(\semanticsusr{\GiveLs}{\usr} \linearcontract \semanticsusr{\texttt{Gives}(u, \res, \texttt{Me})}{\usr} )
\end{align*}
where 
	\begin{align*}
		&\semanticsusr{\texttt{Gives}(u, \res', u')}{\usr} = \res'@\semanticsusr{u}{\usr} \multimap \res'@\semanticsusr{u'}{\usr}\\
		&\semanticsusr{\PredLs}{\usr} = 
		\begin{cases}
		\top 
		&\text{if }\PredLs = \epsilon\\
		p(\semanticsusr{u_1}{\usr}, \dots, \semanticsusr{u_i}{\usr}) \land \semanticsusr{\PredLs'}{\usr} 
			\hspace{-.2cm}&\text{if } \PredLs = p(u_1, \dots, u_i\texttt{)}\PredLs'
		\end{cases}\\
		&\semanticsusr{\GiveLs}{\usr} =
		\begin{cases}
		I 
			&\text{if }\GiveLs = \epsilon\\
		\semanticsusr{\texttt{Gives}(u, \res', u')}{\usr} \otimes \semanticsusr{GiveLs'}{\usr}
			\hspace{-.2cm}&\text{if } \GiveLs = \texttt{Gives}(u, \res', u')\ \GiveLs'
		\end{cases}\\
		&\text{ with } \semanticsusr{u}{\usr} =
		\begin{cases}
			\usr &\text{if } u = \texttt{Me}\\
			u   &\text{otherwise }
		\end{cases}
	\end{align*}
\end{defi}

Some comments are in order.
The compilation of a ruleset $R_\usr$ is a set of non-linear formulas, one for each rule $r \in R_\usr$.
A rule $\texttt{Gives}(\texttt{Me}, \res, u) \texttt{ :- } \GiveLs \ \texttt{with}\ \PredLs$ is compiled as a universally quantified non-linear formula $\Uplambda [u] . \omega \rightarrow G(\delta \linearcontract \delta')$ where:
($i$) $\omega$ encodes the non-linear conditions in $\PredLs$;
($ii$) $\delta$ represents the (linear) exchanges the user asks in return for $\res$;
and ($iii$) $\delta'$ corresponds to the promise of $\usr$ to give $\res$ to the requester if the conditions are met.
Recall that a MuAC statement $\texttt{Gives}(u, \res, u')$ intuitively represents an exchange, where $u$ gives a resource $\res$ to $u'$, i.e., $\res@u \multimap \res@u'$.
As expected, the non-linear requirements of $r$ are joined with $\land$ and the linear ones with $\otimes$.
Finally, user variables $u$ are bound to users in $\Usr$ by the finite universal quantifier $\Uplambda$, with the exception of \texttt{Me}, which is interpreted as $\usr$, the owner of the ruleset.

\begin{exa}\label{ex:somerules}
	Consider the MuAC rulesets of~\autoref{ex:col-MuACpol}.
	The rules \texttt{A1} of Alice's policy, \texttt{B1} of Bob's, and \texttt{C1} of Carl's are compiled as
	\begin{align*}
	&\Uplambda u, u'. \top \rightarrow G((\heavy@u \multimap \heavy@\mathit{Alice}) \linearcontract (\spell@\mathit{Alice} \multimap \spell@u')),\\
	&\Uplambda u, u'. \top \rightarrow G((\spell@u \multimap \spell@\mathit{Bob}) \linearcontract (\light@\mathit{Bob} \multimap \light@u')),\\	
	&\Uplambda u, u'. \top \rightarrow G((\light@u \multimap \light@\mathit{Carl}) \linearcontract (\heavy@\mathit{Carl} \multimap \heavy@u')).
	\end{align*}
\end{exa}

\subsection{Proving the Fairness of Exchanges}

Before completing our tour on applying \MuACL\ to verify the fairness of exchanges, we need to translate states and contexts.

\begin{defi}\label{def:compile-state}
A state $\MuACstate$ is compiled into a multiset of \MuACL\ atoms as follows.
\[
\semden{\MuACstate}(\res@\usr) = \MuACstate(\usr)(\res)
\]
In addition, a context $C$ is compiled as follows
\[
\forall p.\,\semden{C} \Vdash p(\usr, \dots \usr') \text{ iff } (\usr, \dots \usr') \in C(p)
\]
\end{defi}
Note that the definition above constrains us to only consider contexts such that their compilation returns a finite non-linear theory.
\begin{exa}\label{ex:comp-state}
	Consider our running example of Section~\ref{sec:col-runex}.
	The state at the beginning of the exchanges is represented as
	\begin{align*}
		&\Sigma_0 = \{\spell@\mathit{Alice}, \light@\mathit{Bob}, \heavy@\mathit{Carl}, \heavy@\mathit{Carl}, \heavy@\mathit{Carl}, \heal@\mathit{Carl}, \heal@\mathit{Carl}\}.
	\end{align*}
	Since Bob and Carl are paladins, the context is compiled as
	\begin{align*}
		&\semden{C} = \{is\_paladin(\mathit{Bob}), is\_paladin(\mathit{Carl})\}.
	\end{align*}
\end{exa}
In the theorem below the MuAC rulesets, the context, and the current state $\MuACstate$ determine the left part of an initial sequent, whereas the right part is for the next state $\MuACstate'$ reachable with $\MuACstate \xrightarrow{\exc} \MuACstate'$.  
Then, $\exc$ is fair if and only the obtained initial sequent is valid, and its proof is a witness of fairness.
\begin{restatable}[Fairness = Validity]{thm}{correction}\label{th:fair-exchange}
Let  $(\MuACStates, \rightarrow)$ be an exchange environment; 
let $R_\usr$ be the MuAC ruleset of the user $\usr$; 
let $\MuACstate$ and $\MuACstate'$ be states in $\MuACStates$; and
let $C$ be a context.

Then, the transition
$\MuACstate \xrightarrow{\exc} \MuACstate'$ is fair
if and only if
$\biguplus_{\usr \in \Usr} \semden{R_\usr}, \semden{C}; \semden{\MuACstate} \vdash \semden{\MuACstate'}$ is valid in \MuACL.
\end{restatable}

\begin{figure*}
	{
	\begin{center}
		\[
			\prftree[r,double,straight]
			{(L-Weak)}
			{
				\prftree[r,double,straight]
				{(L-$\land$-left)}
				{
					\prftree[r,double,straight]
					{(L-$\rightarrow$-left)}
					{
						\begin{matrix}
						\\
						\\
						\Vdash \top\qquad
						\end{matrix}				
					}
					{
						\prftree[r, double]
						{(G-left)}
						{
							\prftree[double, dashed]
							{ \text{Same as~\autoref{ex:contractreason}}}
							{	\begin{matrix}
									(\heavy@\mathit{Carl} \multimap \heavy@\mathit{Alice}) \linearcontract (\spell@\mathit{Alice} \multimap \spell@\mathit{Bob}),\\
									(\spell@\mathit{Alice} \multimap \spell@\mathit{Bob}) \linearcontract (\light@\mathit{Bob} \multimap \light@\mathit{Carl}),\\
									(\light@\mathit{Bob} \multimap \light@\mathit{Carl}) \linearcontract (\heavy@\mathit{Carl} \multimap \heavy@\mathit{Alice}), \Sigma_0\\
								\end{matrix}
								\vdash \semden{\MuACstate'}
							}
						}
						{
							\begin{matrix}
								G((\heavy@\mathit{Carl} \multimap \heavy@\mathit{Alice}) \linearcontract (\spell@\mathit{Alice} \multimap \spell@\mathit{Bob})),\\
								G((\spell@\mathit{Alice} \multimap \spell@\mathit{Bob}) \linearcontract (\light@\mathit{Bob} \multimap \light@\mathit{Carl})),\\
								G((\light@\mathit{Bob} \multimap \light@\mathit{Carl}) \linearcontract (\heavy@\mathit{Carl} \multimap \heavy@\mathit{Alice})); \Sigma_0\\
							\end{matrix}
							\vdash \semden{\MuACstate'}
						}
					}
					{
						\begin{matrix}
							\top \rightarrow G((\heavy@\mathit{Carl} \multimap \heavy@\mathit{Alice}) \linearcontract (\spell@\mathit{Alice} \multimap \spell@\mathit{Bob})),\\
							\top \rightarrow G((\spell@\mathit{Alice} \multimap \spell@\mathit{Bob}) \linearcontract (\light@\mathit{Bob} \multimap \light@\mathit{Carl})),\\
							\top \rightarrow G((\light@\mathit{Bob} \multimap \light@\mathit{Carl}) \linearcontract (\heavy@\mathit{Carl} \multimap \heavy@\mathit{Alice}));
							\Sigma_0\\
						\end{matrix}
						\vdash \semden{\MuACstate'}
					}
				}
				{
					\begin{matrix}
						\Uplambda u, u' \ldotp \top \rightarrow G((\heavy@u \multimap \heavy@\mathit{Alice}) \linearcontract (\spell@\mathit{Alice} \multimap \spell@u')),\\
						\Uplambda u, u' \ldotp \top \rightarrow G((\spell@u \multimap \spell@\mathit{Bob}) \linearcontract (\light@\mathit{Bob} \multimap \light@u')),\\
						\Uplambda u, u' \ldotp \top \rightarrow G((\light@u \multimap \light@\mathit{Carl}) \linearcontract (\heavy@\mathit{Carl} \multimap \heavy@u'));
						\Sigma_0\\
					\end{matrix}
					\vdash \semden{\MuACstate'}
				}
			}
			{
				\biguplus_{\usr \in \Usr} \semden{R_\usr}, \semden{C}; \semden{\MuACstate} \vdash \semden{\MuACstate'}
			}
		\]
	\end{center}}

		\caption{A \MuACL\ proof for \autoref{ex:col-cir} where double lines represent multiple applications of the same rule and dashed lines represent omitted trivial derivations.}
		\label{fig:logicexcir}
\end{figure*}

\begin{exa}\label{ex:cir-logic}
	Consider~\autoref{ex:col-cir} and $\MuACstate \xrightarrow{\exc} \MuACstate'$ where
	\begin{align*}
		\MuACstate &= \{ (\mathit{Alice}, \{\spell\}), (\mathit{Bob}, \{\light\}), (\mathit{Carl}, \{\heavy, \heavy, \heavy, \heal, \heal\}) \}\\
		\exc &= \{ \mathit{Alice} \xmapsto{\spell} \mathit{Bob}, \mathit{Bob} \xmapsto{\light} \mathit{Carl}, \mathit{Carl} \xmapsto{\heavy} \mathit{Alice} \}\\
		\MuACstate' &= \{ (\mathit{Alice}, \{\heavy\}), (\mathit{Bob}, \{\spell\}), (\mathit{Carl}, \{\light, \heavy, \heavy, \heal, \heal\}) \}
	\end{align*}
	The proof in~\autoref{fig:logicexcir} certifies the fairness of the transition, where $\Sigma_0$ is as in~\autoref{ex:comp-state}.
	We build the proof bottom-up, starting from the initial sequent of~\autoref{th:fair-exchange}.
	We first use the structural rules to select the configuration rules of $\biguplus_{\usr \in \Usr} \semden{R_\usr}$ to apply (in our derivation there is a single occurrence of \texttt{A1}, \texttt{B1} and \texttt{C1}, which are compiled as in~\autoref{ex:somerules}).
	We then use the non-linear rules and (G-left) for obtaining linear contracts $\theta$, where (L-$\rightarrow$-left) guarantees that the conditions of $\PredLs$ are satisfied.
	Finally, we obtain a sequent of the form $\Theta, \Sigma \vdash \sigma$ and we resolve circularity between promises and requirements in $\Theta$ (the three contractual implications in the topmost sequent) by applying the rules for $\linearcontract$ as we did in~\autoref{ex:contractreason}.
\end{exa}

\subsubsection*{Overview of the proof of correctness and completeness}

We first define a mapping from exchanges $\exc$ and policies $\pol_\usr$ to \MuACL\ predicates $\Delta_\exc$ and $\Omega_{\pol_\usr}$.
The mapping is injective up to commutativity and associativity of $\otimes$, hence invertible.

We then consider proofs, showing that a proof for 
$\biguplus_{\usr \in \Usr} \semden{R_\usr}, \semden{C}; \semden{\MuACstate} \vdash \semden{\MuACstate'}$ can always be transformed into one in the following 
form:
\[
	\prfsummary[$\Pi_{Cr \cup Sr}$]
	{
		\prfsummary[$\Pi_{Sr \cup Gr \cup Pr}$]
		{
						\prftree
						{\Pi_{Lr \cup \{ \text{($\Omega$-Ax), (I-right)}\}}}
						{\Delta, \semden{\MuACstate} \vdash \semden{\MuACstate'}}
		}
		{\biguplus_{\usr \in \Usr} \Omega_{\pol_\usr}; \semden{\MuACstate} \vdash \semden{\MuACstate'}}
	}
	{\biguplus_{\usr \in \Usr} \semden{R_\usr}, \semden{C}; \semden{\MuACstate} \vdash \semden{\MuACstate'}}
\]

As a first intermediate result, we show that the encoding of \MuACL\ rulesets is correct and complete.
Roughly, $\Pi_{Cr \cup Sr}$ exists if and only if $\Omega_{\pol_\usr}$ is the encoding of the interpretation of $R_{\usr}$ (i.e., if $\pol_\usr = \bigcup_{r \in R_\usr} \semdeno{r}\ C$).
Then we show that the derivation $\Pi_{Sr \cup Gr \cup Pr}$ can be obtained whenever $\Delta = \Delta_\exc$ for some $\exc$ that is accepted by the policies of all the users and where no double-spending occurs.
Finally, a proof $\Pi_{Lr \cup \{ \text{($\Omega$-Ax), (I-right)}\}}$ exists for $\Delta_\exc, \semden{\MuACstate} \vdash \semden{\MuACstate'}$ if and only if 
$\MuACstate \xrightarrow{\exc} \MuACstate'$ is a valid transition (such that users own what they are giving and where resources are preserved).

\subsection{Eventually fair computations}\label{subsec:fair-computations}

	Fair transitions can be combined by performing subsequent exchanges.
	Note that some states that can be reached with a fair computation, i.e., a sequence of transitions, cannot be reached with a single fair transition.
\begin{exa}
	Consider Alice, Bob and Carl with the following policies:
	\begin{align*}
	\Pol_{Alice} = \{ Alice \xmapsto{\res} Bob \triangleleft \emptyset \}\qquad
	\Pol_{Bob} = \{ Bob \xmapsto{\res} Carl \triangleleft \emptyset \}\qquad
	\Pol_{Carl} = \emptyset
	\end{align*}
	and assume that $st_\usr(\usr')(\res)$ equal $1$ when $\usr' = \usr$ and $0$ otherwise.
	Bob cannot straightly get $\res$ from Alice, as the direct exchange
	\[
	\MuACstate_{Alice} \xrightarrow{Alice\ \xmapsto{\res}\ Carl} \MuACstate_{Carl}
	\]
	is not fair. 
	However, Bob can persuade Carl to help him in that, and succeed in getting $\res$, because the following two-step computation is fair
	\[
	\MuACstate_{Alice} \xrightarrow{Alice\ \xmapsto{\res}\ Bob} \MuACstate_{Bob} \xrightarrow{ Bob\ \xmapsto{\res}\ Carl} \MuACstate_{Carl},
	\]		
	Of course, since the computation is fair, one would expect our framework to deem acceptable the transition
	$\MuACstate_{Alice} \xrightarrow{Alice\ \xmapsto{\res}\ Carl} \MuACstate_{Carl}$.
\end{exa}

Moreover, some exchanges of resources can be performed that are beneficial to all the involved users but can neither be performed in a single step (e.g., because of a missing resource as in the example above), nor be decomposed as a sequence of fair transitions, as exemplified below.
\begin{exa}
	Consider the following policies for Alice, Bob and Carl:
	\begin{align*}
	\Pol_{Alice} &= \{ Alice \xmapsto{\res} Bob \triangleleft \{ Bob \xmapsto{\res'} Alice \} \}\\
	\Pol_{Bob} &= \{ Bob \xmapsto{\res} Charlie \triangleleft \{ Charlie \xmapsto{\res'} Bob \} \}\\
	\Pol_{Carl} &= \{ Carl \xmapsto{\res'} Bob \triangleleft \{ Bob \xmapsto{\res} Carl \} \}
	\end{align*}
	Assume Alice has a $\res$, Carl has a $\res'$ and Bob has nothing.
	Consider the following sequence of events:
	Bob asks Alice $\res$, promising to give $\res'$ in return at some point;
	Alice agrees; Bob exchanges the received resource with Carl, obtaining $\res'$ 
	and keeping his promise by giving it to Alice.
	The computation is 
	\[
	\MuACstate \xrightarrow{Alice\ \xmapsto{\res}\ Bob} \MuACstate' \xrightarrow{Bob\ \xmapsto{\res}\ Carl, Carl\ \xmapsto{\res'}\ Bob} \MuACstate'' \xrightarrow{Bob\ \xmapsto{\res'}\ Alice} \MuACstate'''
	\]	
	Note that each request is eventually satisfied and each promise is kept. 
	Indeed, the exchange $\exc = \{Alice\ \xmapsto{\res}\ Bob, Bob\ \xmapsto{\res}\ Carl, Carl\ \xmapsto{\res'}\ Bob, Bob\ \xmapsto{\res'}\ Alice \}$ is fair.
	Nevertheless, the computation is not fair, and $\MuACstate \xrightarrow{\exc} \MuACstate'''$ is not a legal transition of the exchange environment (since Bob has no $\res'$).
\end{exa}
However, we would like also to have this kind of 
computations that traverse
 non fair configurations, but are sanitised afterwards.
As a matter of fact this is acceptable, provided that the computation is done atomically and under the control of a TTP, as implemented by the smart contract outlined in the next section.
For that, we call \emph{eventually fair} a computation that at the end results in an exchange beneficial to all the participants.
We first define when the global outcome is a many-step computation.
\begin{defi}
We call a computation 
$\MuACstate_0 \xrightarrow{\exc_1} \MuACstate_1 \xrightarrow{\exc_2} \dots \xrightarrow{\exc_n} \MuACstate_n$
\emph{eventually fair} whenever $\biguplus_{i=1}^n \exc_i$ is fair.
\end{defi}

Again logic comes to our rescue.
	The rule ($*$-cut) in \figurename~\ref{fig:starcut} enables us to verify whether the result of a computation as a whole respects the policies
	at hand, even though some of its steps are not fair.
	The correspondence between eventual fairness and \MuACL\ is stated by the following corollary of~\autoref{th:fair-exchange}.

\begin{figure}
	\begin{gather*}
	\prftree[r]
	{($*$-cut)}
	{\Omega;  \Theta, \Delta, \Sigma \vdash \sigma'}
	{\Omega';  \Theta', \Delta', \Sigma', \sigma' \vdash \sigma}
	{\Omega, \Omega';  \Theta, \Theta', \Delta, \Delta', \Sigma, \Sigma' \vdash \sigma}	
	\end{gather*}
	\caption{Linear cut rule for \MuACL.}
	\label{fig:starcut}
\end{figure}

\begin{restatable}[Validity = Eventual fairness of computations]{cor}{Logiccorrectnessstar}\label{thm:correctcompletestar}
Under the same conditions~of \autoref{th:fair-exchange}, the computation $\MuACstate \rightarrow^* \MuACstate'$ is eventually fair if and only if $\ \biguplus_{\usr \in \Usr} \semden{R_\usr}, \semden{C}; \semden{\MuACstate} \vdash \semden{\MuACstate'}$ is valid in \MuACL\ augmented with the cut rule $(*\text{-cut})$.
\end{restatable}

Decidability of \MuACL\ is not affected by the ($*$-cut) rule.
\begin{restatable}[\MuACL\ decidability]{cor}{MuACLsdec}\label{coro:MuACLdecide}
	An always-terminating algorithm exists that decides if an initial sequent is valid in \MuACL\ augmented with the cut rule ($*$-cut).
\end{restatable}

Moreover, eventually fair computations, and therefore the fair exchanges, can be effectively computed, given a context $C$, the MuAC policies and the current state.
This result also means that, given the current state $\MuACstate$ and a set of resources $\res_1, \dots \res_n$ that a user $\usr$ requires, there is an algorithm that terminates always and finds an eventually fair computation, if any, granting $\usr$ all the resources $\res_1, \dots \res_n$.

\begin{restatable}{cor}{computation}\label{thm:fairst}
There exists an always-terminating algorithm that, given the MuAC rulesets $\{ R_\usr \}$, the context $C$, the current state $\MuACstate$, a user $\usr$, and a set of resources 
$\{\res_1, \dots, \res_n \}$ returns an eventually fair computation, if any,
from $\MuACstate$ to some $\MuACstate'$ such that  for $1 \leq i \leq n$, $\MuACstate'(\usr)(\res_i) \geq 1$.
\end{restatable}

\section{MuAC as a Smart Contract}\label{sec:smart}

In this section, we show MuAC at work on an exchange environment supporting the exchanges of Non Fungible Tokens (NFTs for short), a common crypto-asset available on blockchain platforms, e.g., in Ethereum~\cite{tokenswap}.
Our exchange environment is rendered as a smart contract that stores the association between users and resources, as usual for wallet smart contracts.
A user interacts with the exchange environment by calling standard methods.
Moreover, we propose an off-chain application for supporting users to manage their requests.
The application and the smart contract rely on \MuACL\ for certifying and validating the fairness of the proposed exchange.
More in detail, the off-chain application produces \MuACL\ proofs, whereas the smart contract checks their validity.
Note that we delegate the client to perform the expensive part of the calculation.
Also, the smart contract plays the role of the TTP for the exchange environment, because the blockchain guarantees that the contract code is public and cannot be changed.
Actually, we rely on the integrity property of the blockchain to publicly maintain the ownership of resources and the computing capability of the smart contract to check the acceptability of the exchanges.

Below, we briefly discuss some assumptions on the blockchain smart contracts we consider;
we give the workflow for performing an exchange; we present the pseudocode for the NFT exchange environment; 
finally, we discuss its security against the typical 
attacks 
that may occur when exchanging goods.
As anticipated, in \autoref{subsec:fair-computations}, there is an algorithm, dubbed below \texttt{fair\_st}, that provides a user with a fair exchange granting her the required resources, possibly with a many-step, eventually fair computation.
In other words, in our implementation schema we can implement an entire computation as a single transaction.
This offers a further advantage because, at the price of the little extension to \MuACL\ with the ($*$-cut) rule, fairness can be checked on the resulting final exchange 
instead of on each individual transition.

\subsection{Assumptions on the blockchain platform}

In our implementation schema we assume to target a blockchain platform meeting the following conditions. 
There are three kinds of addresses: \emph{user accounts}, \emph{smart contracts} and NFTs.
An NFT is associated with an owner, which may be a user account or a smart contract.
A smart contract has a set of fields, namely its internal state, and exposes a set of functions that users or other contracts can call.
Users and smart contracts interact through messages that cause function invocations and NFT transfers.
Every message includes fields storing the \emph{sender} and \emph{destination} addresses (of users or smart contracts).
Optionally, the message may contain a \emph{function} field with the name of the function to call, a field for the actual \emph{parameters}, and a \emph{token} field containing the NFT.
If the receiver of a message is a smart contract and the function field is not empty, the called function  is executed. 
The execution of a function may change the internal state of the contract and may trigger the contract to send messages in turn.
If the message contains a NFT, the receiver becomes its owner.
We assume the NFTs of a contract to be accessible in its implicit \emph{tokens} field.
Note that a message can invoke a function while transferring a NTF, as shown in the \texttt{add\_resource} function displayed in \autoref{fig:MuACimpl}.

\subsection{User-Client-Smart Contract Interaction}

The workflow of the interaction between a user, the MuAC client and the smart contract are in \autoref{fig:col-arc},
and proceeds as follows:
\begin{enumerate}
\item the user asks the client to find an exchange granting her a list of desired tokens;

\item using the algorithm \texttt{fair\_st}, the client derives, if any, a next state $\MuACstate'$ of the smart contract and a \MuACL\ proof $\Pi$ certifying that $\MuACstate'$ is reachable with a fair exchange.
This computation is done off-chain;

\item if the user confirms that she accepts 
$\MuACstate'$
then a message is sent to the MuAC smart contract with 
$\MuACstate'$
and $\Pi$ attached, asking for the desired exchange to be enforced;

\item the smart contract receives the message from the client, checks the validity of $\Pi$ and updates the state, if 
this is the case.
\end{enumerate}
Note that verifying the \MuACL\ proof is linear on the number of \MuACL\ rules of $\Pi$, which depends on the exchanged resources but not on the participants (see \autoref{app:bcopt}).
Reducing the computational cost of 
verifying the proof 
is critical, because in typical blockchains like Ethereum every executed instruction is paid by the requester using an in-block platform-specific currency
(called \emph{gas} in Ethereum).
With the proposed workflow, all is performed off-chain except for a linear portion of the computation.
Nevertheless, the system guarantees transparency and correctness of the exchanges.
The rules for accessing the resources are in clear on the contract, whose execution is ensured by the blockchain.

\subsection{MuAC Client and Smart Contract}

\begin{figure*}[t]
	\centering
	\begin{subfigure}{\textwidth}
	\centering
		\begin{footnotesize}
			\begin{tikzpicture}
				\node[inner sep=0pt] (usr)  at (0,0) {\begin{tabular}{c}\includegraphics[scale=0.04]{images/user}
						\\[-0.02cm] $\usr$
				\end{tabular}};
				
				\node[inner sep=0pt] (cli)  at (5,0) {\begin{tabular}{c}\includegraphics[scale=0.04]{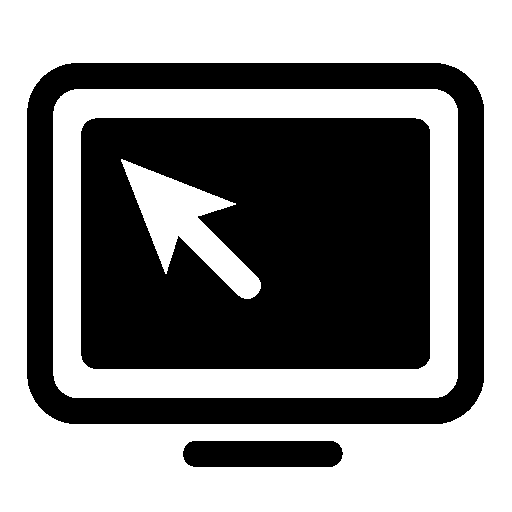}
						\\[-0.02cm] MuAC
						\\[-0.02cm] client
				\end{tabular}};

				\node[inner sep=0pt] (sc)  at (9,0) {\begin{tabular}{c}\includegraphics[scale=0.06]{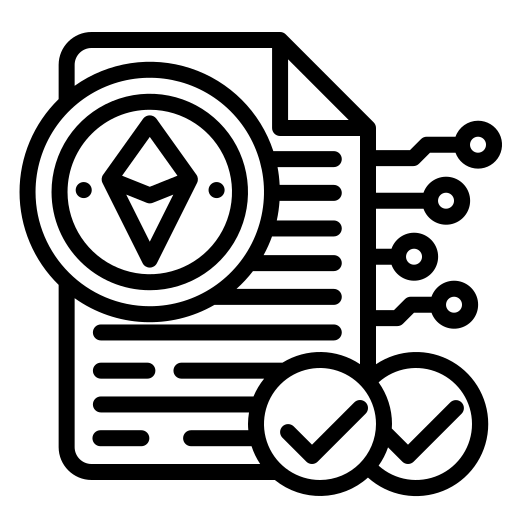}
						\\[-0.02cm] MuAC
						\\[-0.02cm] contract
				\end{tabular}};

				\node[draw, rectangle, minimum width=2.3cm,minimum height=2.25cm] (Eth)  at (9,0.25) {};
				\node[] (Ethname)  at (9,1.2) {Blockchain};
				
				{
					\draw[o->, thick, out=25, in=155] (usr) edge node[above]{(1) $\res_1, \dots \res_n$} (cli);
					\draw[<-o, thick, out=0, in=-180] (usr) edge node[above]{(2) $\MuACstate', \Pi ?$} (cli);
					\draw[o->, thick, out=-25, in=-155] (usr) edge node[above]{(3) yes, no} (cli);
					\draw[o->, thick] (cli) edge node[above]{
						(4) $(\MuACstate', \Pi)$
					} (sc);
				}
			\end{tikzpicture}
		\end{footnotesize}
		\caption{A schema for implementing MuAC on a blockchain platform.}
		\label{fig:col-arc}
	\end{subfigure}

	\begin{subfigure}{0.49\linewidth}
		\centering
		\begin{scriptsize}
			\begin{align*}
				&\texttt{\textbf{Contract} MuAC}\\
				&\qquad \MuACstate : \texttt{user\_address} \times \Res \rightarrow \mathbb{N}\\
				&\qquad Rs : \texttt{user\_address} \rightarrow \texttt{MuAC\_policy}\\
				&\qquad C : \texttt{Non-linear theory}\\[0.1cm]
				&\qquad \texttt{\textbf{function} add\_resource()}\\
				&\qquad\qquad \MuACstate[\texttt{msg.sender}][\texttt{NFT\_to\_Res}(\texttt{msg.token})]\texttt{++}\\[0.1cm]
				&\qquad \texttt{\textbf{function} withdraw\_resource($\res$)}\\
				&\qquad\qquad \texttt{\textbf{Require}(}\MuACstate[\texttt{msg.sender}][\res] > 0\texttt{)}\\
				&\qquad\qquad \MuACstate[\texttt{msg.sender}][\res]\texttt{-\,\!-}\\
				&\qquad\qquad \texttt{token = GetRes(tokens, $\res$)}\\
				&\qquad\qquad \texttt{token.transfer(msg.sender)}\\[0.1cm]
				&\qquad \texttt{\textbf{function} exchange(}\Pi, \MuACstate'\texttt{)}\\[-0.2cm]
				&\qquad\qquad \texttt{\textbf{Require}(Verify(}
				\prftree[]{\Pi}{\biguplus_{\usr} \semden{R[\usr]}, \semden{C}; \MuACstate \vdash \semden{\MuACstate'}} 
				\texttt{))}\\
				&\qquad\qquad \MuACstate \gets \MuACstate'
			\end{align*}
		\end{scriptsize}
		\caption{MuAC contract pseudo-code.}
		\label{fig:MuACsc}
	\end{subfigure}
	\begin{subfigure}{0.49\linewidth}
		\centering
		\begin{scriptsize}
			\begin{align*}
				&\texttt{\textbf{function} serve\_request(}[\res_1, \dots, \res_n], \usr \texttt{)}\\
				&\qquad \MuACstate, C, Rs \gets \texttt{take\_from\_contract()}\\
				&\qquad (\Pi, \MuACstate') \gets \texttt{fair\_st($Rs$,$C$,$\MuACstate$,$\usr$,$\res_1,\dots,\res_n$)}\\
				&\qquad \texttt{\textbf{if}(}\Pi = \texttt{ null) \textbf{then}}\\
				&\qquad\qquad \texttt{print "Error: request denied"}\\
				&\qquad \texttt{\textbf{else}}\\
				&\qquad\qquad \texttt{propose(}\usr, (\Pi, \MuACstate')\texttt{)}\\
				&\qquad\qquad \texttt{\textbf{if} receive(}\usr\texttt{) = yes \textbf{then}}\\
				&\qquad\qquad\qquad \texttt{message} \gets \texttt{empty\_message}\\
				&\qquad\qquad\qquad \texttt{message.function} \gets \texttt{MuAC.evolve}\\
				&\qquad\qquad\qquad \texttt{message.parameters} \gets (\Pi, \MuACstate')\\
				&\qquad\qquad\qquad \texttt{BCsend(message)}
			\end{align*}
		\end{scriptsize}
		\caption{Algorithm of the MuAC client}
		\label{fig:MuACcli}
	\end{subfigure}
	\caption{Implementation of a MuAC system on a blockchain.}
	\label{fig:MuACimpl}
\end{figure*}

The pseudocode of the MuAC smart contract is in \autoref{fig:MuACsc}.
Its internal state consists of the following three fields: $\MuACstate$ is a table storing the assignment of resources to user, namely the state of the exchange environment; $\mathit{Rs}$ is a map associating to each user her MuAC rulesets; and $C$ is a data structure representing the context. 
When a user wants to share a given resource in the system, she transfers the NFTs representing it to the smart contract via the function \texttt{add\_resource}.
The execution of \texttt{add\_resource} assigns the ownership of the NFT to the contract, and updates $\MuACstate$ accordingly.
At any moment, users can withdraw 
some of their resources recorded in the current state
by calling the function \texttt{withdraw\_resource}.
If the resource is currently associated with the requester in $\MuACstate$, this function updates $\MuACstate$  by removing the resource, and sends the user a message carrying the token;
otherwise, the computation fails and the state $\MuACstate$ remains unchanged.
Finally, a user 
proposes exchanges by calling the function \texttt{exchange}
applied to the
new state $\MuACstate'$ for the contract and a \MuACL\ proof $\Pi$
witnessing the fairness of the exchange.
In defining this function, we use the auxiliary one \texttt{verify}, assuming that it uses the \MuACL\ rules of~\autoref{fig:MuACLfull} for checking if $\Pi$ is a valid proof.
If this is the case, then calling \texttt{exchange} causes the current state becomes the wanted $\MuACstate'$.

The pseudocode of the MuAC client is in \autoref{fig:MuACcli}.
Upon a request of resources from a user $\usr$,
the client recovers the MuAC polices, the context $C$ and the current resource assignment $\MuACstate$ from the smart contract.
Then, through the algorithm \texttt{fair\_st}, it finds
a new assignment $\MuACstate'$ satisfying the request and a proof $\ \Pi$ of its fairness, if any. 
If $\usr$ agrees with $\MuACstate'$, then a message is 
sent 
using
the library function \texttt{BCsend} to the blockchain through the user account.
The message has the MuAC contract address as destination and \texttt{exchange} as the function to call.

\begin{exa}
	Take~\autoref{ex:col-cir}, and let the current state of the exchange environment be
	\begin{align*}
		\MuACstate &= \{ (\mathit{Alice}, \{\spell\}), (\mathit{Bob}, \{\light\}), (\mathit{Carl}, \{\heavy, \heavy, \heavy, \heal, \heal\}) \}
	\end{align*}
	Assume Alice makes a request to the client for obtaining a $\heavy$ card.
	Using \texttt{fair\_st}, the client finds a fair exchange satisfying the request, e.g., the one of~\autoref{ex:cir-logic}, and proposes it to Alice.
	If she 	agrees with the proposed exchange, the proof in \autoref{fig:logicexcir} is sent to the smart contract that enforces the exchange by updating the state as
	in~\autoref{ex:col-cir}.
\end{exa}

\subsection{Preventing Attacks}\label{sec:threat}
	The notion of fair transition helps to design our implementation schema so that it resists typical attacks.
	Actually, in our model some users may be dishonest and deceive others for their own advantage.
	The kinds of attacks they can perform are essentially the following.
The attacker can:
\begin{itemize}
\item deceive a honest user into accepting a disadvantageous exchange (\emph{trickery attack});
\item rescind an agreed exchange (\emph{repudiation attack});

\item refuse to give what promised in spite she received something (\emph{infringement attack}).
\end{itemize}
These misbehaviours often depend on a misplaced trust of honest users.
In addition, trickery attacks occur when a honest user has a partial knowledge and misses crucial information, while lack of commitment makes repudiation and infringement attacks easier.
For example, double spending is a form of trickery attack and also of infringement: in the first case, a user can deceive another to pay for a resource already paid, or she can promise the same resource to two different users (this is forbidden by MuAC and by the contract); in the second case, it can pay two different resources sending the same NFT to different users (this is forbidden by the blockchain consensus mechanism).
	Note that in our implementation all unfair transitions are pruned away, because
	the user is required to produce a \MuACL\ proof as a witness of the validity of the exchange, which is then checked by the TTP.
	Actually, the blockchain smart contract \emph{is} the TTP in charge of managing the resources of the users and their transfer. 
(Recall that anyway a TTP is required to ensure fairness of exchange protocols~\cite{Pagnia99}.)
Below, we discuss in details that there are no attacks:
	\begin{description}
		\item[Trickery attacks fail]
		A trickery attack never occurs because it corresponds to an unfair transition.
		This is guaranteed by the existence of a \MuACL\ proof for each exchange, computed off-chain by the user on her own.
		
		\item[Repudiation attacks fail]
		No repudiation attacks occur because only the TTP manages the users' resources, and thus 
		no one
		can refuse to honour a fair agreement.
		
		\item[Infringement attacks fail]
		The TTP has full control over the exchanges, hence no infringement attacks occur.
	\end{description}
Absence of attacks relief the users from carefully inspecting all the consequences of a proposed exchange: the contract manages the resources and evaluates exchange proposals on its own based on the policies.
Note that grieving attacks~\cite{Eckey}, where the attacker tricks
the honest party to pay fees without concluding the exchange, are not convenient for the attacker in this case, because she would have to pay for getting the certificate of the fair transition, and because the smart contract is in charge of actually transferring the resources.

Others security aspects depend on the actual implementation of the chosen blockchain platform.
Since we only present an implementation schema, we leave to developers the burden of taking care of these aspects.

\section{Discussion}\label{sec:discuss}

In this section, we detail some assumptions on which our formal model and MuAC rely and we discuss some limitations of our proposal.

A first assumption is that the context representing users' properties 
is not modified during an exchange.
We believe that this assumption does not hinder the generality of our proposal 
especially because an exchange should be checked for fairness and applied atomically and because the state of the exchange environment should not change during these phases, at least in those parts affected by the exchange.
As a matter of fact, lack of atomicity could jeopardise the fairness of an exchange.
This happens in the house exchange example if Carl gives up his friendship with Bob as soon as he obtains the permission to use Alice's house.
If not granted, atomicity can anyway be enforced by a transaction mechanism that reverts an exchange when its initial conditions cease to hold.

Similar assumptions hold for policies too: we assume users not to change their policies while an exchange is scrutinised and takes place.
Otherwise an extension is in order, e.g., based on transactions, to deal with such forms of volatile policies.

In our model, the policies are assumed public and available to all the members of an exchange environment.
This improves the accountability of a system because policies provide users with a public motivation for each accepted and rejected exchange.

In our proposal, we reduce the problem of verifying the fairness of an exchange to checking the validity of a \MuACL\ proof. 
This check is linear with the proof size.
Given a specific context, the proof size in turn only depends linearly on the number of atomic predicates in the MuAC rules used for the exchange and on the number of exchanged resources (cf.~Section~\ref{sec:fairnesstovalidity}).
Although interesting per se, the study of the properties of the logic and of its decision procedure is outside the scope of the present paper. 
A mitigation of the complexity of proving fairness and of constructing eventually fair computations is to reduce the number of 
involved
policies, e.g., by excluding some users' policies.

Our formalisation is essential and does not 
consider 
time related aspects, like 
expiring resources or offers/requests with a given lifetime.
Back to the home exchange example of~Section~\ref{sec:intro}, Alice may wish to spend two weeks in Rome in June, but Bob can only stay one week in Paris.
The description of the exchange policies, and the definition of the agreements grow richer with such additional information.
Also the difficulty of proving exchange fairness increases accordingly.
However, the overall shape of the exchange environment and the design of the mechanisms for protecting users will not be significantly affected by adding these additional time-dependent aspects.
For home exchanges resources can represent home staying for a given period of the year.
Note that for offers/requests with a given lifetime it is sufficient to (possibly automatically) update the users' policies, which is always safe provided that exchanges are atomic (as in our proposed blockchain implementation).
Another solution would require one to extend the context with information about the time of the requests.

Here, we only focus on token-based resources, and we give no direct mechanism for exchanging a given amount of them, like currency.
For example, the policy that allows one to exchange bitcoins for ethers must be manually encoded by the designer
(see, e.g., \autoref{ex:fair-nonlocal}, where there are two copies of the transfer of $\heal$ from Carl to Bob).

We also do not consider ``contractually conditional'' contracts that require propositions with nested $\linearcontract$.  
Such contracts may express agreements like ``if you trade $\res$ for $\res'$, then I will trade $\res''$ for $\res'''$.''
One can see them as constraints on participants' behaviour, while only exchanging digital resources seems not to require nested contractual implications.

So far, we addressed resources that change owner, but exchange platforms also permit users to share resources, e.g., photographs, without changing their owner.
Hosting this modality is plain: just tag such resources and treat them as if they come in infinitely many copies.
This issue has been addressed in~\cite{itasec20} and we will discuss it in~Section~\ref{sec:col-related}.

\section{Related Work}\label{sec:col-related}

The problem of fairly exchanging electronic assets over a network has been studied since the 80's by different communities.  
In the cryptography community, the focus was on designing protocols that allow several participants to exchange their
assets in such a way that no entity gives away
their own %
resource without also getting the other expected resource. 
In the access control community, the focus was on designing policy languages that allow participants to express the conditions under which an exchange is acceptable and what they expect in return.
Also, mutuality plays a main role in trust negotiation, which permits two parties who do not trust each other to interact.
Finally, linear logic has been used for
modelling resource-aware games and problems in the artificial intelligence community, more precisely in the area of Multi-agent Systems.
Below, we briefly survey these approaches, and some related logic.

\paragraph*{Fair exchange protocols}
The pioneering work by Even and Yacobi~\cite{Even} studied contract-signing protocols, a particular case of fair exchange, and showed that no deterministic protocol exists without a TTP.
Other proposals focused on two party protocols and tried to weaken the need of using a TTP by considering randomised protocols~\cite{Freiling} or the so-called optimistic approach where the TTP intervenes only when a problem arises, e.g., in case of a dispute or crash~\cite{Asokan}.
There are also proposals that address multi-party fair exchanges~\cite{Franklin98,Bao99} where a
group of mutually suspicious parties are involved.
To ensure the fairness of the exchanges a TTP is required also in these protocols. 

More recently, with the growth of blockchain platforms several proposals have been put forward where the TTP is implemented as a smart contract.
Dziembowski et al.~\cite{Dziembowski} proposed FairSwap, a fair exchange protocol that minimises the cost of running the contract and avoids expensive cryptographic primitives. 
The underlying idea is that the initial step of the two parties A and B consists in deploing on the network a smart contract: $A$ deposits the whole price in cryptocurrency and the underlying consensus mechanism of the blockchain guarantees that either $A$ receives the goods and $B$ the money, or $A$ gets her deposit back after the timeout has passed.

Eckey et al.~\cite{Eckey} proposed OptiSwap, which extends FairSwap 
by incorporating an interactive dispute resolution sub-protocol.
It improves the efficiency of the protocol when run by two honest parties and it protects against \emph{grieving attacks}.

Our proposal differs from the above in two main points. 
First, these papers often consider two parties only, while we have no bound on the number of participants.
Second, we focus on the linguistic mechanisms that participants use to 
express the conditions 
when an exchange is acceptable, while these papers only focus on the interactions between the parties for performing an exchange defined previously.

\paragraph*{Access control}
We only consider discretionary access control~\cite{CompSecPrinPra} because it is a natural choice in distributed cooperative settings, where users individually decide the policies for their own resources.
In this context, a main issue is combining individual policies.
To the best of our knowledge, 
existing proposals do not address mutuality, 
but only focus on the resolution of conflicts~\cite{policy-composition, collac, surveyCCCS}.

In the widespread world of social networks, mutuality plays a prominent role, but it is scarcely regulated.
A remarkable exception is~\cite{SACMAT19}, which allows for the  definition of mutual access control policies.
This is done by introducing a new grant, called \textit{mutual}, in addition to the usual \textit{accept} and \textit{deny}.
Suppose that an access request from user $A$ to resource $r$ of $B$ evaluates to \textit{mutual}.
Intuitively, the request is served if and only if a request from $B$ for a \emph{similar} resource $r'$ of $A$ will evaluate to \textit{accept} or \textit{mutual}.
Similarity is fixed once and for all, and it is not user-defined.
A first difference of our proposal is that we allow users to explicitly state what they require in return for the resource they give.
In addition, mutuality in MuAC may involve many users, as in \autoref{ex:col-cir}, and we target consumable resources.

Some of us proposed a logically-based policy language to state conditions about what a user receives in return for allowing an access~\cite{itasec20}. 
Differently from this paper, in~\cite{itasec20} the authors rely only on non-linear logic and focus on data sharing, as in social networks. 
Moreover, they 
neither
propose a formal semantics nor an implementation.

\paragraph*{Trust negotiation}
K\`olar et al.~\cite{KolarGL18} propose a multi-round protocol where the parties exchange pieces of private information (\emph{credentials}) so as to increase their mutual trust.
Each party defines an individual policy specifying the conditions that the other party must satisfy to obtain credentials.
The overall goal is to balance the disclosure of information and the mutual benefit
gained by each party.
Logical languages for specifying trust policies have been proposed, e.g.,
Cassandra~\cite{BeckerS04} and SecPal4P~\cite{secpal4p}.
The main difference with respect to MuAC is that these proposals use
classical logic, and thus circular conditions do not lead to an agreement.

\paragraph*{Logical Modelling of Resource Games in Artificial Intelligence}

Linear logic has been used to model resource aware reasoning in various AI contexts, in particular for Multi-agent Systems.
They all describe the desire of agents in terms of their goals or value functions, and derive or recognise reasonable offers and strategies.
A contribution of ours is instead a way of directly modelling what users offer via an exchange policy language, hence offering a descriptive approach rather than a prescriptive one.
Value function-based policies and our explicit exchange-based policies are introduced in~\cite{Ecai2024} and are discussed below.

In~\cite{Harland02}, Harland et al.~show how linear logic enables reasoning about negotiations, encoding agents' goals and what they offer.
Linear logic proofs recognise the negotiation outcomes that satisfy all parties.

In~\cite{Kungas2003,Kungas2004}, Küngas et al.~propose a model of cooperative problem
solving, and use linear logic for encoding agents’ states, goals and capabilities. Then, each agent determines whether it can solve the problem in isolation. If it cannot, then it starts negotiating with other agents in order to find a cooperative solution.
Partial deduction~\cite{partialdeduction} is used to derive possible deals. 
In~\cite{Kungas2006,Kungas2008}, the authors extend their work by considering coalition formation.

In~\cite{Porello2010}, Porello et al.~target distributed resource allocation.
They encode resource ownership and transfers, as well as value functions representing user preferences in (various fragments of) linear and affine logic.
They show how logic proofs discriminate mutually satisfactory exchanges that increase the value of the assignment for every user, thus recovering a notion of social welfare (in terms of Pareto optimality).
They do not model offers and negotiation,  
because the users value functions used to decide upon exchanges are assumed as known.
They prove that 
every
 sequence of
individually rational deals will always converge to an allocation with
maximal social welfare, as known from~\cite{Sandholm2002}.
In contrast, we directly encode the user exchange policies instead of their value function, and we investigate agreements and reachable resource associations.
Moreover, we extend (a fragment of) linear logic with a contractual implication and we recover decidability results.

In~\cite{Troquard2018}, Troquard models the interaction of resource-conscious agents who share resources to achieve their goals in cooperative games.
Algorithms are proposed for deciding whether a group of agents can form a coalition and act together in a way that
satisfies them all. 
The complexity classes of various related problems for various fragments of linear and affine logic are discussed.
Our focus is instead on resource exchanges, and our context is a mixture of cooperative and competitive behaviour.
In the subsequent work~\cite{Troquard2020}, Troquard studies how a central authority can redistribute the resources in order to modify the set of Nash equilibria of cooperative games based on resource sharing.  
The complexity of this problem is discussed in terms of the chosen (fragment of) resource-sensitive logic.

This paper is closely related to~\cite{Ecai2024}, in which
we tailor exchange environments to Multiagent Systems and adapt the notion of agreement and fair exchange to also take care of values assigned to resources by users' \emph{valuation functions}.
In this way, users accept exchanges that increase the value of the resources they posses and that respect also their policies, which are similar to the ones we use here.
The first difference between \MuACL\ and the Contractual Exchange Logic (CEL) of~\cite{Ecai2024}  
is that \MuACL\ also includes a non-linear fragment, combined with the linear one in the style of LNL~\cite{LNL}.
Secondly, the rules for the linear contractual implication differ: premises in the implications are treated linearly in CEL, while they are affine in \MuACL. 
Intuitively, this reflects the fact that in~\cite{Ecai2024} users must explicitly declare when they accept a resource for free, whereas the policies used here implicitly state that users are always willing to receive such a gift.
Finally, further differences are that here we propose MuAC, a logical language for expressing exchange policies, that we give a compilation procedure targeting \MuACL, and that we show how our machinery can be employed to exchange crypto-assets in a blockchain smart contract scenario.
Differently from~\cite{Ecai2024}, here we do not consider valuation functions when deciding fairness of exchanges, because these functions often require
an agent to know the resources of the other users to evaluate a proposed agreement, which is not always the case in concrete contexts.

\paragraph*{Logic and Types}
We formalised the contractual aspects following the pioneering PCL proposed by Bartoletti and Zunino in~\cite{BZ}, which is a logic for modelling 
the peculiar circular reasoning of contracts.
Our operator $\linearcontract$ is actually a linear version of their $\contract$.
The main difference with respect to PCL is that from $p \contract p', p' \contract p$ one can derive $p$, $p'$, and $p \land p'$, but in our system only the whole pair $p \otimes p'$ can be derived from  $p \linearcontract p', p' \linearcontract p$.
This is critical when dealing with consumable resources, as it guarantees that all the users get what is promised by the agreement.
In addition, our logic mixes linear and non-linear terms
by following the approach of~\cite{LNL}.
Our sequents are inspired by~\cite{Kanovich94}, where a computational fragment of the linear logic is proposed for reasoning about computations with consumable resources.

MuACL extends linear logic with an ad hoc operator.
On a similar account, various generalizations of linear logic and linear type systems have been proposed in the past few years, including~\cite{parmod1,parmod2,parmod3,parmod4,parmod5,parmod6,parmod7,parmod8,parmod9,parmod10}.
Their generality derives from a parametrization over some abstract algebraic structure of grades, extending the standard exponential $!$ modality.
The parameters are usually taken from a semi-ring, the choice of which allows expressing linear, affine or relevant usage of resources, as well as other kinds of ``weighted" usages (integer, non-negative real numbers etc.).
In addition, these generalizations allow expressing effects and coeffects,
i.e., changes that typed programs cause
on the environment, as well as 
requirements over the environment's properties (availability of a resource, security levels, etc.).
These solutions are not immediately applicable to our context, since none of the proposed instances deal with our kind of circular reasoning.
Specularly, we do not consider these generalizations in our work, which could benefit the applicability and impact of our proposal.
A closer investigation on the matter is postponed to future work.

Our work shares some intuitions with session types, which describe desirable properties of protocols~\cite{stypes,stypes2,mstypes}.
\emph{Multiparty Session Types (MPSTs)} work on systems with multiple participants as MuAC does, however they differ for a number of reasons:
(1)~MPSTs deal with message communications, whereas MuAC with resources and their exchange;
(2)~in MPSTs the objective is a protocol specification, expressed  
as a \emph{global type}, roughly a collection of possible sequences of communications, whereas we consider MuAC policies, locally defined by users in isolation and more abstract than a list of possible exchanges;
(3)~in MPSTs, the global type is usually \emph{projected} into \emph{local types} to describe how the behaviour of single participants can enforce the global specification, whereas we start from local policies, and derive the global behaviour as those exchanges obtained by the MuACL proof system;
(4)~in MPSTs correctness guarantees that channels' arity, types, polarity and multiplicity are respected, as well as the property of \emph{deadlock freedom} (which is violated, e.g., by a circular wait),
whereas MuAC abstracts away from these problems and focusses on the reciprocity and fairness of exchanges;
(5)~MPSTs are commonly associated with 
process calculi to explicitly express participants' behaviour,
whereas we only consider policies and abstract away from actual implementation.
As future work, we plan to rephrase MuACL in order to describe MuAC policies as local session types, from which to extract global types that express all and only legal exchanges.
This extraction could exploit a procedure similar to the one for deciding MuACL sequents.

An extension of the Curry–Howard correspondence has been proposed relating linear logic and session types:
proofs correspond to processes, propositions to types, and proof normalization to communications \cite{CHlin, CHlin2}.
In this setting, duality of linear logic corresponds to compatibility of send and receive instructions in binary session types.
In MPSTs, duality is generalized by the notion of \emph{coherence} proposed in~\cite{mstypes2}, which shares some similarity with our notion of agreement.
For example, the dual of sending is receiving and vice-versa, hence this pair of operations is coherent.
Similarly, we could define the dual of the promise of ``giving $\delta$ if $\delta'$ is given in return'' ($\delta \linearcontract \delta'$) as ``giving $\delta'$ if $\delta$ is given in return'' ($\delta' \linearcontract \delta$), which indeed leads to an agreement.
However, coherence corresponds to the correct compatibility of communication instructions, whereas agreements require some 
circularity between preconditions and results of promises.
Indeed, communications in MPSTs correspond to applications of some generalized 
version of the cut rule based on coherence, which is an admissible rule in 
standard linear logic. 
Thus, the set of derivable sequents are not affected by their new notion
of coherence, while our agreements can neither be expressed in linear logic directly, 
nor they can be encoded via some homomorphic map.
Another major difference is 
that send and receive instructions are atomic, whereas promises may speak of collections of exchanges, and a single promise may be the dual of more promises, as, e.g.
$\delta \otimes \delta'' \linearcontract \delta' \otimes \delta'''$ agrees with 
$(\delta \linearcontract \delta') \otimes (\delta'' \linearcontract \delta''')$.
In addition, circularity in the order of send and receive leads to a deadlock, whereas circularity is often required to reach an agreement in MuAC.
The existence of a Curry–Howard correspondence for MuACL remains an open question.

\section{Conclusions and future work}\label{sec:col-conclude}

We considered digital platforms through which users exchange resources in accordance to exchange policies that express what users are willing to give and what they want in return.
We formalised these environments as labeled transition systems, where states record the ownership of the resources and transitions represent title transfers.
We mainly focussed on exchange policies formalising them and characterising as fair those that obey all the policies of the users involved and that avoids double spending. 
We provided users with a Datalog-like language that supports an easy definition of exchange policies and that is equipped with a formal semantics.

A crucial issue is ensuring that resource exchanges never violate the policies in force, so that malicious users cannot take advantage of honest ones. 
To do that, we resorted to logic and defined 
\MuACL, which combines classical non-linear and linear aspects with a 
novel
contractual operator, not expressible with the 
standard
operators. 
Since MuAC is compiled in \MuACL, determining 
the fairness of an exchange amounts to finding a proof in \MuACL, which is decidable.

Finally, we show our proposal effective, by providing a schema for implementing a blockchain smart contract for exchanging NFTs, which records the assignment of NFTs to users and is in charge of managing them.
The main characterising feature of our implementation schema is that users compute in isolation a \MuACL\ proof
witnessing that the desired exchange is fair, 
and propose it to the contract, which efficiently verifies its validity.
If the proof is valid, then the exchange takes place, otherwise it is denied. 
We show then that our implementation is robust against our threat model.

As future work, first we plan to more accurately determine the complexity of the logic \MuACL\ and to define an efficient decision procedure.
Then, we will enrich the MuAC policy language.
For the time being, MuAC has positive \texttt{Gives} grants only.
Having negative rules will significantly extend the language's expressivity, but requires resolving potential conflicts.
Another extension is allowing rules in which a user should \emph{not} perform some exchanges to obtain a resource.
For example, Alice gives an apple to Bob if Bob gives nothing to Carl.
This kind of negative requirement appears necessary in policies regulating conflicts of interest.
Also, we adopted so far a ``default deny'' approach for the evaluation of user 
policies, while ``default allow'' seems sometimes useful.
Furthermore, we plan to enhance the expressivity of exchange environments by attaching a value to resources, along the lines of our preliminary investigation of~\cite{Ecai2024}, discussed in Section~\ref{sec:col-related}.

Finally, we would like to study and define a high-level language for defining exchange platforms, embedding MuAC.
A suitable compilation will then be needed in order to map such a language in an exchange environment.

\section*{Acknowledgment}
This work was partially supported by project SERICS (PE00000014) PNRR MUR - M4C2 - I 1.3 and by project PRIN PNRR AM$\forall$DEUS (P2022EPPHM) M4C2 I 1.1 - D53D23017420001 under the MUR National Recovery and Resilience Plan funded by the European Union - NextGenerationEU.
Part of this work has been carried out while Luca Viganò was visiting professor at the IMT School for Advanced Studies Lucca. 

\bibliographystyle{alphaurl}
\bibliography{references}

\appendix

\section{Notations and Symbols}
{\footnotesize \begin{center}
\begin{tabular}{ l l l }
	\toprule
	\hspace{1cm} & \textbf{Notation} & \textbf{Description}\\
	\midrule
	\textbf{Exchange Environment}\quad\quad &$\mathit{Res} \ni \mathit{res}$ & Resources\\
	&$\mathit{Usr} \ni \mathit{usr}$ & Users\\
	&$\mathit{St} \ni \mathit{st}$ & States\\
	&$\mathit{Tr} \ni \mathit{tr}$ & Transfers\\
	&$\mathit{exc} \ni \mathit{Exc}$ & Exchanges\\
	&$\mathit{Pol} \ni \mathit{pol}$ & Exchange policies\\
	\midrule
	\textbf{MuAC} &$U \ni u$ & User variables\\
	&$\texttt{Me}$ & Variable representing the owner of the policy\\
	&$r \in R$ & MuAC rule in a ruleset\\
	&$C$ & Context\\
	&$R \ni r$ & MuAC ruleset and contained rule\\
	\midrule
	\textbf{MuAC Logic} &$res@usr$ & Atomic linear proposition\\
	&$\Omega \ni \omega$ & Multiset of non-linear propositions\\
	&$\Theta \ni \theta$ & Multiset of linear propositions using $\linearcontract$\\
	& $\Delta \ni \delta$ & Multiset of linear propositions using $\multimap$\\
	& $\Sigma \ni \sigma$ & Multiset of conjunctions atomic linear propositions\\
	\bottomrule	
\end{tabular}
\end{center}
}

\section{Decidability results}\label{app:proofs}

In the following we write \MuACLs\ for \MuACL\ augmented with the cut rule $(*\text{-cut})$.
Moreover, we extend the set of rules $Lr$ of~\autoref{not:rulesets} with ($*$-cut).

To show that proofs can be normalised, we introduce a new logic, named \MuACLst,
derived from \MuACLs by substituting ($\otimes$-right) with the following rule
{\normalsize \begin{gather*}
	\prftree[r]
	{($\otimes$-right')}
	{\Omega; \Theta, \Delta, \Sigma \vdash \sigma\qquad}
	{\Omega'; \Theta', \Delta', \Sigma' \vdash \sigma'}
	{\Omega, \Omega'; \Theta,  \Theta', \Delta, \Delta', \Sigma, \Sigma'\vdash \sigma \otimes \sigma'}
\end{gather*}}

We show that every
 \MuACLs\ proof can be transformed into one in \MuACLst\  for the same sequent.
Then, we reorder the \MuACLst\  proof, 
and we transform it into an equivalent \MuACLs\ where a final reordering takes place, thus obtaining a normal proof.

Recall that we use a double line to represent multiple applications of the same rule.
\begin{lem}\label{MuACLsto2}
	If a \MuACLs\ proof exists for a sequent, then there exists an equivalent one in \MuACLst\  that uses ($*$-cut) only if the original one does.
\end{lem}
\begin{proof}
	Follows from ($\otimes$-right') being derivable in \MuACLs\ without using ($*$-cut).
	Every occurrence of ($\otimes$-right') can be substituted with the following derivation.
	{\normalsize 
	\[
		\prftree[r]
		{($\otimes$-right)}
		{
			\prftree[r, double]
			{(L-Weak)}
			{\Omega;  \Theta, \Delta, \Sigma \vdash \sigma}
			{\Omega, \Omega';  \Theta, \Delta, \Sigma \vdash \sigma}
		}
		{
			\prftree[r, double]
			{(L-Weak)}
			{\Omega';  \Theta', \Delta', \Sigma' \vdash \sigma'}
			{\Omega, \Omega';  \Theta', \Delta', \Sigma' \vdash \sigma'}
		}
		{\Omega, \Omega'; \Theta,  \Theta', \Delta, \Delta', \Sigma, \Sigma' \vdash \sigma \otimes \sigma'}\qedhere
	\]}
\end{proof}

We prove the following auxiliary lemmata about reordering rules in \MuACLst, where $Lr'$ is the set of \MuACLst\  rules defined as $(Lr \setminus \{\text{($\otimes$-right)}\}) \cup \{\text{($\otimes$-right')}\}$.

\begin{lem}\label{thm:LP-GCS}
	Any \MuACLst\  derivation where $r \in Sr \cup Cr \cup Gr$ is applied before $r' \in Lr' \cup Gr \cup Pr$ can be rewritten as an equivalent derivation where all the rules in $Sr \cup Cr \cup Gr$ are applied after the rules in $Lr' \cup Gr \cup Pr$.
	In addition, the equivalent derivation uses ($*$-cut) only if the original one does.
\end{lem}
\begin{proof}
	The only non trivial cases are $r$ = (L-$\rightarrow$-left) or ($\Omega$-Cut), and $r'$ = ($*$-cut) or ($\otimes$-right').
	Take $r$ = (L-$\rightarrow$-left) and $r'$ = ($\otimes$-right'), and let $r$ be applied 
	in the left premise.
	\[
{\normalsize 	\prftree[r]
	{($\otimes$-right')}
	{
		\prftree[r]
		{(L-$\rightarrow$-left)}
		{\Omega \Vdash \omega}
		{
				\Omega', \omega';
				\Theta, \Delta, \Sigma
			\vdash \sigma}
		{
				\Omega, \Omega', \omega \rightarrow \omega';
				\Theta, \Delta, \Sigma
			\vdash \sigma}
	}
	{
		\prftree[noline]
		{
				\Omega'';
				\Theta', \Delta', \Sigma'
			\vdash \sigma'}
	}
	{\Omega, \Omega', \Omega'', \omega \rightarrow \omega'; \Theta,  \Theta', \Delta, \Delta', \Sigma, \Sigma'\vdash \sigma \otimes \sigma'}
}	\]
	Then, swap the rules as follows:
	\[
{\normalsize 	\prftree[r]
	{(L-$\rightarrow$-left)}
	{
		\prftree[noline]
		{\Omega \Vdash \omega}
	}
	{
		\prftree[r]
		{($\otimes$-right')}
		{\Omega''; \Theta', \Delta', \Sigma' \vdash \sigma'}
		{\Omega', \omega'; \Theta, \Delta, \Sigma \vdash \sigma}
		{\Omega', \Omega'', \omega'; \Theta,  \Theta', \Delta, \Delta', \Sigma, \Sigma'\vdash \sigma \otimes \sigma'}
	}
	{\Omega, \Omega', \Omega'', \omega \rightarrow \omega'; \Theta,  \Theta', \Delta, \Delta', \Sigma, \Sigma'\vdash \sigma \otimes \sigma'}}
	\]
	Similarly if $r$ is ($\Omega$-Cut) or it is applied to the derivation of the right premise.
	\\
Take $r$ = (L-$\rightarrow$-left) and $r'$ = ($*$-cut), and let $r$ be applied to the  left premise derivation:
	\[
{\normalsize 	\prftree[r]
	{($*$-cut)}
	{
		\prftree[r]
		{(L-$\rightarrow$-left)}
		{\Omega \Vdash \omega}
		{
				\Omega', \omega';
				 \Theta, \Delta, \Sigma
			\vdash \sigma}
		{
				\Omega, \Omega', \omega \rightarrow \omega';
				 \Theta, \Delta, \Sigma
			\vdash \sigma}
	}
	{
		\prftree[noline]
		{
				\Omega'';
				 \Theta', \Delta', \Sigma', \sigma
			\vdash \sigma'}
	}
	{\Omega, \Omega', \Omega'', \omega \rightarrow \omega'; \Theta,  \Theta', \Delta, \Delta', \Sigma, \Sigma' \vdash \sigma'}}
	\]
	Then, swap the rules as follows:
	\[
	{\normalsize \prftree[r]
	{(L-$\rightarrow$-left)}
	{\Omega \Vdash \omega}
	{
		\prftree[r]
		{($*$-cut)}
		{\Omega', \omega'; \Theta, \Delta, \Sigma \vdash \sigma}
		{
			\prftree[noline]
			{\Omega''; \Theta', \Delta', \Sigma', \sigma \vdash \sigma'}
		}
		{\Omega', \Omega'', \omega'; \Theta, \Theta', \Delta, \Delta', \Sigma, \Sigma' \vdash \sigma'}
	}
	{\Omega, \Omega', \Omega'', \omega \rightarrow \omega'; \Theta,  \Theta', \Delta, \Delta', \Sigma, \Sigma' \vdash \sigma'}
}	\]
	Similarly if $r$ is ($\Omega$-Cut) or it is applied to the derivation of the right premise.
\end{proof}

\begin{lem}\label{thm:G-CS}
	Any \MuACLst\  derivation $\Pi$ where $r' \in Gr$ is applied immediately after $r \in Sr \cup Cr$ can be transformed in an equivalent derivation $\Pi'$
	where no rule in $Gr$ follows a rule in $Sr \cup Cr$.
	Also, the equivalent derivation uses ($*$-cut) only if the original one does.
\end{lem}
\begin{proof}
	Let $r$ and $r'$ be (L-$\rightarrow$-left) and (G-left-$\theta$), respectively, i.e., let $\Pi$ be
	\[
{\normalsize 	\prftree[r]
	{(G-left-$\theta$)}
	{
		\prftree[r]
		{(L-$\rightarrow$-left)}
		{\Omega \Vdash \omega}
		{\Omega', \omega'; \Theta, \theta, \Delta, \Sigma \vdash \sigma}
		{\Omega, \Omega', \omega \rightarrow \omega'; \Theta, \theta, \Delta, \Sigma, \vdash \sigma}
	}
	{\Omega, \Omega', \omega \rightarrow \omega', G(\theta); \Theta, \Delta, \Sigma, \vdash \sigma}}
	\]
	Then, $\Pi'$ is as follows:
	\[
	{\normalsize \prftree[r]
	{(L-$\rightarrow$-left)}
	{\Omega \Vdash \omega}
	{
		\prftree[r]
		{(G-left-$\theta$)}
		{\Omega', \omega'; \Theta, \theta, \Delta, \Sigma \vdash \sigma}
		{\Omega', \omega', G(\theta); \Theta, \Delta, \Sigma, \vdash \sigma}
	}
	{\Omega, \Omega', \omega \rightarrow \omega', G(\theta); \Theta, \Delta, \Sigma, \vdash \sigma}}
	\]
	Similarly for 
	every
	 $r \in Cr$ and $r' \in Gr$.
	
	Let $r$ and $r'$ be (L-Weak) and (G-left-$\theta$), respectively, i.e., let $\Pi$ be
	as below on the left. Then $\Pi'$ is on the right, and the proof for (G-left-$\delta$) is similar.	
		\begin{gather*}
		{\normalsize \prftree[r]
		{(G-left-$\theta$)}
		{
			\prftree[r]
			{(L-Weak)}
			{\Omega; \Theta, \theta, \Delta, \Sigma \vdash \sigma}
			{\Omega, \omega; \Theta, \theta, \Delta, \Sigma, \vdash \sigma}
		}
		{\Omega, \omega, G(\theta); \Theta, \Delta, \Sigma, \vdash \sigma}}
		\qquad
		{\normalsize \prftree[r]
		{(L-Weak)}
		{
			\prftree[r]
			{(G-left-$\theta$)}
			{\Omega; \Theta, \theta, \Delta, \Sigma \vdash \sigma}
			{\Omega, G(\theta); \Theta, \Delta, \Sigma, \vdash \sigma}
		}
		{\Omega, \omega, G(\theta); \Theta, \Delta, \Sigma, \vdash \sigma}}
		\end{gather*}
	
	Let $r$ and $r'$ be (L-Cont) and (G-left-$\theta$), respectively, i.e., let $\Pi$ be
	as below on the left. Then $\Pi'$ is on the right, and the proof for (G-left-$\delta$) is similar.		
	\begin{gather*}
	{\normalsize \prftree[r]
	{(G-left-$\theta$)}
	{
		\prftree[r]
		{(L-Cont)}
		{\Omega, \omega, \omega; \Theta, \theta, \Delta, \Sigma \vdash \sigma}
		{\Omega, \omega; \Theta, \theta, \Delta, \Sigma, \vdash \sigma}
	}
	{\Omega, \omega, G(\theta); \Theta, \Delta, \Sigma, \vdash \sigma}}
	\qquad
	{\normalsize \prftree[r]
	{(L-Cont)}
	{
		\prftree[r]
		{(G-left-$\theta$)}
		{\Omega, \omega, \omega; \Theta, \theta, \Delta, \Sigma \vdash \sigma}
		{\Omega, \omega, \omega, G(\theta); \Theta, \Delta, \Sigma, \vdash \sigma}
	}
	{\Omega, \omega, G(\theta); \Theta, \Delta, \Sigma, \vdash \sigma}
}\qedhere	
	\end{gather*}
\end{proof}

\begin{lem}\label{thm:P-L}
Any  \MuACLst\  derivation $\Pi$ where $r' \in Lr'$ is applied immediately after $r \in Pr$ can be transformed in an equivalent derivation $\Pi'$
	where no rule in $Lr'$ follows a rule in $Pr$.
\mbox{
Also, the
equivalent derivation does not use ($*$-cut) if the original derivation does not.
}
\end{lem}
\begin{proof}
	The only non trivial cases are when $r'$ = ($\otimes$-right') or ($*$-cut).
\\	
	Take $r'$ = ($\otimes$-right'), and let $r$ = ($\linearcontract$-split) be applied to derivation of the left premise
	\[
{\normalsize 	\prftree[r]
	{($\otimes$-right')}
	{
		\prftree[r]
		{($\linearcontract$-split)}
		{
				\Omega; \Theta,
				 \delta \otimes \delta'' \linearcontract \delta' \otimes \delta''',
				 \theta, \Delta, \Sigma 
			\vdash \sigma}
		{
				\Omega; \Theta,
				\delta \linearcontract \delta',
				\delta'' \linearcontract \delta''',
				\Delta, \Sigma 
			\vdash \sigma}
	}
	{
		\prftree[noline]
		{
				\Omega'; \Theta',
				\Delta', \Sigma' 
			\vdash \sigma'}
	}
	{
			\Omega, \Omega'; \Theta, \Theta', \delta \linearcontract \delta', \delta'' \linearcontract \delta''',
			\Delta, \Delta', \Sigma, \Sigma'
		\vdash \sigma \otimes \sigma'}}
	\]
	Then, swap the rules as follows:
	\[
	{\normalsize \prftree[r]
	{($\linearcontract$-split)}
	{
		\prftree[r]
		{($\otimes$-right')}
		{
				\Omega; \Theta, \delta \otimes \delta'' \linearcontract \delta' \otimes \delta''',
				\Delta, \Sigma 
			\vdash \sigma\qquad}
		{
			\prftree[noline]
			{
					\Omega'; \Theta',
					\Delta', \Sigma'
				\vdash \sigma'}
		}
		{
				 \Omega, \Omega'; \Theta, \Theta', \delta \otimes \delta'' \linearcontract \delta' \otimes \delta''',
				 \Delta, \Delta', \Sigma, \Sigma'
			\vdash \sigma \otimes \sigma'}
	}
	{
			\Omega, \Omega'; \Theta, \Theta', \delta \linearcontract \delta', \delta'' \linearcontract \delta''',
			\Delta, \Delta', \Sigma, \Sigma'
		\vdash \sigma \otimes \sigma'}
}	\]
	
	Similarly for ($\linearcontract$-left) and for ($\linearcontract$-split) applied to the derivation of the right premise.
\\	
	Take $r'$ = ($*$-cut), and let $r$ = ($\linearcontract$-left) be applied to the derivation of the left premise
	\[
	{\normalsize \prftree[r]
	{($*$-cut)}
	{
		\prftree[r]
		{($\linearcontract$-left)}
		{\delta \subseteq \delta'}
		{
				\Omega; \Theta, \Delta,
				\delta', \Sigma 
			\vdash \sigma}
		{
				\Omega; \Theta, 
				\delta \linearcontract \delta', \Delta, \Sigma 
			\vdash \sigma}
	}
	{
		\prftree[noline]
		{
				\Omega'; \Theta',
				\Delta', \Sigma', \sigma
		\vdash \sigma'}
	}
	{\Omega, \Omega'; \Theta,  \Theta', \delta \linearcontract \delta', \Delta, \Delta', \Sigma, \Sigma' \vdash \sigma'}
}	\]
	Then, swap the rules as follows:
	\[
	{\normalsize \prftree[r]
	{($\linearcontract$-left)}
	{\delta \subseteq \delta'}
	{
		\prftree[r]
		{($*$-cut)}
		{
				\Omega; \Theta, \Delta,
				\delta', \Sigma
			\vdash \sigma\qquad}
		{
			\prftree[noline]
			{
					\Omega'; \Theta',
					\Delta', \Sigma', \sigma 
				\vdash \sigma'}
		}
		{\Omega, \Omega'; \Theta,  \Theta', \Delta, \Delta', \delta', \Sigma, \Sigma' \vdash \sigma'}
	}
	{\Omega, \Omega'; \Theta,  \Theta', \delta \linearcontract \delta', \Delta, \Delta', \Sigma, \Sigma' \vdash \sigma'}
}	\]
	Similarly for ($\linearcontract$-split) and for ($\linearcontract$-left) applied to the derivation of the right premise.
\end{proof}

We define now normal proofs for \MuACLst.
\begin{defi}
A \MuACLst\  proof is \emph{normalised} if it can be decomposed in
\[
{\normalsize \begin{matrix}
	\Pi_{\{ \text{($\Sigma$-Ax), (I-right)} \}}\\
	\begin{matrix}
		\qquad\qquad
		&
		\rotatebox[origin=c]{90}{
			\ ..................\ \ 
		}
		&
		\begin{matrix}
			\Pi_{Lr'}\\
			\Pi_{Pr}\\
			\Pi_{Gr}\\
			\Pi_{Cr \cup Sr}
		\end{matrix}
	\end{matrix}\\
	\Omega; \Theta, \Delta, \Sigma \vdash \sigma
\end{matrix}
}\]
\end{defi}

Normalised proofs are general for \MuACLst, as shown by the following lemma.

\begin{lem}\label{thm:col-normalization2}
	Any \MuACLst\  proof for a sequent $\Omega; \Theta, \Delta, \Sigma \vdash \sigma$ can be rewritten as an equivalent normalised proof
	that uses ($*$-cut) only if the original one does.
\end{lem}
\begin{proof}
	Given a proof $\Pi$ in \MuACLst\  for the sequent, we rewrite it using \autoref{thm:LP-GCS} until applicable, obtaining %
	\[
	{\normalsize \begin{matrix}
		\Pi_{\{ \text{($\Sigma$-Ax), (I-right)} \}}\\
		\begin{matrix}
			\qquad\qquad\ 
			&
			\rotatebox[origin=c]{90}{
				\ ........... \ 
			}
			&
			\begin{matrix}
				\Pi_{Lr' \cup Pr}\\
				\Pi_{Gr \cup Cr \cup rS}
			\end{matrix}
		\end{matrix}\\
		\Omega; \Theta, \Delta, \Sigma \vdash \sigma
	\end{matrix}
}	\]	
	We rewrite $\Pi_{Gr \cup Cr \cup Sr}$ using \autoref{thm:G-CS},
	and $\Pi_{Lr' \cup Pr}$ using \autoref{thm:P-L} until applicable, obtaining a normalised proof.
\end{proof}

We now establish some auxiliary results about reordering rules in \MuACLs.

\begin{lem}\label{thm:merge-spend}
	Any \MuACLs\ derivation $\Pi$ where ($\linearcontract$-left) is applied immediately after ($\linearcontract$-split) can be transformed in an equivalent derivation $\Pi'$ where the two rule applications are swapped.
	In addition, the equivalent derivation uses ($*$-cut) only if the original one does.
\end{lem}
\begin{proof}
	Let $\Pi$ be
	\[
	{\normalsize \prftree[r]
	{($\linearcontract$-left)}
	{\delta_0 \subseteq \delta_0'}
	{
		\prftree[r]
		{($\linearcontract$-split)}
		{\Omega; \Theta, \delta \otimes \delta'' \linearcontract \delta' \otimes \delta''',  \Delta, \delta_0', \Sigma \vdash \sigma}
		{\Omega; \Theta, \delta \linearcontract \delta', \delta'' \linearcontract \delta''', \Delta, \delta_0', \Sigma \vdash \sigma}
	}
	{\Omega; \Theta, \delta_0 \linearcontract \delta_0', \delta \linearcontract \delta', \delta'' \linearcontract \delta''', \Delta, \Sigma \vdash \sigma}
}	\]
	The derivation $\Pi'$ then is
	\[
	{\normalsize \prftree[r]
	{($\linearcontract$-split)}
	{
		\prftree[r]
		{($\linearcontract$-left)}
		{\delta_0 \subseteq \delta_0'\quad}
		{\Omega; \Theta, \delta \otimes \delta'' \linearcontract \delta' \otimes \delta''',  \Delta, \delta_0', \Sigma \vdash \sigma}
		{\Omega; \Theta, \delta_0 \linearcontract \delta_0', \delta \otimes \delta'' \linearcontract \delta' \otimes \delta''', \Delta, \Sigma \vdash \sigma}
	}
	{\Omega; \Theta, \delta_0 \linearcontract \delta_0', \delta \linearcontract \delta', \delta'' \linearcontract \delta''', \Delta, \Sigma \vdash \sigma}
}\qedhere	\]
\vspace{-7mm}
\end{proof}

\begin{lem}\label{thm:spend-spendmerge}
	If 
	\[
	{\normalsize \prftree[r]
	{($\linearcontract$-left)}
	{\delta \subseteq \delta'}
	{
		\prftree[r]
		{($\linearcontract$-left)}
		{\delta'' \subseteq \delta'''}
		{
			\prftree[noline]
			{\Omega; \Theta, \Delta, \delta', \delta''' \Sigma \vdash \sigma}
		}
		{\Omega; \Theta, \delta'' \linearcontract \delta''', \Delta, \delta' \Sigma \vdash \sigma}
	}
	{\Omega; \Theta, \delta \linearcontract \delta', \delta'' \linearcontract \delta''', \Delta, \Sigma \vdash \sigma}}
	\]
	is a \MuACLs\ derivation, then so is also
	\[
	{\normalsize \prftree[r]
	{($\linearcontract$-split)}
	{
		\prftree[r]
		{($\linearcontract$-left)}
		{\delta \otimes \delta'' \subseteq \delta' \otimes \delta'''}
		{
			\prftree[noline]
			{\Omega; \Theta, \Delta, \delta', \delta''' \Sigma \vdash \sigma}
		}
		{\Omega; \Theta, \delta \otimes \delta'' \linearcontract \delta' \otimes \delta''', \Delta \Sigma \vdash \sigma}
	}
	{\Omega; \Theta, \delta \linearcontract \delta', \delta'' \linearcontract \delta''', \Delta, \Sigma \vdash \sigma}}
	\]
\end{lem}
\begin{proof}
	$\delta \subseteq \delta'$ and $\delta'' \subseteq \delta'''$ 
	clearly imply
	$\delta \otimes \delta'' \subseteq \delta' \otimes \delta'''$.
\end{proof}

We extend the definition of normal forms to \MuACLs\ by adding ($*$-cut) to $Lr$ in~\autoref{def:normalforms}.
Hereafter,
we can only consider normal proofs, as stated by the following theorem 
(subsuming~\autoref{thm:col-normal-ncr}).
Recall that initial sequents are
of the form $\Omega; \Sigma \vdash \sigma$.
\begin{restatable}[Normal proofs]{thm}{CLNLnormalform}\label{thm:col-normal}
	Let $\Omega; \Sigma \vdash \sigma$ be an initial sequent. Then $\Omega; \Sigma \vdash \sigma$ is valid in \MuACL\ (resp. \MuACLs) if and only if a \MuACL\ (resp. \MuACLs) normal proof exists for it.
\end{restatable}
\begin{proof}
	If a normalised proof exists, then the sequent is valid. 
	Assume $\Omega; \Sigma \vdash \sigma$ 
	is proved in \MuACLs\ by $\Pi$.
	First, we rewrite every occurrence of ($\Sigma$-Ax) where $\Omega \neq \emptyset$ as follows
	\[
	{\normalsize \prftree[double, r]
	{(L-Weak)}
	{	
		\prftree[r]
		{($\Sigma$-Ax)}
		{}
		{\res@\usr \vdash \res@\usr}
	}
	{\Omega; \res@\usr \vdash \res@\usr}}
	\]
	obtaining the equivalent proof 
	$\Pi'$. 
	
	Then, we rewrite 
	$\Pi'$
	as an equivalent proof 
	$\Pi_2$
	in \MuACLst\  using \autoref{MuACLsto2}.
	By \autoref{thm:col-normalization2}, the following normalised proof 
	$\Pi_2'$
	exists 
	{\normalsize \[
	\begin{matrix}
		\Pi_{\{ \text{($\Sigma$-Ax), (I-right)} \}}\\
		\begin{matrix}
			\qquad\quad
			&
			\rotatebox[origin=c]{90}{\ .................\ }
			&
			\begin{matrix}
			\\[-.25cm]
				\Pi_{Lr'}\\
				\Pi_{Pr}\\
				\Pi_{Gr}\\
				\Pi_{Cr \cup Sr}
			\end{matrix}
		\end{matrix}
		\\
		\Omega; \Sigma \vdash \sigma
	\end{matrix}
	\]
}	
Since no (L-Weak) rule appears above $\Pi_{Cr \cup Sr}$, and $\Omega = \emptyset$ in the leaves by construction,
	in the derivation $\Pi_{Lr'}$, the non-linear part of the sequent $\Omega$ is $\emptyset$. Thus, 
	$\Pi_2'$
	is a \MuACLs\ proof as well (note that 
	($\otimes$-right') and ($\otimes$-right) coincide when $\Omega = \emptyset$).
	
	If $\Pi_{Pr}$ is empty, then 
	$\Pi_2'$
	is in the normal form 1, otherwise there exist
	$\Omega_G$, $\Theta$, $\Delta$, $\Delta'$ such that 
	$\Pi_2'$
	is 
	as below on the left, and we can rewrite $\Pi_{Pr}$ using \autoref{thm:merge-spend} until applicable, obtaining the proof an the right. 
		\begin{gather*}
		{\normalsize \prfsummary[$\Pi_{Cr \cup Sr}$]
		{
			\prfsummary[$\Pi_{Gr}$]
			{
				\prfsummary[
				$
				\begin{matrix}
					\Pi_{Pr}
				\end{matrix}$
				]
				{
					\prftree
					{\Pi_{Lr \cup \{ \text{($\Sigma$-Ax), (I-right)} \}}}
					{\Delta', \Sigma \vdash \sigma}
				}
				{
					{\Theta, \Delta, \Sigma \vdash \sigma}
				}
			}
			{\Omega_G; \Sigma \vdash \sigma}
		}
		{\Omega; \Sigma \vdash \sigma}}
	\qquad
		{\normalsize \prfsummary[$\Pi_{Cr \cup Sr}$]
		{
			\prfsummary[$\Pi_{Gr}$]
			{
				\prfsummary[
				$
				\begin{matrix}
					\Pi_{\text{($\smlinearcontract$-left)}}\\
					\Pi_{\text{($\smlinearcontract$-split)}}
				\end{matrix}$
				]
				{
					\prftree
					{\Pi_{Lr \cup \{ \text{($\Sigma$-Ax), (I-right)} \}}}
					{\Delta', \Sigma \vdash \sigma}
				}
				{
					{\Theta, \Delta, \Sigma \vdash \sigma}
				}
			}
			{\Omega_G; \Sigma \vdash \sigma}
		}
		{\Omega; \Sigma \vdash \sigma}}
		\end{gather*}

	Finally, rewrite $\Pi_{\text{($\smlinearcontract$-left)}}$ using \autoref{thm:spend-spendmerge} until applicable, obtaining a proof in the normal form 2.
	For \MuACL\ the same holds, and the resulting proof does not use ($*$-cut).
\end{proof}

\subsection{Decidability of \MuACL\ and \MuACLs}

We now prove our main result, i.e., that \MuACL\ and \MuACLs\ are decidable.
We first focus on \MuACLs, as the case for \MuACL\ can be derived easily.

By~\autoref{thm:col-normal}, we only consider normal proofs.
In the following, we verify if a proof in the normal form 1 exists for an initial sequent, then we show how to reduce the normal form 2 case 
to the normal form 1 case.

\subsubsection{Solving the Normal Form 1}

\begin{lem}\label{thm:vvdashomega}
	If a sequent $\Omega; \Theta, \Delta, \Sigma \vdash \sigma$	is derivable from a sequent $\Omega'; \Theta', \Delta', \Sigma' \vdash \sigma'$ only using rules in $Cr \cup Sr$, then $\Omega \Vdash \omega$ holds for all $\omega \in \Omega'$. 
\end{lem}
\begin{proof}
	Trivial by rule induction.
\end{proof}

The existence of the derivation $\Pi_{Cr \cup Sr}$ can be easily verified.
Actually, such a derivation from $\Omega_G; \Sigma \vdash \sigma$ to $\Omega; \Sigma \vdash \sigma$ exists iff $\Omega \Vdash G(\delta)$ for each $G(\delta) \in \Omega_G$.
We actually prove a stronger fact: we give a specific, computable, $\Omega_\star$ that subsumes all the possible cases.

In the following, given a multiset $X$, we write $X^n$ for $\uplus_{i = 1}^n X$.
Moreover, given two linear formulas $x$ and $y$, we identify $\{x \otimes y\}$ with $\{x \} \uplus \{y \}$ as always (given the commutativity and associativity of $\otimes$), and we also write $x^n$ for $\otimes_{i = 1}^n x$.
\begin{restatable}{lem}{CLNLPiC}\label{thm:firsttostar}
	A normal proof exists for $\Omega; \Sigma \vdash \sigma$ if and only if a proof in the same normal form exists for $\Omega_\star; \Sigma \vdash \sigma$  where $\Omega_\star$ contains a single occurrence of 
	every
	 $G(\delta)$ and $G(\theta)$ such that $\Omega \Vdash G(\delta)$ and $\Omega \Vdash G(\theta)$.
	In addition, $\Omega_\star$ can be effectively built starting from $\Omega$.
\end{restatable}
\begin{proof}
	Assume a proof $\Pi$ exists for $\Omega_\star; \Sigma \vdash \sigma$, and
	let $\Omega_\star$ be 
	\[
	{\normalsize \{ G(\theta_i) \mid i \in [1, n] \} \cup \{ G(\delta_i) \mid i \in [1, m] \}}
	\]
	Then a proof of $\Omega, \Sigma \vdash \sigma$ is
\vspace{-5mm}
	\[
	{\normalsize \prftree[r, double]
	{(L-Cont)}
	{
		\prftree[r]
		{(CL-Cut)}
		{\Omega \vdash G(\theta_1)}
		{
			\prfsummary[$\Pi_{\{ \text{(CL-Cut)} \}}$]
			{
				\prftree[r]
				{(CL-Cut)}
				{\Omega \vdash G(\delta_1)\quad}
				{
					\prfsummary[$\Pi_{\{ \text{(CL-Cut)} \}}$]
					{
						\prftree
						{\Pi}
						{\Omega_\star; \Sigma \vdash \sigma}
					}
					{
							\Omega^{m-1},
							G(\theta_1), \dots, G(\theta_n),
							 G(\delta_1); \Sigma 
						\vdash \sigma}
				}
				{\Omega^{m}, G(\theta_1), \dots, G(\theta_n); \Sigma \vdash \sigma}
			}
			{\Omega^{n+m-1} G(\theta_1); \Sigma \vdash \sigma}		
		}
		{\Omega^{n+m}; \Sigma \vdash \sigma}
	}
	{\Omega; \Sigma \vdash \sigma}}
	\]	
	Assume that a normal proof $\Pi$ exists for $\Omega; \Sigma \vdash \sigma$ as follows,
	with $\Pi'$ in normal form too.
	\[
	{\normalsize \prfsummary[$\Pi_{Cr\cup Sr}$]
	{
		\prftree[]
		{\Pi'}
		{\Omega_G; \Sigma \vdash \sigma}
	}
	{\Omega; \Sigma \vdash \sigma}
}	\]
	
	By \autoref{thm:vvdashomega}, all the elements in $\Omega_G$ also occur in $\Omega_\star$ with a single occurrence.
	We write $\Omega_G = \Omega_{\star} \cup \Omega_{cont} \setminus \Omega_{weak}$ where $\Omega_{cont}$ contains the extra occurrences in $\Omega_G$ with respect to $\Omega_\star$, and $\Omega_{weak}$ contains the elements of $\Omega_\star$ that are not in $\Omega_G$.
	
	Then, the following is a proof for $\Omega_\star; \Sigma \vdash \sigma$, where the same normal form is maintained:
	\[
	{\normalsize \prftree[r, double]
	{L-Cont}
	{
		\prftree[r, double]
		{L-Weak}
		{
			\prftree
			{\Pi'}
			{\Omega_G; \Sigma \vdash \sigma}		
		}
		{\Omega_{\star}, \Omega_{cont}; \Sigma \vdash \sigma}
	}
	{\Omega_\star; \Sigma \vdash \sigma}}
	\]
Since $\Vdash$ is just the same as in the multiplicative fragment of intuitionistic logic, $\Omega_{\star}$ can be computed using the decision procedure for intuitionistic propositional logic of~\cite{decideMIL}.
\end{proof}
Since a proof in the normal form 1 contains no rule for 
$\linearcontract$,
 we can 
avoid
 considering $G(\theta)$ when computing $\Omega_\star$ (they are discarded by (L-Weak) rules in 
 every
  valid proof).

Consider now the proof obtained by composing the derivations  
$\Pi_{Lr \cup \{ \text{($\Sigma$-Ax), (I-right)}\}}$ and $\Pi_{Gr \cup Sr}$.
We give an algorithm for deciding if such a proof exists in \MuACLs.
\begin{restatable}{lem}{CLNLdecidefform}\label{thm:firststardec}
	An always terminating algorithm exists that, given $\Omega_\star, \Sigma$, and $\sigma$, decides if $\Omega_\star; \Sigma \vdash \sigma$ is provable in \MuACLs\ only using rules in $Gr \cup Sr \cup Lr \cup \{ \text{($\Sigma$-Ax), (I-right)}\}$.
\end{restatable}
\begin{proof}
	This result derives from a similar one by Kanovich~\cite{Kanovich94}, which is stated for a computational fragment of linear logic that coincides with the sequents that we consider in $\Pi_{Gr \cup Sr}$ and $\Pi_{Lr \cup \{ \text{($\Sigma$-Ax), (I-right)}\}}$.
	Kanovich considers \emph{simple products}, i.e., linear conjunctions of atomic predicates; \emph{Horn-implications}, i.e., linear implications of simple products; and \emph{!-Horn-implications}, i.e., Horn implications preceded by $!$ in the linear logical sense.
	Moreover, he defines \emph{!-Horn-sequents}, i.e., sequents with !-Horn-implications, Horn-implications and simple products as left parts and simple products as right part.
	
	A translation from $\Omega_\star; \Sigma \vdash \sigma$ to !-Horn-sequents is trivially defined: $\Sigma$ and $\sigma$ are simple products, while $\Omega_\star$ is translated by replacing $G$ with $!$ (recall that $\Omega_\star$ contains no contractual implications in the normal form 1).
	Indeed, because of the restriction we have in $\Pi_{Gr \cup Sr}$, the rules applicable to propositions preceded by $G$ are exactly to same of linear logic where $G$ stands for $!$.
	Thus, the sequent $\Omega_\star$; $\Sigma \vdash \sigma$ is provable if and only if the !-Horn-sequent is valid.
	Finally, the problem of checking the validity of a !-Horn sequent (and thus also of our computational sequent) is reduced in~\cite{Kanovich94} to reachability in Petri Nets, which can be decided using the algorithm proposed in~\cite{PetriReach}.
	Roughly, atomic proposition corresponds to places of the Petri Net, and linear implication to transitions.
	The number of tokens in a given place represents the occurrences of the corresponding atomic proposition, and changes according to linear implications that we can use ad libitum.
\end{proof}

\begin{lem}[Normal form 1 decidability]\label{thm:fstnfdecide}
	An always-terminating algorithm exists that decides if an initial sequent is provable in \MuACLs\ using a proof in the normal form 1.
\end{lem}
\begin{proof}
	Trivially derives from~\autoref{thm:firsttostar} and \ref{thm:firststardec}.
\end{proof}

We recover~\autoref{thm:fstnfdecide-ncr} as a special case of the Lemma above.
\fstnfdecidencr*
\begin{proof}
It suffices to adapt Kanovich's encoding as follows, therefore forbidding the outcome of linear implications to be used in transitions.
For each atomic proposition $p$ we define two places of the Petri Net $p_s$ and $p_t$.
For each linear implication $\delta = \sigma \multimap \sigma'$ such that $G(\delta) \in \Omega_{\star}$, we define a transition in the Petri Net
which consumes the tokens from $p_s$, with $p \in \sigma$ and produces the ones for $p'_t$ for $p' \in \sigma'$.
Moreover, we add transitions from each $p_s$ to $p_t$ allowing atomic propositions to be taken as they are (through the ($\Sigma$-Ax) rule). 
Note that we can still use linear implications ad libitum, but we cannot reuse their outcome as an input for others linear implications.
\end{proof}

\subsubsection{Reducing the Normal Form 2 to the Normal Form 1}

\begin{lem}\label{thm:numberG}
	A derivation that only uses $Gr \cup Sr$ exists from $\Theta, \Delta, \Sigma \vdash \sigma$ to $\Omega_G, \Sigma \vdash \sigma$, with
	$
	\Omega_G = \{ G(\theta_i) \mid i \in [1, n] \} \cup \{ G(\delta_j) \mid j \in [1, m] \}
	$
	if and only if $x_1, \dots x_n$ and $z_1, \dots z_m$ nonnegative integers exist such that 
	\begin{align*}
		\Theta = \{ \theta_i^{x_i} \mid i \in [1, n] \} \qquad
		\Delta = \{ \delta_j^{z_j} \mid j \in [1, m] \}
	\end{align*}
\vspace{-3mm}
\end{lem}
\begin{proof}
	Assume a derivation exists.
	By rule induction over the rules in $Gr \cup Sr$ one proves that the linear propositions $\delta$ and $\theta$ appearing in $\Theta, \Delta, \Sigma \vdash \sigma$ are the same that appear in $\Omega_G, \Sigma \vdash \sigma$ preceded by $G$, possibly with a different number of occurrences.
	Let $x_i$ and $z_j$ be such occurrences.
	The thesis trivially follows.
	
	Assume $\Theta$ and $\Delta$ are defined as in the formula above.
	Let $\Omega_G'$ and $\Omega_G''$ be
	\begin{align*}
		\Omega_G' =\ &\{ G(\theta_i) \mid \theta_i^{x_i} \in \Theta \land x_i \neq 0 \}\ \cup\ \{ G(\delta_j) \mid \delta_j^{z_i} \in \Delta \land z_i \neq 0 \}\\
		\Omega_G'' =\ &\{ G(\theta_i^{x_i}) \mid \theta_i^{x_i} \in \Theta \land x_i \neq 0 \}\ \cup\ \{ G(\delta_j^{z_j}) \mid \delta_j^{z_j} \in \Delta \land z_j \neq 0 \}
	\end{align*}
	A derivation exists from $\Theta, \Delta, \Sigma \vdash \sigma$ to $\Omega_G, \Sigma \vdash \sigma$ as follows.
	\[
	{\normalsize \prftree[r, double]
	{(L-Weak)}
	{
		\prftree[r, double]
		{(L-Cont)}
		{
			\prfsummary[$\Pi_{Gr}$]
			{\Theta, \Delta, \Sigma \vdash \sigma}
			{
				\Omega_G'', \Sigma \vdash \sigma
			}
		}
		{\Omega_G', \Sigma \vdash \sigma}	
	}
	{\Omega_G, \Sigma \vdash \sigma}}\qedhere
	\]
\end{proof}

\begin{lem}\label{thm:colapp-merge}
	A derivation that only uses ($\linearcontract$-split) exists from $\theta, \Delta, \Sigma \vdash \sigma$ to $\Theta, \Delta, \Sigma \vdash \sigma$, with 
	$\Theta = \{ \delta_i \linearcontract \delta_i' \mid i \in [0, n] \}$
	if and only if
	\[
	\theta = \bigotimes_{i = 1}^n \delta_i \linearcontract \bigotimes_{i = 1}^n \delta_i'
	\]
\end{lem}
\begin{proof}
	Follows because ($\linearcontract$-split) preserves both the multisets of instances of $\delta$ that appear to the left of $\linearcontract$ and the multiset of the instances of $\delta$ that appear to the right of $\linearcontract$.
\end{proof}

\begin{nota}\label{def:conditions}
	Let  $\mathcal{L}_{\Omega_G} = \{ \ell_1, \dots \ell_p \}$ be the set of linear implications between atomic propositions $\res@\usr \multimap \res'@\usr'$ appearing as terms in $\Omega_G$.
	For every $G(\delta) \in \Omega_G$, let $u_\delta$ be a vector of length $p$ associating each index $k$ with the number of occurrences of $\ell_k$ in $\delta$; 
	and for every $G(\theta) = G(\delta \linearcontract \delta') \in \Omega_G$, let $u_\theta$ and $v_\theta$ be vectors of length $p$ associating each index $k$ with the number of occurrences of $\ell_k$ in $\delta$ and $\delta'$, respectively.
	Moreover, let $u_{\Delta}$ be a vector of length $p$ associating each index $k$ with the sum of the occurrences of $\ell_k$ in every $\delta \in \Delta$.
		Finally, let $A_{\Omega_G}$ be the matrix with a column $u_\delta$ for each $G(\delta) \in \Omega_G$, and let $B_{\Omega_G}$ and $C_{\Omega_G}$ be the matrices with a column $u_\theta$, and $v_\theta$, respectively, for each $G(\theta) \in \Omega_G$.
	{\normalsize \begin{gather*}
		A_{\Omega_G} =
		\begin{bmatrix}
			\vert & \vert \\
			u_{\delta_1} & u_{\delta_n} \\
			\vert & \vert
		\end{bmatrix}
		\qquad
		B_{\Omega_G} =
		\begin{bmatrix}
			\vert & \vert \\
			u_{\theta_1} & u_{\theta_m} \\
			\vert & \vert
		\end{bmatrix}
		\qquad
		C_{\Omega_G} =
		\begin{bmatrix}
			\vert & \vert \\
			v_{\theta_1} & v_{\theta_m} \\
			\vert & \vert
		\end{bmatrix}
	\end{gather*}
}
	Consider the following conditions.
{\normalsize \begin{align}\label{prop:col-1}
	u_{\Delta} =
	\begin{bmatrix}
		\begin{array}{c|c}
			& \\
			{A_{\Omega_G}} & {C_{\Omega_G}}\\
			& \\
		\end{array}
	\end{bmatrix}
	\begin{bmatrix}
		x_{1} \\
		\vdots \\
		x_{n}\\
		z_1\\
		\vdots\\
		z_m
	\end{bmatrix}
\end{align}
}
{\normalsize \begin{align}\label{prop:col-2}
	\begin{bmatrix}
		\begin{array}{c}
			\\
			C_{\Omega_G} - B_{\Omega_G}\\
			\\
		\end{array}
	\end{bmatrix}
	\begin{bmatrix}
		z_{1} \\
		\vdots \\
		z_{m}
	\end{bmatrix}
	\geq
	\begin{bmatrix}
		0 \\
		\vdots \\
		0
	\end{bmatrix}
\end{align}
}
\end{nota}

\begin{restatable}{lem}{CLNLtoLA}\label{thm:col-lltola}
	For %
	every
	 $\Omega_G, \Delta, \Sigma, \sigma$, a derivation exists from $\Delta, \Sigma \vdash \sigma$ to $\Omega_G; \Sigma \vdash \sigma$, in one of the following forms
	{\normalsize \begin{center}	
		\begin{tabular}{c c}
			\prfsummary[$\Pi_{Gr \cup Sr}$]
			{
				\Delta, \Sigma \vdash \sigma
			}
			{\Omega_G; \Sigma \vdash \sigma}	
			&\qquad
			\prfsummary[$\Pi_{Gr \cup Sr}$]
			{
				\prfsummary[
				$\Pi_{\text{($\smlinearcontract$-split)}}$]
				{
					\prftree[r]
					{($\linearcontract$-left)}
					{\Delta, \Sigma \vdash \sigma}
					{\theta, \Delta', \Sigma \vdash \sigma}
				}
				{
					{\Theta, \Delta', \Sigma \vdash \sigma}
				}
			}
			{\Omega_G; \Sigma \vdash \sigma}
		\end{tabular}
	\end{center}
}	
\noindent
if and only if there exist nonnegative integers $x_1, \dots x_n$ and $z_i, \dots z_m$ such that conditions~(\ref{prop:col-1}) and~(\ref{prop:col-2}) of~\autoref{def:conditions} hold.
\end{restatable}
\begin{proof}
	Let 
	$\Omega_G = \{ G(\theta_i) \mid i \in [1, n] \} \cup \{ G(\delta_j) \mid j \in [1, m] \}$.
	Consider a derivation of the form 1.
	By \autoref{thm:numberG}, such a derivation exists if and only if there exist $x_1, \dots x_n$ and $z_1, \dots z_m$ nonnegative integers such that 
	{\normalsize \begin{align*}
		\emptyset = \Theta &= \{ \theta_i^{x_i} \mid i \in [1, n] \} \qquad \text{i.e. $x_i = 0$ for all $i \in [1, n]$} \\
		\Delta &= \{ \delta_j^{z_j} \mid j \in [1, m] \}
	\end{align*}}
	
	Consider now a derivation in the normal form 2.
	By \autoref{thm:numberG} the derivation $\Pi_{Gr \cup Sr}$ exists if and only if there exist $x_1, \dots x_n$ and $z_1, \dots z_m$ nonnegative integers such that
	{\normalsize \begin{align*}
		\Theta = \{ \theta_i^{x_i} \mid i \in [1, n] \}\quad\text{and}\quad
		\Delta' = \{ \delta_j^{z_j} \mid j \in [1, m] \}.
	\end{align*}}
	Then, by \autoref{thm:colapp-merge}, the derivation $\Pi_{\text{($\smlinearcontract$-split)}}$ exists if and only if
{\normalsize 	\[
	\theta = \bigotimes_{j = 1}^m \delta_j^{z_j} \linearcontract \bigotimes_{j = 1}^m (\delta_j')^{z_j}
	\]}
	and the rule ($\linearcontract$-left) is applicable iff both the following hold
	{\normalsize \begin{align}\label{cond:1app}
		\Delta = \Delta' \cup \{ \bigotimes_{j = 1}^m (\delta_j')^{z_j} \}\\
		\label{cond:2app}
		\bigotimes_{j = 1}^m \delta_j^{z_j} \subseteq \bigotimes_{j = 1}^m (\delta_j')^{z_j}
	\end{align}}
	
	Note that these conditions reduce to the ones in the normal form 1 when $x_i = 0$ for 
	every
	 $i \in [1, n]$.
	Thus, we can conclude that a derivation exists if and only if conditions~(\ref{cond:1app}) and ~(\ref{cond:2app}) are met.
	
	We conclude by showing that (\ref{cond:1app}) is equivalent to (\ref{prop:col-1}) and (\ref{cond:2app}) to (\ref{prop:col-2}) of~\autoref{def:conditions}.	
	Recall that $\mathcal{L}_{\Omega_G} = \{ \ell_1, \dots \ell_p \}$ is the set of linear implications between atomic propositions appearing as terms in $\Omega_G$.
	By definition of $\delta$, we can rewrite conditions~(\ref{cond:1app}) and (\ref{cond:2app}) respectively as follows.
{\normalsize 	\begin{gather*}
		\Delta = \{ (\bigotimes_{k = 1}^{p} \ell_k^{A_{k, i}})^{x_i} \mid i \in [1, n] \} \cup \{ \bigotimes_{j = 1}^m (\bigotimes_{k = 1}^{p} \ell_k^{C_{k, j}})^{z_j} \}\\
		\bigotimes_{j = 1}^m (\bigotimes_{k = 1}^{p} \ell_k^{B_{k, j}})^{z_j} \subseteq \bigotimes_{j = 1}^m (\bigotimes_{k = 1}^{p} \ell_k^{C_{k, j}})^{z_j}
	\end{gather*}}
	where for each $i$, and $k$, $A_{k, i}$ is the number of occurrences of $\ell_k$ in $\delta_i$; for each $j$, and $k$, $B_{k, j}$ is the number of occurrences of $\ell_k$ in $\delta_j$, and $C_{k, j}$ is the number of occurrences of $\ell_k$ in $\delta_j'$.
	By definition, $A_{\Omega_G}$, contains $A_{k, i}$ in row $k$, column $i$; and $B_{\Omega_G}$, and $C_{\Omega_G}$ contains $B_{k, j}$ and $C_{k, j}$ in row $k$, column $j$ respectively.
	The equivalence between conditions (\ref{prop:col-1}) and (\ref{cond:1app}) follows straightforwardly.
	
	Take any $z_1, \dots z_m$.
	Condition (\ref{cond:2app}) holds iff, for 
	every
	 $\ell_k$ the number of occurrences in the left part of (\ref{cond:2app}) is greater than the number of occurrences in the right part, i.e.,
{\normalsize 	\[
	\bigotimes_{j = 1}^m (\ell_k^{B_{k, j}})^{z_j} \subseteq \bigotimes_{j = 1}^m (\ell_k^{C_{k, j}})^{z_j}\qquad{\text{for every $\ell_k$}}
	\]
}	By definition, this holds if and only if the $k$-th rows $B_k$ of $B_{\Omega_G}$ and $C_k$ of $C_{\Omega_G}$ are such that
{\normalsize 	\[
	\begin{bmatrix}
		C_{k, 1} & \dots & C_{k, 1}
	\end{bmatrix}
	\begin{bmatrix}
		z_1 \\
		\vdots \\
		z_m
	\end{bmatrix}
	\geq
	\begin{bmatrix}
		B_{k, 1} & \dots & B_{k, 1}
	\end{bmatrix}
	\begin{bmatrix}
		z_1 \\
		\vdots \\
		z_m
	\end{bmatrix}
	\]
}	which, in turn, is true for every  $\ell_k$ if and only if condition (\ref{prop:col-2}) of~\autoref{def:conditions} holds.
\end{proof}

Consider condition~(\ref{prop:col-2}) of~\autoref{def:conditions}. For the Hilbert basis theorem~\cite{HB}, the set of nonnegative integer solutions can be expressed as 
\begin{align*}
	{\normalsize \begin{bmatrix}
		z_{1} \\
		\vdots \\
		z_{m}
	\end{bmatrix}
	=
	\begin{bmatrix}
		\\
		H_{\Omega_G}\\
		\\
	\end{bmatrix}
	\begin{bmatrix}
		y_{1} \\
		\vdots \\
		y_{q}
	\end{bmatrix}}
\end{align*}
with $y_1, \dots y_p$ nonnegative integers, and $H$ can be computed using~\cite{computeHB}.

Thus, we build the following system precisely characterising the solutions of conditions~(\ref{prop:col-1}) and~(\ref{prop:col-2}) of~\autoref{def:conditions}:
\begin{align*}
	{\normalsize \begin{bmatrix}
		\begin{array}{c|c}
			& \\
			{A_{\Omega_G}} & {C_{\Omega_G}} \cdot H_{\Omega_G}\\
			& \\
		\end{array}
	\end{bmatrix}
	\begin{bmatrix}
		x_{1} \\
		\vdots \\
		x_{n}\\
		y_1\\
		\vdots\\
		y_q
	\end{bmatrix}}
\end{align*}
Let $D_{\Omega_G}$ be the matrix above.
We take a multiset $\Omega_G'$ containing a proposition $G(\delta_v)$ for each column $v$ of $D_{\Omega_G}$ with
\vspace{-3mm}
{\normalsize \begin{align}\label{prop:col-3}
	\delta_v = \bigotimes_{\ell_k \in \mathcal{L}_{\Omega_G}} \ell_k^{v_k}
\end{align}}
where $v_k$ is the value with index $k$ in $v$, $\ell^0 = I$ and $\ell^{n+1} = \ell \otimes \ell^n$.

The derivability of $\Omega_G; \Sigma \vdash \sigma$ is the same as $\Omega_G'; \Sigma \vdash \sigma$.
Formally:
\sftoff*
\begin{proof}
	Let $\Omega_G'$ be as in condition~(\ref{prop:col-3}).
	The lemma holds by construction and~\autoref{thm:col-lltola}.
\end{proof}

\begin{thm}[\MuACLs\ decidability]\label{thm:MuACLsdecide}
	An always-terminating algorithm exists that decides if an initial sequent is valid in \MuACLs.
\end{thm}
\begin{proof}
	\autoref{thm:col-sftoff} reduces the problem of finding a proof in the normal form 2 to finding a proof in the normal form 1, which is proved decidable in~\autoref{thm:fstnfdecide}.
\end{proof}

\MuACLsdec*
\begin{proof}
	Follows directly by \autoref{thm:col-sftoff}, which reduces the problem of finding a proof in the normal form 2 to finding a proof in the normal form 1, proved decidable in~\autoref{thm:fstnfdecide-ncr}.
	Note that the ($*$-cut) rule is not used in the reductions we target, as the derivations in~\autoref{thm:col-sftoff} only use a common subset of the rules of \MuACL\ and \MuACLs.
\end{proof}

\section{Linear logic does not include \MuACL}

We prove a stronger result than~\autoref{thm:homo}, namely that a homomorphic map that works well with initial sequents does not exist.
\begin{lem}\label{thm:homolemma}
Let $m(\cdot)$ be a homomorphic map and $\Phi$ be a (multi)set of \exlo\  propositions $\varphi$. Then we have that both
$\quad (i) \ m(\varphi) = \varphi \quad$ and
$\quad (ii) \ \Phi \vdash_{\MuACL} \varphi$ iff $\Phi \vdash_{\exlo} \varphi$.
\end{lem}
\begin{proof}
By induction and case analysis.
\end{proof}

\begin{thm}\label{thm:homoinit}
There is no homomorphic map of  \MuACL\ to \exlo\ that is complete and correct for initial sequents.
\end{thm}
\begin{proof}
Let $\Sigma = \{\res@\usr, \res'@\usr'\}$, $\delta = \res@\usr \multimap \res@\usr'$, $\delta' = \res'@\usr' \multimap \res'@\usr$.
Assume by refutation that there exists a homomorphic map $m(\cdot)$ complete and correct for initial sequents, and that $m(\cdot)$ maps $\delta' \linearcontract \delta$ to some unspecified \exlo\ proposition $m(\delta' \linearcontract \delta)$.
Consider the following propositions and multisets:
\begin{align*}
\Omega &= \{ G(\delta \linearcontract \delta'), G(m(\delta' \linearcontract \delta)) \} & \sigma = \res'@\usr \otimes \res@\usr' &&
\Omega' &= \{ G(\delta \linearcontract \delta'), G(\delta' \linearcontract \delta) \} 
\end{align*}
Note that $\Omega', \Sigma \vdash_{\MuACL} \sigma$.
Hence, since $m(\cdot)$ is complete, $m(\Omega'), m(\Sigma)  \vdash_{\exlo} m(\sigma)$ must hold.
Since $m(m(\delta' \linearcontract \delta)) = m(\delta' \linearcontract \delta)$ holds by~\autoref{thm:homolemma}, we have
\begin{align*}
m(\Omega') 
&= \{ m(G(\delta \linearcontract \delta')), m(G(\delta' \linearcontract \delta)) \} && \\
 &= \{ G(m(\delta \linearcontract \delta')), G(m(\delta' \linearcontract \delta)) \} &&  m(\cdot)\text{ is homomorphic}\\
 &= \{ G(m(\delta \linearcontract \delta')), G(m(m(\delta' \linearcontract \delta))) \} && \text{by~\autoref{thm:homolemma}}\\
 &= \{ m(G(\delta \linearcontract \delta')), m(G(m(\delta' \linearcontract \delta)))) \} && m(\cdot)\text{ is homomorphic}\\
&= m(\Omega).
\end{align*}
As a result, $m(\Omega), m(\Sigma) \vdash_{\exlo} m(\sigma)$ must hold, but, since $m(\cdot)$ is correct, $\Omega, \Sigma \vdash_{\MuACL} \sigma$ must hold, too.
By~\autoref{thm:col-normal}, a proof in normal form must exist for $\Omega, \Sigma \vdash_{\MuACL} \sigma$.
If it is in normal form 2, then, 
by~\autoref{thm:numberG} and \ref{thm:colapp-merge}, the proposition
$\theta$ of the normal form (\autoref{def:normalforms}) must be composed as 
$(\bigotimes_{i = 1}^n \delta) \linearcontract (\bigotimes_{i = 1}^n \delta')$ for some $n \geq 1$.
Since $\delta \not\subseteq \delta'$, it must be that $(\bigotimes_{i = 1}^n \delta) \not\subseteq (\bigotimes_{i = 1}^n \delta')$, and the rule ($\linearcontract$-left) cannot apply.
Therefore, only normal form 1 is possible, but then $\Omega_G$ of the normal form can only contain $G(m(\delta' \linearcontract \delta)) = m(G(\delta' \linearcontract \delta))$ because $m(\cdot)$ is homomorphic.
Thus, by the shape of normal form 1, it must be that $m(G(\delta' \linearcontract \delta)), \Sigma \vdash_{\MuACL} \sigma$, and, by~\autoref{thm:homolemma}, 
$m(G(\delta' \linearcontract \delta)), \Sigma \vdash_{\exlo} \sigma$.
Hence, $G(\delta' \linearcontract \delta), \Sigma \vdash_{\MuACL} \sigma$, because $m(\cdot)$ is correct, but the sequent is trivially not valid.
\end{proof}

\nohomo*
\begin{proof}
Follows from~\autoref{thm:homoinit}.
\end{proof}

\section{Correctness and Completeness of the Compilation}\label{app:compile}

We start with an auxiliary result stating a general property of \MuACL\ and \MuACLs: both are monotone with respect to the non-linear part of the antecedent of sequents.
\begin{prop}[Non-linear Monotony]~\label{thm:monotonicity}
	For each $\Omega, \Omega', \Theta, \Delta, \Sigma, \sigma$, if $\Omega; \Theta, \Delta, \Sigma \vdash \sigma$ is valid in \MuACL\ (or \MuACLs), then 
	$\Omega, \Omega'; \Theta, \Delta, \Sigma \vdash \sigma$ is valid in \MuACL\ (or \MuACLs).
\end{prop}
\begin{proof}
	Assume a proof $\Pi$ exists for $\Omega; \Theta, \Delta, \Sigma \vdash \sigma$. Then a proof for $\Omega, \Omega'; \Theta, \Delta, \Sigma \vdash \sigma$ is
{\normalsize 	\[
	\prftree[r, doubleline]
	{(L-Weak)}
	{
		{
			\begin{matrix}
				\Pi\\
				\Omega; \Theta, \Delta, \Sigma \vdash \sigma				
			\end{matrix}
		}
	}
	{\Omega, \Omega'; \Theta, \Delta, \Sigma \vdash \sigma}\qedhere
	\]}
\vspace{-3mm}
\end{proof}

We prove now that the compilation of MuAC into \MuACL\ is correct and complete, and we estimate the size of a \MuACL\ proof for fair transitions and computations.
In the following, we assume as given a ruleset $R_{\usr}$ for each $\usr$, and a context $C$.
We first present some notation that links exchange environments and \MuACL\ theories.
\begin{nota}
	Given a transfer $\tr = \usr \xmapsto{\res} \usr'$, an exchange $\exc = \{\tr_i\}_{i \in I}$, a policy $\pol_{\usr}$, and a user $\usr$, we write:
	\begin{itemize}
		\item $\Delta_{\tr}$ for $\res@\usr \multimap \res@\usr'$;
		\item $exc_{+\usr}$ for $\{ \res@\usr' \multimap \res@\usr \in \exc \}$;
		\item $exc_{-\usr}$ for $\{ \res@\usr \multimap \res@\usr' \in \exc \}$;
		\item $\Delta_{\exc}$ for $\biguplus_{i \in I}\Delta_{\tr_i}$;
		\item $\Theta_{\pol_\usr}$ for 
		    $\{ \Delta_{\exc} \linearcontract (\res@\usr' \multimap \res@\usr'') \mid
			 (\usr' \xmapsto{\res} \usr'' \triangleleft \exc) \in \pol_\usr \}$;
		\item $\Omega_{\pol_\usr}$ for $\{ G(\theta) \mid \theta \in \Theta_{\pol_\usr} \}$;
		\item $\Omega_{\pol}$ for $\biguplus_{\usr \in \Usr} \Omega_{\pol_\usr}$;
		\item $\Omega_C$ for $\semantics{C}$;
		\item $\Omega_{R}$ for $\biguplus_{\usr \in \Usr} \semden{R_\usr}$.
	\end{itemize}
\end{nota}

Moreover, we 
define the \emph{size} of a derivation as the number of inference rules occurring in it.
Note that $\Omega_R$ is composed by formulas $\omega \rightarrow G(\theta)$ with $\omega$ being the conjunction of a number of non-linear atomic propositions $p(\usr_1, \dots, \usr_n)$.
For each atomic proposition $p(\usr_1, \dots, \usr_n)$, we call $\CS_{p(\usr_1, \dots, \usr_n)}$ the smallest proof for $\semantics{C} \Vdash p(\usr_1, \dots, \usr_n)$, if any.
We let $\CS_{C}$ be the maximal size of such proofs, and write $\CS_{C, R}$ for $\CS_{C}$ times the maximum number of atomic propositions in $\omega$ for $\omega \rightarrow G(\theta)$ in $\Omega_R$.
We write $|\Sigma|, |\Delta|, |\Theta|$ for the number of elements in the multisets, identifying $\{ \sigma \otimes \sigma' \}$ with $\{ \sigma, \sigma' \}$.

\subsection{Correctness}

Hereafter, we only consider proofs of initial sequents that are the encoding of MuAC rulesets, states and contexts. 
We start by noticing that the normal forms for such proofs have specific constraints, as shown in the following lemma.

\begin{lem}[MuAC normal form]\label{thm:specnormal}
	A proof $\Pi$ for a sequent $\biguplus_{\usr \in \Usr} \semden{R_\usr}, \semden{C}; \semden{\MuACstate} \vdash \semden{\MuACstate'}$ in \MuACLs\ or \MuACL\ is \emph{normal} iff it can be decomposed in either form	
{\normalsize 	\begin{center}
		\begin{tabular}{c}
			\begin{tabular}{c}
				\prfsummary[$\Pi_{Cr \cup Sr}$]
				{
					\prfsummary[$\Pi_{Gr \cup Sr}$]
					{
						\prftree
						{\Pi_{Lr \cup \{ \text{($\Sigma$-Ax), (I-right)}\}}}
						{\Sigma \vdash \sigma}
					}
					{\Omega_{G}; \Sigma \vdash \sigma}
				}
				{\Omega_R, \Omega_C; \Sigma \vdash \sigma}\\[0.2cm]
				\textit{{\normalsize normal form 1}}
			\end{tabular}
			
			\begin{tabular}{c}
				\hspace{0.5cm}	
				\prfsummary[$\Pi_{Cr \cup Sr}$]
				{
					\prfsummary[$\Pi_{Gr \cup Sr}$]
					{
						\prfsummary[$\Pi_{\{ \text{($\smlinearcontract$-split)} \}}$]
						{
							\prftree[r]
							{($\linearcontract$-left)}
							{
								\prftree
								{\Pi_{Lr \cup \{ \text{($\Sigma$-Ax), (I-right)}\}}}
								{\Delta, \Sigma \vdash \sigma}
							}
							{\theta, \Sigma \vdash \sigma}
						}
						{
							{\Theta, \Sigma \vdash \sigma}
						}
					}
					{\Omega_G; \Sigma \vdash \sigma}
				}
				{\Omega_R, \Omega_C; \Sigma \vdash \sigma}\\[0.2cm]
				\textit{{\normalsize normal form 2}}
			\end{tabular}
		\end{tabular}
	\end{center}
}\end{lem}
\begin{proof}
	The \MuACL\ encoding of~\autoref{def:compile} implies every $\omega \in \Omega$ to be of the form 
	\[
	\omega ::= \top \mid p(\usr_1, \dots, \usr_n) \mid \omega \land \omega \mid \omega \rightarrow \omega \mid \omega \rightarrow G \theta \\
	\]	
	Take the normal form 1 of~\autoref{def:normalforms}.
	We must show that $\Delta = \emptyset$.
	Clearly, this is the case, since $G \delta$ is not a subterm of $\Omega_G$.
\\	
	For the same reason, in a proof in the normal form 2 of~\autoref{def:normalforms}, $\Delta'$ must be empty.
\end{proof}

Proofs in the normal form 1 are trivial because they correspond to proofs where the state does not change (and thus both the correctness and completeness in this case follow trivially).
Hence in the following we will only consider proofs in the normal form 2, i.e., the ones corresponding to nonempty exchanges.

\begin{lem}\label{thm:polinter}
	For 
	every
	 $\Sigma, \sigma$, a \MuACL\ derivation exists from $\Omega_\pol; \Sigma \vdash \sigma$ to $\Omega_R, \Omega_C; \Sigma \vdash \sigma$ of size $O(|\Omega_{\pol}| \cdot \CS_{C, R} + |\Omega_C|)$.
\end{lem}
\begin{proof}
	It follows from the property below
	by using ($\Omega$-cut) and 
	since  $\pol_\usr = \semdeno{R_\usr} C$.
	\[
	\text{If }\tr \triangleleft \exc \in \semdeno{R_{\usr}} C \text{ then } \semden{R_{\usr}}, \semden{C} \Vdash G(\Delta_{\exc} \linearcontract \Delta_{\tr}).
	\]
	Assume $\tr \triangleleft \exc \in \semdeno{R_{\usr}} C$, then $\tr \triangleleft \exc \in \semdeno{r} C$ for some MuAC rule $r \in R_\usr$.
	By definition and since $\tr \triangleleft \exc \in \semdeno{R_{\usr}} C$, $\semdeno{PredLs} C \rho$ holds for some $\rho$.
	Then, for the same $\rho$, by~\autoref{def:compile-state}, $\semden{C} \Vdash \semden{PredLs}[\rho(u) / u]$.
	
	Finally, let $\semdeno{r} C$ be $\Uplambda[u].\omega \rightarrow G(\theta)$, with $\omega = \semden{PredLs}_\usr$.
	The derivation can be constructed with a single application of (L-$\rightarrow$-left), and the size of the proof for the left premise is $O(\mathit{CS})$ by definition.
	
	The size of the derivation is thus $\CS_{C, R}$ for 
	each
	 formula in $\Omega_{\pol}$, plus an instance of the (L-Weak) rule for 
	 each
	  formula in $\Omega_C$.
\end{proof}

Hereafter, we write $\Omega^M_{\pol}$ for a multiset defined on $\Omega_{\pol}$, i.e., a function from $\Omega_{\pol}$ to $\mathbb{N}$.

\begin{lem}\label{thm:formavaltofair}
	If a \MuACL\ proof $\Pi$ exists for a sequent $\biguplus_{\usr \in \Usr} \semden{R_\usr}, \semden{C}; \semden{\MuACstate} \vdash \semden{\MuACstate'}$ in the normal form 2, then it can be decomposed as
{\normalsize 	\[
	\prfsummary[$\Pi_{Cr \cup Sr}$]
	{
		\prfsummary[$\Pi_{Sr}$]
		{
			\prftree[doubleline, r]
			{(G-left-$\theta$)}
			{
				\prfsummary[$\Pi_{\{ \text{($\smlinearcontract$-split)} \}}$]
				{
					\prftree[r]
					{($\linearcontract$-left)}
					{
						\prftree
						{\Pi_{Lr \cup \{ \text{($\Omega$-Ax), (I-right)}\}}}
						{\Delta, \Sigma \vdash \sigma}
					}
					{\theta, \Sigma \vdash \sigma}
				}
				{
					{\Theta, \Sigma \vdash \sigma}
				}
			}
			{\Omega^M_{\pol}; \Sigma \vdash \sigma}
		}
		{\Omega_{\pol}; \Sigma \vdash \sigma}
	}
	{\Omega_R, \Omega_C; \Sigma \vdash \sigma}
	\]
}\end{lem}
\begin{proof}
	By~\autoref{thm:polinter}, the derivation $\Pi_{Cr \cup Sr}$ exists.
	Moreover, by~\autoref{thm:specnormal} a derivation exists from $\Theta, \Sigma \vdash \sigma$ to $\Omega_G; \Sigma \vdash \sigma$.
	Note that, by definition, $\Omega_G \subseteq \Omega_{\pol}$, hence, by~\autoref{thm:monotonicity}, a proof exists from $\Theta, \Sigma \vdash \sigma$ to $\Omega_{\Pol}; \Sigma \vdash \sigma$.
	Finally, note that structural rules in $Sr$ can be moved in the bottom of the derivation by~\autoref{thm:G-CS}.
\end{proof}

\begin{lem}\label{thm:acceptderivetofair}
	Let $R$, $C$ be such that the following is a valid \MuACL\ derivation,
{\normalsize 	\[
\prfsummary[$\Pi_{Cr \cup Sr}$]
{
	\prfsummary[$\Pi_{Sr}$]
	{
		\prftree[doubleline, r]
		{(G-left-$\theta$)}
		{
			\prfsummary[$\Pi_{\{ \text{($\smlinearcontract$-split)} \}}$]
			{
				\prftree[r]
				{($\linearcontract$-left)}
				{
					{\Delta_\exc, \Sigma \vdash \sigma}
				}
				{\theta, \Sigma \vdash \sigma}
			}
			{
				{\Theta, \Sigma \vdash \sigma}
			}
		}
		{\Omega^M_{\pol}; \Sigma \vdash \sigma}
	}
	{\Omega_{\pol}; \Sigma \vdash \sigma}
}
{\Omega_R, \Omega_C; \Sigma \vdash \sigma}
\]
}	Then $\exc$ is a fair exchange.
\end{lem}
\begin{proof}
	The result follows from the following property implied by $\pol_\usr = \semdeno{R_\usr} C$.
	\[
	\text{If } \semden{R_{\usr}}, \semden{C} \Vdash G(\Delta_{\exc} \linearcontract \Delta_{\tr}) \text{ then } \tr \triangleleft \exc \in \semdeno{R_{\usr}} C.
	\]
	Note indeed the $\Delta_{\exc} \linearcontract \Delta_{\tr} \in \Theta$ implies $\tr \triangleleft \exc \in \semdeno{R_{\usr}} C$ for some $\usr$, and hence
	$\pol_{\usr} \vDash_{\exc} \tr$ and $\pol_{\usr'} \vDash_{\emptyset} \tr$ for $\usr' \neq \usr$.
	Note also that $\pol_{\usr} \vDash_{\exc_1'} \exc_1$ and $\pol_{\usr} \vDash_{\exc_2'} \exc_2$ implies $\pol_{\usr} \vDash_{\exc_1' \uplus \exc_2'} \exc_1 \uplus \exc_2$.
	
	By definition, we also have that $\theta = \Delta_{\exc'} \linearcontract \Delta_{\exc}$ .
	Due to the previous result, we know that all the user policies accept the given exchange $\pol_{\usr} \vDash_{\exc_{\usr}} \exc$.
	Finally, the left premise of ($\linearcontract$-left) guarantees that no double spending occur, since 
	\[
		\Delta_{\exc'} = \biguplus_{\Delta_{\exc''} \linearcontract \Delta_{\tr} \in \Theta} \Delta_{\exc''} = \biguplus_{\usr \in \Usr} \Delta_{\exc_{\usr}}\qedhere
	\]
\end{proof}

\begin{lem}\label{thm:deltaexctofair}
	For 
	every
	 $\MuACstate, \MuACstate', \exc$, if $\Delta_{\exc}, \semantics{\MuACstate} \vdash \semantics{\MuACstate'}$ is valid in \MuACL\ then $\MuACstate \xrightarrow{\exc} \MuACstate'$.
\end{lem}
\begin{proof}
By induction on the rules of \MuACL.
We can ignore rules that are not applicable due to the form of the sequent.
The property trivially holds for ($\Sigma$-Ax) with $\exc = \emptyset$, and for ($\otimes$-left) since we assume formulas combined with $\otimes$ to be the same as multisets.
Consider the rule ($\otimes$-right), the property follows by the induction hypothesis since 
$\st_1 \xrightarrow{\exc_1} \st_1'$ and $\st_2 \xrightarrow{\exc_2} \st_2'$ implies $(\st_1 \uplus \st_2) \xrightarrow{\exc_1 \uplus \exc_2} (\st_1' \uplus \st_2')$.
Consider the rule ($\multimap$-left), and note that $\semantics{\st} \vdash \semantics{\st'}$ implies $\st = \st'$.
Then, the result follows by definition of $\Delta_{\exc}$.
\end{proof}

\begin{lem}[validity implies fairness]\label{thm:validitytofairness}
	For
	every
	 $\st, \st'$, if $\Omega_R, \Omega_C; \semantics{\st} \vdash \semantics{\st'}$ is valid in \MuACL, then 
	$\st \xrightarrow{\exc} \st'$ is a fair transition for some $\exc$.
\end{lem}
\begin{proof}
Follows from~\autoref{thm:formavaltofair}, \autoref{thm:acceptderivetofair} and ~\autoref{thm:deltaexctofair}
\end{proof}

\begin{lem}\label{thm:deltaexctofaircomp}
	For 
	every
	 $\MuACstate_0, \MuACstate_n$ and $\exc$, if $\Delta_{\exc}, \semantics{\MuACstate_0} \vdash \semantics{\MuACstate_n}$ is valid in \MuACLs\ then $\MuACstate_0 \xrightarrow{\exc_1} \MuACstate_1 \xrightarrow{\exc_2} \dots \xrightarrow{\exc_n} \MuACstate_n$ is a computation with $\uplus_{i = 1}^n \exc_i = \exc$.
\end{lem}
\begin{proof}
By induction on the rules of \MuACLs.
We can ignore rules that are not applicable due to the form of the sequent.
For ($\Sigma$-Ax), ($\otimes$-left), ($\otimes$-right) and ($\multimap$-left) the result follows from~\autoref{thm:deltaexc} (note that a transition is a computation of length $1$).

Consider the rule 
{\normalsize \[
\prftree[r]
{($*$-cut)}
{\Delta_{\exc_1}, \semantics{\st} \vdash \semantics{\st'}}
{\Delta_{\exc_2}, \semantics{\st'} \vdash \semantics{\st''}}
{\Delta_{\exc_1}, \Delta_{\exc_2}, \semantics{\st} \vdash \semantics{\st''}}
\]
}By the induction hypothesis, we know that $\st \xrightarrow{\exc_{1,1}} \dots \xrightarrow{\exc_{1,n}} \st'$ and $\st' \xrightarrow{\exc_{2,1}} \dots \xrightarrow{\exc_{2,m}} \st''$ are computations, with $\biguplus_{i = 1}^n \exc_{1,i} = \exc_1$ and $\biguplus_{i = 1}^m \exc_{2,i} = \exc_2$.
The result then trivially derives from noticing that
$\Delta_{\exc_1} \uplus \Delta_{\exc_2} = \Delta_{\exc_1 \uplus \exc_2}$.
\end{proof}

\begin{lem}[validity implies eventual fairness]\label{thm:validitytofairnesscomp}
	For 
	every
	 $\st, \st'$, if $\Omega_R, \Omega_C; \semantics{\st} \vdash \semantics{\st'}$ is valid in \MuACLs, then 
	$\st \rightarrow^* \st'$ is an eventually fair computation.
\end{lem}
\begin{proof}
Follows from~\autoref{thm:formavaltofair}, \autoref{thm:acceptderivetofair} and ~\autoref{thm:deltaexctofaircomp}
\end{proof}

\subsection{Completeness and Size of \MuACL\ Proofs}\label{sec:fairnesstovalidity}

\begin{lem}\label{thm:acceptderive}
	For 
	every
	 $\exc$ and $\exc' \neq \emptyset$, if $\pol_\usr \vDash_{\exc'} \exc$ then, for 
	 every
	  $\Sigma$ and $\sigma$, 
	a \MuACL\ derivation exists of size $O(|\exc'|)$ from $\Omega_{\pol_\usr}; \Delta_{\exc'} \linearcontract \Delta_{\exc^{\ }_{-\usr}}, \Sigma \vdash \sigma$ to $\Omega_{\pol_\usr}; \Sigma \vdash \sigma$.
\end{lem}
\begin{proof}
	By induction on the definition of $\vDash$.
	The base case is trivial.
	Let $\pol_{\usr} \vDash_{\exc \uplus \exc''} \{\tr \} \uplus \exc \uplus \exc'$.
	Then $\tr \triangleleft \exc \in \pol_{\usr}$ and $\pol_{\usr} \vDash_{\exc''} \exc'$.
	
	Assume $\exc'' \neq \emptyset$.
	We can write the following, where $\Pi$ of size $O(|\exc''|)$ exists by the induction hypothesis.
{\normalsize 	\[
	\prftree[r]
	{(L-Cont)}
	{
		\prftree[r]
		{($G$-left-$\theta$)}
		{
			\prfsummary[$\Pi$]
			{
				\prftree[r]
				{($\linearcontract$-split)}
				{
						\Omega_{\pol_\usr};
						(\Delta_{\exc} \otimes \Delta_{\exc''}) \linearcontract (\Delta_{tr} \otimes \Delta_{\exc'_{-\usr}}),
						\Sigma 
					\vdash \sigma}
				{
						\Omega_{\pol_\usr};
						\Delta_{\exc} \linearcontract \Delta_{tr},
						\Delta_{\exc''} \linearcontract \Delta_{\exc'_{-\usr}},
						\Sigma 
					\vdash \sigma}
			}
			{
				\Omega_{\pol_\usr}; \Delta_{\exc} \linearcontract \Delta_{tr}, \Sigma \vdash \sigma}
		}
		{G(\Delta_{\exc} \linearcontract \Delta_{tr}), \Omega_{\pol_\usr}, \Sigma \vdash \sigma}
	}
	{\Omega_{\pol_\usr}; \Sigma \vdash \sigma}
	\]
}	The case for $\exc'' = \emptyset$ trivially follows by noticing that $\Delta_{({\{\tr \} \uplus \exc \uplus \exc'})_{-\usr}} = \Delta_{\tr}$.
\end{proof}

\begin{lem}\label{thm:fairderive}
	For 
	every
	 fair $\exc$, and for 
	 every
	 $\Sigma, \sigma$, a derivation exists from $\Delta_{\exc}, \Sigma \vdash \sigma$ to
	$\Omega_R, \Omega_C; \Sigma, \sigma$ of size $O(|\exc| + |\Omega_R|\cdot \CS_{C, R} + |\Omega_C|)$.
\end{lem}
\begin{proof}
	Consider the following derivation.
	{\normalsize \[
	\prfsummary[$\Pi'$]
	{
		\prfsummary[$\Pi$]
		{
			\prftree[noline]
			{\Delta_{\exc}, \Sigma \vdash \sigma}
		}
		{\Omega_{\pol}; \Sigma \vdash \sigma}
	}
	{\Omega_R, \Omega_C; \Sigma \vdash \sigma}
	\]}
\mbox{
$\!\!\!\!\!$
By~\autoref{thm:polinter}, $\exists\, \Pi'$ of size $O(|\Omega_{\pol}| \cdot \CS_{C, R} + |\Omega_C|)$, it suffices then showing that $\Pi$ exists.
}
\\	
	By~\autoref{def:fairexc}, $\exc'_{\usr}$ exists for each user such that 
\mbox{$\pol_{\usr} \vDash_{\exc'_{\usr}} \exc$,
with $\uplus_{\usr} \exc'_{\usr} \subseteq \exc$.}
	Then by~\autoref{thm:acceptderive}, 
	every
	 $\Omega_{\pol_\usr}; \Sigma \vdash \sigma$ is derivable from $\Omega_{\pol_\usr}; \Delta_{\exc'_{\usr}} \linearcontract \Delta_{\exc_{-\usr}}; \Sigma \vdash \sigma$.
	We can easily
	compose these derivations obtaining the following.
	{\normalsize \[
	\prfsummary[]
	{
		\prftree[noline]
		{\Omega_{\pol}; \biguplus_{\usr} (\Delta_{\exc'_{\usr}} \linearcontract \Delta_{\exc_{-\usr}}); \Sigma \vdash \sigma}
	}
	{\Omega_{\pol}; \Sigma \vdash \sigma}
	\]}
	The size of this derivation is $O(\sum_{\usr}|\exc'_{\usr}|)$, which is limited by $O(|\exc|)$ since $\uplus_{\usr} \exc'_{\usr} \subseteq \exc$.
	We then build the top of $\Pi$ as follows.
	{\normalsize \[
	\prftree[r, doubleline]
	{($\linearcontract$-split)}
	{
		\prftree[r]
		{($\linearcontract$-left)}
		{\biguplus_{\usr} \Delta_{\exc'_\usr} \subseteq \biguplus_{\usr} \Delta_{exc_{-\usr}}}
		{
			\prftree[r, doubleline]
			{(L-Weak)}
			{\biguplus_{\usr}\Delta_{\exc_{-\usr}}, \Sigma \vdash \sigma}
			{
					\Omega_{\pol};
					\biguplus_{\usr}\Delta_{\exc_{-\usr}},
					\Sigma
				\vdash \sigma}
		}
		{\Omega_{\pol}; (\biguplus_{\usr} \Delta_{\exc'_{\usr}}) \linearcontract (\biguplus_{\usr}\Delta_{\exc_{-\usr}}); \Sigma \vdash \sigma}
	}
	{\Omega_{\pol}; \biguplus_{\usr} (\Delta_{\exc'_{\usr}} \linearcontract \Delta_{\exc_{-\usr}}); \Sigma \vdash \sigma}
	\]
}	
	Note that $\exc = \biguplus_{\usr} exc_{-\usr}$ by definition, and the left premise of ($\linearcontract$-left) is satisfied because $\biguplus_{\usr} \exc'_{\usr} \subseteq \exc$. 
	Note also that, from the previous formula, the number of (\mbox{$\linearcontract$-}split) applications  is bounded by the size of $\exc$. 
	Moreover, the number of (L-Weak) applications is bounded by $O(\Omega_{\pol})$, which in turns is less then $O(\Omega_R)$.
	Hence the total size of the derivation is $O(|\exc| + |\Omega_R|\cdot \CS_{C, R} + |\Omega_C|)$.
\end{proof}

\begin{lem}\label{thm:deltaexc}
	For 
	every
	 $\MuACstate, \MuACstate'$ and $\exc$, if $\MuACstate \xrightarrow{\exc} \MuACstate'$ then $\Delta_{\exc}, \semantics{\MuACstate} \vdash \semantics{\MuACstate'}$ is provable in \MuACL\ proof of size $O(|\semantics{\MuACstate}|)$.
\end{lem}
\begin{proof}
	By induction on the size of $\exc$.
	If $\exc = \emptyset$ then we can build a proof with
	 $O(\semantics{\MuACstate})$ applications of ($\otimes$-left-$\Sigma$) and of ($\otimes$-right) and of ($\Sigma$-Ax).

	Given a proof for $\Delta_{\exc}, \semantics{\MuACstate} \vdash \semantics{\MuACstate'}$, let $\exc$ be $\exc' \cup \{ \tr \}$ with $\tr = \usr \xmapsto{\res} \usr'$.
	By definition, $\st = \{ (\usr, \res) \} \uplus \st''$, $\st' = \{ (\usr', \res) \} \uplus \st'''$ with $\st'' \xrightarrow{\exc'} \st'''$.
	
	By induction hypothesis, a proof $\Pi$ exists of size $O(|\semantics{\st''}|)$ for $\Delta_{\exc'}, \semantics{\st''} \vdash \semantics{\st'''}$.
	Then the following is a proof for $\Delta_{\exc}, \semantics{\st} \vdash \semantics{\st'}$.
	
{\normalsize 	\[
	\prftree[r]
	{($\otimes$-right)}
	{
		\prftree[r]
		{($\Sigma$-Ax)}
		{\res@\usr \vdash \res@\usr}
		{
				\Delta_{\tr},
				\res@\usr 
			\vdash \res@\usr'
		}
	}
	{
		\prftree
		{\Pi}
		{
				\Delta_{\exc'},
				\semantics{\st''}
			\vdash \semantics{\st'''}}
	}
	{\Delta_{\exc}, \semantics{\st} \vdash \semantics{\st'}}\qedhere
	\]
}\end{proof}

\begin{lem}[fairness implies validity]\label{thm:fairnesstovalidity}
	For 
	every
	 $\st, \st', \exc$, if $\st \xrightarrow{\exc} \st'$ is a fair transition, then $\Omega_R, \Omega_C; \semantics{\st} \vdash \semantics{\st'}$ is provable in \MuACL\ with a proof of size $O(|\exc| + |\Omega_R|\cdot \CS_{C, R} + |\Omega_C| + |\semantics{\MuACstate}|)$.
\end{lem}
\begin{proof}
	By composing the derivations of~\autoref{thm:fairderive} and \autoref{thm:deltaexc}.
\end{proof}

We can now prove the compilation from MuAC to \MuACL\ to be correct and complete.
\correction*
\begin{proof}
	By~\autoref{thm:validitytofairness} and~\autoref{thm:fairnesstovalidity}.
\end{proof}

We investigate now computations, and prove that eventual fairness implies validity.
We first need an intermediate result.
\begin{lem}\label{thm:deltaexccmp}
	If $\MuACstate_0 \xrightarrow{\exc_1} \MuACstate_1 \xrightarrow{\exc_2} \dots \xrightarrow{\exc_n} \MuACstate_n$
	then $\Delta_{\uplus_{i = 1}^n \exc_i}, \semantics{\MuACstate_0} \vdash \semantics{\MuACstate_n}$ is a \MuACLs\ proof of size $O(n \cdot |\semantics{\MuACstate_0}|)$.
\end{lem}
\begin{proof}
	By induction on $n$.
	The base case is given by~\autoref{thm:deltaexc}.
	Taken then $\MuACstate_0 \xrightarrow{\exc_1} \MuACstate_1 \xrightarrow{\exc_2} \dots \xrightarrow{\exc_n} \MuACstate_n$.
	Note that $|\semantics{st}_0| = |\semantics{st}_1| = \cdots = |\semantics{st}_n|$, because resources are neither created nor destroyed.
	
	By induction hypothesis, a proof $\Pi$ exists of size $O( (n-1) \cdot |\semantics{\MuACstate_0}|)$ for $\Delta_{\uplus_{i = 1}^{n-1} \exc_i}, \semantics{\MuACstate_0} \vdash \semantics{\MuACstate_{n-1}}$.
	Moreover, by~\autoref{thm:deltaexc}, a proof $\Pi'$ exists of size $O(|\semantics{\st_0}|)$ for $\Delta_{\exc_n}, \semantics{\MuACstate_{n-1}} \vdash \semantics{\MuACstate_{n}}$.
	The result is eventually proved by composing $\Pi$ and $\Pi'$ as follows.	
{\normalsize 	\[
	\prftree[r]
	{($*$-cut)}
	{
		\prftree
		{\Pi}
		{\Delta_{\uplus_{i = 1}^{n-1} \exc_i}, \semantics{\st_0} \vdash \semantics{\st_{n-1}}}
	}
	{
		\prftree
		{\Pi'}
		{\Delta_{\exc_n}, \semantics{\st_{n-1}} \vdash \semantics{\st_n}}
	}
	{\Delta_{\exc}, \semantics{\st_0} \vdash \semantics{\st_n}}\qedhere
	\]
}\end{proof}

\begin{lem}[eventual fairness implies validity for computations]\label{thm:fairnestovaliditycomp}
	If the computation $\MuACstate_0 \xrightarrow{\exc_1} \MuACstate_1 \xrightarrow{\exc_2} \dots \xrightarrow{\exc_n} \MuACstate_n$
 is eventually fair, then $\Omega_R, \Omega_C; \semantics{\MuACstate_0} \vdash \semantics{\MuACstate_n}$ 
	has a \MuACLs\ proof of size $O(\sum_{\st_{i-1} \xrightarrow{\exc_i} \st_{i+1}}  (|\exc_i| + |\Omega_R|\cdot \CS_{C, R} + |\Omega_C| + |\semantics{\MuACstate_0}|))$.
\end{lem}
\begin{proof}
	By composing the derivations of~\autoref{thm:fairderive} and \autoref{thm:deltaexc}.
\end{proof}

The following corollary states that fair computations are exactly the ones encoded by valid initial sequents.
\Logiccorrectnessstar*
\begin{proof}
	By~\autoref{thm:validitytofairnesscomp} and~\autoref{thm:fairnestovaliditycomp}.
\end{proof}

\section{Exploring Reachable States}\label{app:solutions}

We prove some properties useful for exploring the states that are reachable with 
fair transitions or eventually fair computations.

Given a state $\MuACstate$, the problem is to find a state $\MuACstate'$ reachable through a fair transition or an eventually fair computation where $\MuACstate'$ satisfies some desired properties, or to asses that there is none.
\autoref{thm:MuACLdecide} and~\autoref{thm:MuACLsdecide} do not help much because there is an infinite number of possible candidates for $\MuACstate'$.
We solve the problem by showing invariant properties on the reachable states, restricting our candidates to a finite set of possibilities.

Intuitively, the \emph{quantity} of a linear proposition is the number of atomic linear propositions appearing in it which are not bound by logical connectives other than $\otimes$.
\begin{defi}
Let the \emph{quantity} of a linear formula $\sigma$ be the number of occurrences of atomic linear propositions appearing in $\sigma$.

The quantity of a set of linear propositions $\Sigma$ is the sum of the quantity of its elements.
\end{defi}

The following simple property about quantity preservation holds for initial sequents.
\begin{lem}\label{thm:quantity}
	For each $\Omega, \Sigma, \sigma$, if $\Omega; \Sigma \vdash \sigma$ is valid either in \MuACL\ or \MuACLs, then $q(\Sigma) = q(\sigma)$.
\end{lem}
\begin{proof}
	Trivially holds by rule induction.
\end{proof}

\begin{defi}
Let the \emph{atomic linear subformulas} of a formula $\theta, \delta, \sigma$ or $\omega$ be as follows:
\begin{gather*}
asub(r@u)  = \{ r@u \} \qquad\quad asub(G\varphi) = asub(\varphi) \\
	asub(\varphi \star \varphi') = asub(\omega) \cup asub(\omega') \text{\ \ \ \ with } \star \in \{\otimes, \multimap, \linearcontract, \land, \rightarrow\}
\end{gather*}
We homomorphically extend this definition to multisets of linear and non-linear predicates.
\end{defi}

\begin{lem}\label{thm:subformula}
	For each $\Omega, \Sigma, \sigma$, if $\Omega; \Sigma \vdash \sigma$ is valid either in \MuACL\ or \MuACLs, then \\
	$asub(\sigma) \subseteq asub(\Omega) \cup asub(\Sigma)$.
\end{lem}
\begin{proof}
	By rule induction.
\end{proof}

\computation*
\begin{proof}
	The problem to solve is equivalent to find a reachable (in a fair way) $\MuACstate'$ such that $\res_i@\usr \in \semden{\MuACstate'}$ for $i = 1,\dots, n$.
	By Theorems~\ref{thm:correctcompletestar} and~\ref{th:fair-exchange}, the fairness of transitions and computations can be reduced to proving \MuACL\ and \MuACLs\ sequents with $\sigma = \semden{\MuACstate'}$.
	The propositions $\sigma$ to consider are finite by~\autoref{thm:quantity} and~\ref{thm:subformula}, and for each of them we can check validity by~\autoref{thm:MuACLdecide} and \autoref{thm:MuACLsdecide}.
	Finally, note that the encoding of MuAC states into \MuACL\ is clearly a bijection, hence we can recover $\MuACstate'$ from $\sigma$.
\end{proof}

\section{Optimizations for the Blockchain Implementation Schema}\label{app:bcopt}

Till now, we have always assumed the client to send a proof for a \emph{complete} initial sequent of the form $\uplus_{\usr \in \Usr} \semantics{R_{\usr}}, \semantics{C}; \semantics{\st} \vdash \semantics{\st'}$, with $R_{\usr}$ the entire ruleset of $\usr$, $C$ the whole context and $\st$, $\st'$ states of the exchange environment.
An optimisation consists of allowing the user to send proofs for a smaller sequent $\Omega; \Sigma \vdash \sigma$, with $\Omega \subseteq \uplus_{\usr \in \Usr} \semantics{R_{\usr}}, \semantics{C}$, $\Sigma \subseteq \semantics{\st}$ and $\sigma \subseteq \semantics{\st'}$, while maintaining the same guarantees as before (namely that validity of the sequent implies fairness).

The non-linear monotonicity of \MuACL, stated by~\autoref{thm:monotonicity}, allows the smart contract to accept proofs with $\Omega \subseteq \uplus_{\usr \in \Usr} \semantics{R_{\usr}}, \semantics{C}$.
For example, the compilation of the rulesets of the users that are not involved in the exchange can be omitted, as well as the part of the context and the rules of the involved users that are not necessary  for the exchange. 
By verifying the received proof, the smart contract certifies the validity of the complete initial sequent, and thus the fairness of the exchange.
For the linear part, we rely on the following result, allowing us to send proofs with $\Sigma \subseteq \semantics{\st}$ and $\sigma \subseteq \semantics{\st'}$.
\begin{prop}\label{thm:compositional}
	For each $\Omega, \Theta, \Delta, \Sigma, \sigma, \sigma'$, if $\Omega; \Theta, \Delta, \Sigma \vdash \sigma$ is valid in \MuACL\ (or \MuACLs), then 
	$\Omega; \Theta, \Delta, \Sigma, \sigma' \vdash \sigma \otimes \sigma'$ is valid in \MuACL\ (or \MuACLs).
\end{prop}
\begin{proof}
	Let $\Pi$ be a proof for $\Omega; \Theta, \Delta, \Sigma \vdash \sigma$, then the following is a proof for $\Omega; \Theta, \Delta, \Sigma, \sigma' \vdash \sigma \otimes \sigma'$.
	{\normalsize \[
	\prftree[r, doubleline]
	{($\otimes$-right)}
	{
		{
			\begin{matrix}
				\Pi\\
				\Omega; \Theta, \Delta, \Sigma \vdash \sigma				
			\end{matrix}\qquad
		}
		{
			\begin{matrix}
				\Pi'_{\{(\Sigma\text{-Ax}), (\otimes\text{-right}), (\otimes\text{-left})\}}\\
				\sigma' \vdash \sigma'
			\end{matrix}
		}
	}
	{\Omega, \Omega'; \Theta, \Delta, \Sigma, \sigma' \vdash \sigma \otimes \sigma'}
	\]}
	$\!\!\!\!\!$
with $\Pi'$  defined by induction on the size of $\sigma'$ using ($\Sigma$-Ax), ($\otimes$-right) and ($\otimes$-left).
\end{proof}
Note that $\semantics{\st}$ and $\semantics{\st'}$ represent resource associations also for users and resources that are not involved in the exchange.
Instead, due to the previous result, we can just send a proof that only involves the linear atomic propositions of resource associations that change during the transition or computation.
\\
Consider~\autoref{thm:fairnesstovalidity} and assume the context $C$ to be fixed.
The size of the \MuACL\ proof can thus be reduced from $O(|\exc| + |\Omega_R|\cdot \CS_{C, R} + |\Omega_C| + |\semantics{\MuACstate}|)$ to $O(|\exc|\cdot \CS_{C, R})$ since
\begin{itemize}
	\item $\Omega_R$ can be reduced by~\autoref{thm:monotonicity} to the needed formulas, which are of the same size of $\exc$ since each of them results in at least a transfer of resources.
	\item $|\Omega_C|$ is assumed to be constant (note that it can also be reduced by~\autoref{thm:monotonicity});
	\item $|\semantics{\MuACstate}|$ can be reduced by~\autoref{thm:compositional} to the set of linear atomic propositions that are involved in the exchange, 
and this set has the same size of the exchange itself.
\end{itemize}

The same result holds for eventually fair computations $\MuACstate_0 \xrightarrow{\exc_1} \MuACstate_1 \xrightarrow{\exc_2} \dots \xrightarrow{\exc_n} \MuACstate_n$, for which the size of the \MuACLs\ proof can be reduced to $O(\sum_{i = 1}^n |\exc_i|\cdot \CS_{C, R})$.

\end{document}